%% file: Arxiv_submission.tex
\documentclass[11pt,twoside]{article}

\usepackage{fullpage}

\usepackage{epsf}
\usepackage{fancyhdr}
\usepackage{graphics}
\usepackage{graphicx}
\usepackage{psfrag}
\usepackage{microtype}
\usepackage{subfigure}
\usepackage{algorithmic}

\usepackage[linesnumbered,ruled]{algorithm2e}
\DontPrintSemicolon
\usepackage{color}

\usepackage{amsthm}
\usepackage{amsfonts}
\usepackage{amsmath}
\usepackage{amssymb,bbm}
\usepackage{wrapfig}
\usepackage{tikz}
\usetikzlibrary{positioning}

\usepackage{natbib}

\usepackage{url}
\usepackage[colorlinks,linkcolor=magenta,citecolor=blue, pagebackref=true,backref=true]{hyperref}
\renewcommand*{\backrefalt}[4]{%
    \ifcase #1 \footnotesize{(Not cited.)}%
    \or        \footnotesize{(Cited on page~#2.)}%
    \else      \footnotesize{(Cited on pages~#2.)}%
    \fi}

\long\def\comment#1{}
\usepackage{nicefrac}
\usepackage[normalem]{ulem}
\usepackage{chngpage}

 \usepackage{tabularx}%

\usepackage{enumitem}
\usepackage{booktabs}
\usepackage{caption}

\usepackage{mathtools}

\usepackage{fullpage}

\newtheorem{theorem}{Theorem}[section]
\newtheorem{corollary}[theorem]{Corollary}
\newtheorem{lemma}[theorem]{Lemma}
\newtheorem{proposition}[theorem]{Proposition}

\newtheorem{definition}{Definition}

\newtheorem{remark}[theorem]{Remark}

\newcommand{\Br}{\mathbb{R}}

\newcommand{\st}{\textnormal{s.t.}}

\newcommand{\argmin}{\mathop{\rm argmin}}
\newcommand{\argmax}{\mathop{\rm argmax}}
\newcommand{\LCal}{\mathcal{L}}

\newcommand{\br}{\mathbb{R}}

\newcommand{\ba}{\begin{array}}
\newcommand{\ea}{\end{array}}

\newcommand{\Rspace}{\mathbb{R}}
\newcommand{\one}{\textbf{1}}
\newcommand{\zero}{\textbf{0}}
\newcommand{\bigO}{O}
\newcommand{\bigOtil}{\widetilde{O}}
\newcommand{\mydefn}{:=}

\begin{document}


\begin{center}

{\bf{\LARGE{On the Efficiency of Entropic Regularized \\ [.2cm] Algorithms for Optimal Transport}}}

\vspace*{.2in}
{\large{
\begin{tabular}{ccc}
Tianyi Lin$^\diamond$ & Nhat Ho$^\star$ & Michael I. Jordan$^{\diamond, \dagger}$ \\
 \end{tabular}
}}

\vspace*{.2in}

\begin{tabular}{c}
Department of Electrical Engineering and Computer Sciences$^\diamond$ \\
Department of Statistics$^\dagger$ \\ 
University of California, Berkeley \\
Department of Statistics and Data Sciences, University of Texas, Austin$^\star$
\end{tabular}

\vspace*{.2in}

\today

\vspace*{.2in}

\begin{abstract} 
We present several new complexity results for the entropic regularized algorithms that approximately solve the optimal transport (OT) problem between two discrete probability measures with at most $n$ atoms. First, we improve the complexity bound of a greedy variant of Sinkhorn, known as \textit{Greenkhorn}, from $\bigOtil(n^2\varepsilon^{-3})$ to $\bigOtil(n^2\varepsilon^{-2})$. Notably, our result can match the best known complexity bound of Sinkhorn and help clarify why Greenkhorn significantly outperforms Sinkhorn in practice in terms of row/column updates as observed by~\citet{Altschuler-2017-Near}. Second, we propose a new algorithm, which we refer to as \textit{APDAMD} and which generalizes an adaptive primal-dual accelerated gradient descent (APDAGD) algorithm~\citep{Dvurechensky-2018-Computational} with a prespecified mirror mapping $\phi$. We prove that APDAMD achieves the complexity bound of $\bigOtil(n^2\sqrt{\delta}\varepsilon^{-1})$ in which $\delta>0$ stands for the regularity of $\phi$. In addition, we show by a counterexample that the complexity bound of $\bigOtil(\min\{n^{9/4}\varepsilon^{-1}, n^2\varepsilon^{-2}\})$ proved for APDAGD before is invalid and give a refined complexity bound of $\bigOtil(n^{5/2}\varepsilon^{-1})$. Further, we develop a \textit{deterministic} accelerated variant of Sinkhorn via appeal to estimated sequence and prove the complexity bound of $\bigOtil(n^{7/3}\varepsilon^{-4/3})$. As such, we see that accelerated variant of Sinkhorn outperforms Sinkhorn and Greenkhorn in terms of $1/\varepsilon$ and APDAGD and accelerated alternating minimization (AAM)~\citep{Guminov-2021-Combination} in terms of $n$. Finally, we conduct the experiments on synthetic and real data and the numerical results show the efficiency of Greenkhorn, APDAMD and accelerated Sinkhorn in practice.
\end{abstract}
\end{center}

\input{sec/introduction}

\input{sec/background}

\input{sec/greenkhorn}

\input{sec/APDAMD}

\input{sec/acceleration}

\input{sec/experiment}

\input{sec/conclusion}

\section{Acknowledgments}
This work was supported in part by the Mathematical Data Science program of the Office of Naval Research under grant number N00014-18-1-2764 to MJ, and by the NSF IFML 2019844 award and research gifts by UT Austin ML grant to NH.

\bibliographystyle{plainnat}
\bibliography{ref}
\end{document}

%% file: sec/introduction.tex
\section{Introduction}
From its origins in the seminal works by~\citet{Monge-1781-Memoire} and~\citet{Kantorovich-1942-Translocation} respectively in the eighteenth and twentieth centuries, and through to present day, the optimal transport (OT) problem has played a \textit{determinative} role in the theory of mathematics~\citep{Villani-2009-Optimal}. It also has found a wide range of applications in problem domains beyond the original setting in logistics. In the current era, the strong and increasing linkage between optimization and machine learning has brought new applications of OT to the fore;~\citep[see, e.g.,][]{Nguyen-2013-Convergence, Cuturi-2014-Fast, Srivastava-2015-WASP, Rolet-2016-Fast, Peyre-2016-Averaging, Nguyen-2016-Borrowing, Carriere-2017-Sliced, Arjovsky-2017-Wasserstein, Gulrajani-2017-Improved, Courty-2017-Optimal, Srivastava-2018-Scalable, Dvurechenskii-2018-Decentralize, Tolstikhin-2018-Wasserstein, Sommerfeld-2019-Optimal, Lin-2019-Sparsemax, Ho-2019-Probabilistic}. In these data-driven applications, the focus is on the probability distributions underlying the OT formulation; indeed, these distributions are either empirical distributions which are obtained by placing unit masses at data points, or are probability models of a putative underlying data-generating process. The OT problem accordingly often has a direct inferential meaning --- as the definition of an estimator~\citep{Dudley-1969-Speed, Fournier-2015-Rate, Weed-2019-Sharp, Lei-2020-Convergence}, the definition of a likelihood~\citep{Sommerfeld-2018-Inference, Bernton-2019-Parameter, Blanchet-2019-Quantifying}, or as the robust variant of an estimator~\citep{Blanchet-2019-Robust, Paty-2019-Subspace, Balaji-2020-Robust}. The key challenge is computational~\citep{Peyre-2019-Computational}. Indeed, the underlying distributions generally involve high-dimensional data and complex models in machine learning (ML) applications.

We study the OT problem in a discrete setting where we assume that the target and source probability distributions each have at most $n$ atoms. In this setting, the OT problem can be solved exactly using linear programming (LP) solver based on specialized interior-point methods~\citep{Pele-2009-Fast, Lee-2014-Path, Van-2021-Minimum}, reflecting the LP formulation of the OT problem. In this context,~\citet{Van-2021-Minimum} have provided a bunch of randomized interior-point algorithms with improved runtimes for solving linear programs with two-sided constraints, leading to a new OT algorithm based on the Laplacian system solvers that achieved the best known complexity bounds of $\tilde{O}(n^2)$. However, it does not provide an effective solution to large-scale machine learning problems in practice since efficient implementations of Laplacian approach are yet unknown.  Furthermore, many combinatorial techniques give exact algorithms for the OT problem. Indeed, the Hungarian algorithm~\citep{Kuhn-1955-Hungarian, Kuhn-1956-Variants, Munkres-1957-Algorithms} solves the assignment problem in $O(n^3)$ time while there are several combinatorial algorithms that can solve the OT problem exactly in $\tilde{O}(n^{2.5})$ time~\citep{Gabow-1991-Faster, Orlin-1992-New}. Combined with the scaling technique, the network simplex algorithms~\citep{Orlin-1993-Polynomial, Orlin-1997-Polynomial} can be used to solve the OT problem exactly in $\tilde{O}(n^3)$ time and~\citet{Lahn-2019-Graph} have recently developed a faster approximation algorithm for the OT problem via appeal to the modification of the algorithm developed in~\citet{Gabow-1991-Faster}. However, computing the OT problem exactly results in an output that is \textit{not} differentiable with respect to measures' locations or weights~\citep{Bertsimas-1997-Introduction}. Moreover, OT suffers from the curse of dimensionality~\citep{Dudley-1969-Speed, Fournier-2015-Rate} and is thus likely to be meaningless when used on samples from high-dimensional densities. 

An alternative to solve the OT problem is a class of approximation algorithms based on the entropy regularization which has been investigated in optimization and transportation science long before~\citep{Sinkhorn-1974-Diagonal,Schneider-1990-Comparative, Kalantari-1996-Complexity, Knight-2008-Sinkhorn, Kalantari-2008-Complexity, Chakrabarty-2018-Better}. It was~\citet{Cuturi-2013-Sinkhorn} that popularized the use of entropy regularization for OT in the machine learning community and then initiated a productive line of research where an entropic regularization was imposed to approximate the non-negative constraints in the original OT problem. The resulting problem is referred to as \textit{entropic regularized OT} and the corresponding class of approximation algorithms are called \textit{entropic regularized algorithms}. It is worth mentioning that the entropic regularized OT has many favorable properties that the OT does not enjoy, motivating us to study the computational efficiency of entropic regularized algorithms in this paper. More specifically, from a statistical point of view, the entropic regularized OT enjoys significantly better sample complexity that is polynomial in the dimension~\citep{Genevay-2019-Sample, Mena-2019-Statistical, Chizat-2020-Faster}, demonstrating that adding an entropy regularization will reduce the curse of dimensionality. Even from a computational point of view, such regularization in OT leads to \textit{Sinkhorn} which attains a first near-linear time guarantee for the OT problem~\citep{Cuturi-2013-Sinkhorn, Altschuler-2017-Near, Dvurechensky-2018-Computational}, and also makes the problem differentiable with regards to distributions~\citep{Feydy-2019-Interpolating}; hence, the entropic regularized algorithms are more easily applicable to deep learning applications~\citep{Courty-2017-Optimal, Cuturi-2019-Differentiable, Balaji-2020-Robust} as opposed to combinatorial algorithms. This point was highlighted in~\citet{Dong-2020-study} and further necessitated the development of faster entropic regularized algorithms. In this regard, the greedy variant of Sinkhorn -- Greenkhorn -- was proposed and shown to outperform Sinkhorn empirically~\citep{Altschuler-2017-Near}. However, a sizable gap exists here since the best known complexity bound of $\bigOtil(n^2\varepsilon^{-3})$ for Greenkhorn~\citep{Altschuler-2017-Near} is worse than that of $\bigOtil(n^2\varepsilon^{-2})$ for Sinkhorn~\citep{Dvurechensky-2018-Computational}. 

Further progress has been made by adapting first-order optimization algorithms for the OT problem~\citep{Cuturi-2016-Smoothed, Genevay-2016-Stochastic, Blondel-2018-Smooth, Dvurechensky-2018-Computational, Altschuler-2019-Massively, Guo-2020-Fast, Guminov-2021-Combination}. Among these approaches, two of representatives are an adaptive primal-dual accelerated gradient descent (APDAGD) algorithm~\citep{Dvurechensky-2018-Computational} with the claimed complexity bound of $\bigOtil(\min\{n^{9/4}\varepsilon^{-1}, n^2\varepsilon^{-2}\})$ and an accelerated alternating minimization (AAM) algorithm~\citep{Guminov-2021-Combination} with the complexity bound of $\bigOtil(n^{5/2}\varepsilon^{-1})$. Moreover, there are several second-order optimization algorithms~\citep{Allen-2017-Much, Cohen-2017-Matrix} which are adapted for the OT problem~\citep{Blanchet-2018-Towards, Quanrud-2019-Approximating} and guaranteed to achieve the improved complexity bound of $\bigOtil(n^2\varepsilon^{-1})$. However, the aforementioned second-order algorithms do not provide effective solutions to large-scale machine learning problems due to the lack of efficient implementations in practice. 

\paragraph{Contributions.} Given the advantages of entropic regularization in OT, we focus in his paper the computational efficiency of a class of entropic regularized algorithms for the OT problem and our theoretical analysis lead to several improvements over the state-of-the-art algorithms in the literature. We summarize the contributions as follows:
\begin{enumerate}
\item We improve the complexity bound of Greenkhorn from $\bigOtil(n^2\varepsilon^{-3})$ to $\bigOtil(n^2\varepsilon^{-2})$, which matches the best existing bound of Sinkhorn. The proof techniques are new and different from that used in~\citet{Dvurechensky-2018-Computational} for analyzing Sinkhorn. In particular, Greenkhorn only updates a single row or column at each iteration and quantifying the per-iteration progress is more difficult than the measurement in Sinkhorn. 

\item We propose an adaptive primal-dual accelerated mirror descent (APDAMD) algorithm which generalizes APDAGD with a prespecified mirror mapping $\phi$ and prove that APDAMD achieves the complexity bound of $\bigOtil(n^2\sqrt{\delta}\varepsilon^{-1})$ where $\delta > 0$ refers to the regularity of $\phi$ w.r.t. $\ell_\infty$ norm. We show by a counterexample that the complexity bound of $\bigOtil(\min\{n^{9/4}\varepsilon^{-1}, n^2\varepsilon^{-2}\})$ proved for APDAGD~\citep{Dvurechensky-2018-Computational} is invalid and give a refined complexity bound of $\bigOtil(n^{5/2}\varepsilon^{-1})$ which is worse than the claimed bound in terms of $n$. 

\item We propose a deterministic accelerated variant of Sinkhorn via appeal to an estimated sequence and prove the complexity bound of $\bigOtil(n^{7/3}\varepsilon^{-4/3})$. In particular, accelerated Sinkhorn consists in an exact minimization for main iterates accompanied by another sequence of iterates based on coordinate gradient updates and monotone search. Our results show that accelerated Sinkhorn outperforms Sinkhorn and Greenkhorn in terms of $1/\varepsilon$ and APDAGD and AAM in terms of $n$. 
\end{enumerate}
We note that a preliminary version with only the analysis for Greenkhorn and APDAMD has been accepted by ICML~\citep{Lin-2019-Efficient}. After our conference paper was published, some new algorithms were developed for solving the OT problem~\citep{Jambulapati-2019-Direct, Lahn-2019-Graph}. In particular, ~\citet{Jambulapati-2019-Direct} developed a dual extrapolation algorithm with the complexity bound $\bigOtil \left(n^2\varepsilon^{-1}\right)$ using an area-convex mapping~\citep{Sherman-2017-Area}. Despite the theoretically sound complexity bound, the lack of simplicity and ease-of-implementation make this algorithm less competitive with Sinkhorn and Greenkhorn which remain the baseline solution methods in practice~\citep{Flamary-2017-Pot}. 

Different from the algorithm in~\citet{Jambulapati-2019-Direct}, the combinatorial algorithm in~\citet{Lahn-2019-Graph} is a practical solution method for the OT problem. It is worth mentioning that the algorithm in~\citet{Lahn-2019-Graph} and other combinatorial algorithms, e.g., the Hungarian algorithm, outperform our algorithms in practice. This is in consistence with the observation in~\citet{Dong-2020-study} who pointed out that combinatorial algorithms can outperform entropic regularized algorithms in speed even the latter ones are asymptotically faster for OT (i.e., the case of large $n$). However, we believe our results are still valuable due to the importance of entropic regularized algorithms as mentioned before.

\paragraph{Organization.} The remainder is organized as follows. In Section~\ref{sec:setup}, we present the basic setup for the primal and dual form of the entropic regularized OT problem. In Section~\ref{sec:greenkhorn}, we provide the complexity analysis for Greenkhorn. In Section~\ref{sec:apdamd}, we propose APDAMD for solving entropic regularized OT and provide several results on the complexity bound of APDAGD and APDAMD. In Section~\ref{sec:acceleration}, we propose and analyze an accelerated variant of Sinkhorn. In Section~\ref{sec:experiments}, we conduct the experiments on synthetic and real data and the numerical results show the efficiency of our algorithms. We conclude this paper in Section~\ref{sec:conclusion}.

\paragraph{Notation.} For $n \geq 2$, we let $[n]$ be the set $\{1, 2, \ldots, n\}$ and $\Rspace^n_+$ be the set of all vectors in $\Rspace^n$ with non-negative coordinates. The notation $\Delta^n = \{v \in \Rspace_+^n: \sum_{i = 1}^{n} v_{i} = 1\}$ stands for a probability simplex in $n-1$ dimensions. For a vector $x \in \Rspace^n$ and let $1 \leq p < +\infty$, the notation $\|x\|_p$ stands for the $\ell_p$-norm and $\|x\|$ indicates an $\ell_2$-norm. $\text{diag}(x)$ is a diagonal matrix which has the vector $x$ on its diagonal. $\one_n$ and $\zero_n$ are $n$-dimensional vector with all components being $1$ and 0. For a matrix $A \in \Rspace^{n \times n}$, we denote $\text{vec}(A)$ as the vector in $\Rspace^{n^2}$ obtained from concatenating the rows and columns of $A$. The notation $\|A\|_{1 \rightarrow 1}$ stands for $\sup_{\|x\|_1=1} \|Ax\|_1$ and the notations $r(A) = A\one_n$ and $c(A) = A^\top\one_n$ stand for the row and column sums respectively. For a function $f$, the notation $\nabla_x f$ denotes a partial derivative with respect to $x$. For the dimension $n$ and tolerance $\varepsilon > 0$, the notations $a = \bigO(b(n, \varepsilon))$ and $a = \Omega(b(n,\varepsilon))$ indicate that $a \leq C_1 \cdot b(n, \varepsilon)$ and $a \geq C_2 \cdot b(n, \varepsilon)$ respectively where $C_1$ and $C_2$ are independent of $n$ and $\varepsilon$. We also denote $a = \Theta(b(n,\varepsilon))$ iff $a = \bigO(b(n,\varepsilon)) = \Omega(b(n,\varepsilon))$. Similarly, we denote $a = \bigOtil(b(n, \varepsilon))$ to indicate the previous inequality where $C_1$ depends on some logarithmic function of $n$ and $\varepsilon$. 

%% file: sec/background.tex
\section{Problem Setup}\label{sec:setup}
In this section, we first present the linear programming (LP) representation of the optimal transport (OT) problem as well as a specification of an approximate transportation plan. We also present an entropic regularized variant of the OT problem and derive the dual form where the objective function is in the form of the logarithm of sum of exponents. Finally, we establish several properties of that dual form which are useful for the subsequent analysis.

\subsection{Linear programming representation}
According to~\citet{Kantorovich-1942-Translocation}, the problem of approximating the OT distance is equivalent to solving the following linear programming (LP) problem:
\begin{equation}\label{prob:OT}
\min\limits_{X \in \br^{n \times n}} \langle C, X\rangle \quad \st \ X\one_n = r, X^\top\one_n = c, X \geq 0. 
\end{equation}
In the above formulation, $X$ refers to the \textit{transportation plan}, $C = (C_{ij}) \in \br_+^{n \times n}$ stands for a cost matrix with non-negative components, and $r \in \br^n$ and $c \in \br^n$ are two probability distributions in the simplex $\Delta^n$. 

We see from Eq.~\eqref{prob:OT}, that the OT problem is a LP with $2n$ equality constraints and $n^2$ variables and can be solved by the interior-point method; however, this method performs poorly on large-scale problems due to its high per-iteration computational cost. In general, the solution that the algorithms aim at achieving is an $\varepsilon$-approximate transportation plan $\widehat{X} \in \br_+^{n \times n}$ satisfying the marginal distribution constraints $\widehat{X}\one_n = r$ and $\widehat{X}^\top\one_n = c$ and the inequality given by 
\begin{equation*}
\langle C, \widehat{X}\rangle \leq \langle C, X^\star\rangle + \varepsilon.
\end{equation*}
Here $X^\star$ is defined as an optimal transportation plan for the OT problem. For simplicity, we respectively denote $\langle C, \hat{X}\rangle$ an \emph{$\varepsilon$-approximate transportation cost} and $\hat{X}$ an \emph{$\varepsilon$-approximate transportation plan} for the original problem. Formally, we have the following definition of $\varepsilon$-approximate transportation plan.
\begin{definition}
\label{def:eps-approximation}
The matrix $\hat{X} \in \br_+^{n \times n}$ is called an $\varepsilon$-approximate transportation plan if $\hat{X}\one_n = r$ and $\hat{X}^\top\one_n = c$ and the following inequality holds true,
\begin{equation*}
\langle C, \widehat{X}\rangle \leq \langle C, X^\star\rangle + \varepsilon.
\end{equation*}
where $X^\star$ is defined as an optimal transportation plan for the OT problem. 
\end{definition}
With this definition in mind, the goal of this paper is to study the OT problem from a computational point of view. Indeed, we hope to derive an improved complexity bound of the current state-of-the-art algorithms and seek new practical algorithms whose running time required to obtain an $\varepsilon$-approximate transportation plan has better dependence on $1/\varepsilon$ than the benchmark algorithms in the literature. The aforementioned new algorithms are favorable in the machine learning applications where high precision ($\varepsilon$ is small) is necessary. 

\subsection{Entropic regularized OT and its dual form}
Seeking another formulation for OT distance that is more amenable to computationally efficient algorithms,~\citet{Cuturi-2013-Sinkhorn} proposed to solve an entropic regularized version of the OT problem in Eq.~\eqref{prob:OT}, which is given by
\begin{equation}\label{prob:OT_regularized}
\min\limits_{X \in \br^{n \times n}} \langle C, X\rangle - \eta H(X), \quad \st \ X\one_n = r, X^\top\one_n = c, 
\end{equation}
where $\eta > 0$ denotes the regularization parameter and $H(X)$ denotes the entropic regularization term, which is given by:
\begin{equation*}
H(X) \mydefn - \langle X, \log(X)-\mathbf{1}_{n \times n}\rangle.
\end{equation*}
Note that, the optimal solution of the entropic regularized OT problem exists since the objective function $\langle C, X\rangle -\eta H(X)$ is continuous and the feasible region $\{X \in \Br^{n \times n}: X \geq 0, X\textbf{1}_n = r, X^\top\textbf{1}_n = c\}$ is compact. Furthermore, that optimal solution is also unique since the objective function $\langle C, X\rangle -\eta H(X)$ is strongly convex over the feasible region with respect to $\ell_1$-norm. However, the optimal value of the entropic regularized OT problem (cf. Eq~\eqref{prob:OT_regularized}) yields a poor approximation to the unregularized OT problem if $\eta$ is large. An additional issue of entropic regularization is that the sparsity of the solution is lost. Even though an $\varepsilon$-approximate transportation plan can be found efficiently, it is not clear how different the sparsity pattern of this solution is with respect to the solution of the actual OT problem. In contrast, the actual OT distance suffers from the curse of dimensionality~\citep{Dudley-1969-Speed, Fournier-2015-Rate, Weed-2019-Sharp} and is significantly worse than its entropic regularized version in terms of the sample complexity~\citep{Genevay-2019-Sample, Mena-2019-Statistical, Chizat-2020-Faster}.

While there is an ongoing debate in the literature on the merits of solving the OT problem~\textit{v.s.} its entropic regularized version, we adopt here the viewpoint that reaching an additive approximation of the actual OT cost matters and therefore propose to scale $\eta$ as a function of the desired accuracy of the approximation. Then, we proceed to derive the dual form of the entropic regularized OT problem in Eq.~\eqref{prob:OT_regularized} and show that it remains an unconstrained smooth optimization problem. By introducing the dual variables $\alpha, \beta \in \br^n$, we define the Lagrangian function of the entropic regularized OT problem as follows:
\begin{equation}\label{opt:Lagrangian}
\LCal(X, \lambda_1, \ldots, \lambda_m) = \langle C, X\rangle - \eta H(X) - \alpha^\top(X\one_n - r) - \beta^\top(X^\top\one_n - c). 
\end{equation}
In order to derive the smooth dual objective function, we consider the following minimization problem:
\begin{equation*}
\min_{X: \|X\|_1=1} \langle C, X\rangle - \eta H(X) - \alpha^\top(X\one_n - r) - \beta^\top(X^\top\one_n - c). 
\end{equation*}
The above objective function is strongly convex over the domain $\{X \in \br_+^{n \times n} \mid \|X\|_1=1\}$. Thus, the optimal solution is unique. After the simple calculations, the optimal solution $\bar{X} = X(\alpha, \beta)$ has the following form:
\begin{equation}\label{opt:OT_plan}
\bar{X}_{ij} = \frac{e^{\eta^{-1}(\alpha_i + \beta_j - C_{ij})}}{\sum_{1 \leq i, j \leq n} e^{\eta^{-1}(\alpha_i + \beta_j - C_{ij})}}. 
\end{equation}
Plugging Eq.~\eqref{opt:OT_plan} into Eq.~\eqref{opt:Lagrangian} yields that the dual form is: 
\begin{equation*}
\max_{\alpha, \beta} \ \left\{- \eta\log\left(\sum_{1 \leq i, j \leq n} e^{\eta^{-1}(\alpha_i + \beta_j - C_{ij})}\right) + \alpha^\top r + \beta^\top c\right\}.
\end{equation*}
In order to streamline our presentation, we perform a change of variables, $u = \eta^{-1}\alpha$ and $v = \eta^{-1}\beta$, and reformulate the above problem as
\begin{equation*}
\min_{\alpha, \beta} \ \varphi(\alpha, \beta) \mydefn \log\left(\sum_{1 \leq i, j \leq n} e^{u_i + v_j - \frac{C_{ij}}{\eta}}\right) - u^\top r - v^\top c.  
\end{equation*}
To further simplify the notation, we define $B(u, v) \mydefn (B_{ij})_{1 \leq i, j \leq n} \in \br^{n \times n}$ by
\begin{equation*}
B_{ij} = e^{u_i + v_j - \frac{C_{ij}}{\eta}}. 
\end{equation*}
To this end, we obtain the \emph{dual entropic regularized OT problem} defined by
\begin{equation}\label{prob:OT_regularized_dual}
\min \limits_{u, v} \ \varphi(u, v) \mydefn \log(\|B(u, v)\|_1) - u^\top r - v^\top c.  
\end{equation}
\begin{remark}\label{remark:OT_regularized_dual}
The first part of the objective function $\varphi$ is in the form of the logarithm of sum of exponents while the second part is a linear function. This is different from the objective function used in previous dual entropic regularized OT problem~\citep{Cuturi-2013-Sinkhorn, Altschuler-2017-Near, Dvurechensky-2018-Computational, Lin-2019-Efficient}. Notably, Eq.~\eqref{prob:OT_regularized_dual} is a special instance of a softmax minimization problem, and the objective function $\varphi$ is known to be smooth~\citep{Nesterov-2005-Smooth}. Finally, we point out that the same formulation has been derived in~\citet{Guminov-2021-Combination} for analyzing AAM. 
\end{remark}
In the remainder of the paper, we also denote $(u^\star, v^\star) \in \Rspace^{2n}$ as an optimal solution of the dual entropic regularized OT problem in Eq.~\eqref{prob:OT_regularized_dual}.

\subsection{Properties of dual entropic regularized OT}
We present several useful properties of the dual entropic regularized OT in Eq.~\eqref{prob:OT_regularized_dual}. In particular, we show that there exists an optimal solution $(u^\star, v^\star) \in \Rspace^{2n}$ such that it has an upper bound in terms of the $\ell_\infty$-norm. 
\begin{lemma}\label{lemma:dual-bound-infinity}
For the dual entropic regularized OT problem in Eq.~\eqref{prob:OT_regularized_dual}, there exists an optimal solution $(u^\star, v^\star)$ such that 
\begin{equation*}
\|u^\star\|_\infty \leq R, \qquad \|v^\star\|_\infty \leq R, 
\end{equation*}
where $R \mydefn \eta^{-1}\|C\|_\infty + \log(n) - \log(\min_{1 \leq i, j \leq n} \{r_i, c_j\})$ depends on $C$, $r$ and $c$. 
\end{lemma}
\begin{proof}
First, we claim that there exists an optimal solution $(u^\star, v^\star)$ such that
\begin{equation}\label{claim-dual-bound-first}
\max_{1 \leq i \leq n} u_i^\star \geq 0 \geq \min_{1 \leq i \leq n} u_i^\star, \qquad \max_{1 \leq i \leq n} v_i^\star \geq 0 \geq \min_{1 \leq i \leq n} v_i^\star. 
\end{equation}
Indeed, letting $(\widehat{u}^\star, \widehat{v}^\star)$ be an optimal solution to Eq.~\eqref{prob:OT_regularized_dual}, the claim holds true if $(\widehat{u}^\star, \widehat{v}^\star)$ satisfies Eq.~\eqref{claim-dual-bound-first}. Otherwise, we define the shift term given by 
\begin{eqnarray*}
\widehat{\Delta}_u & = & \frac{\max_{1 \leq i \leq n} \widehat{u}_i^\star + \min_{1 \leq i \leq n} \widehat{u}_i^\star}{2}, \\
\widehat{\Delta}_v & = & \frac{\max_{1 \leq i \leq n} \widehat{v}_i^\star + \min_{1 \leq i \leq n} \widehat{v}_i^\star}{2},  
\end{eqnarray*}
and define $(u^\star, v^\star)$ by  
\begin{equation*}
u^\star = \widehat{u}^\star - \widehat{\Delta}_u \one_n, \qquad v^\star = \widehat{v}^\star - \widehat{\Delta}_v \one_n. 
\end{equation*}
By definition, we have $(u^\star, v^\star)$ satisfies Eq.~\eqref{claim-dual-bound-first}. Since $\one_n^\top r = \one^\top_n c = 1$, we have $(u^\star)^\top r = (\widehat{u}^\star)^\top r - \widehat{\Delta}_u$ and $(v^\star)^\top c = (\widehat{v}^\star)^\top c - \widehat{\Delta}_v$. In addition, $\log(\|B(u^\star, v^\star)\|_1) = \log(\|B(\widehat{u}^\star, \widehat{v}^\star)\|_1) + \widehat{\Delta}_u + \widehat{\Delta}_v$. Putting these pieces together yields $\varphi(u^\star, v^\star) = \varphi(\widehat{u}^\star, \widehat{v}^\star)$. Therefore, $(u^\star, v^\star)$ is an optimal solution of the dual entropic regularized OT that satisfies Eq.~\eqref{claim-dual-bound-first}.

Then, we show that  
\begin{align}
\max_{1 \leq i \leq n} u_i^\star - \min_{1 \leq i \leq n} u_i^\star & \leq \frac{\|C\|_\infty}{\eta} - \log\left(\min_{1 \leq i, j \leq n} \{r_i, c_j\}\right), \label{claim-dual-bound-u} \\
\max_{1 \leq i \leq n} v_i^\star - \min_{1 \leq i \leq n} v_i^\star & \leq \frac{\|C\|_\infty}{\eta} - \log\left(\min_{1 \leq i, j \leq n} \{r_i, c_j\}\right). \label{claim-dual-bound-v}
\end{align}
Indeed, for any $1 \leq i \leq n$, we derive from the optimality condition of $(u^\star, v^\star)$ that
\begin{equation*}
\frac{e^{u_i^\star}(\sum_{j=1}^n e^{v_j^\star - \eta^{-1}C_{ij}})}{\|B(u^\star, v^\star)\|_1} = r_i, \quad \textnormal{for all } i \in [n].  
\end{equation*}
Since $C_{ij} \geq 0$ for all $1 \leq i, j \leq n$ and $r_i \geq \min_{1 \leq i, j \leq n} \{r_i, c_j\}$ for all $1 \leq i \leq n$, we have
\begin{equation*}
u_i^\star \geq \log\left(\min_{1 \leq i, j \leq n} \{r_i, c_j\}\right) - \log \left(\sum_{j=1}^n e^{v_j^\star} \right) + \log(\|B(u^\star, v^\star)\|_1), \quad \textnormal{for all } i \in [n].  
\end{equation*}
Since $0 < r_i \leq 1$ and $C_{ij} \leq \|C\|_\infty$, we have
\begin{equation*}
u_i^\star \leq \frac{\|C\|_\infty}{\eta} -  \log \left(\sum_{j=1}^n e^{v_j^\star} \right) + \log(\|B(u^\star, v^\star)\|_1), \quad \textnormal{for all } i \in [n].  
\end{equation*}
Putting these pieces together yields Eq.~\eqref{claim-dual-bound-u}. By the similar argument, we can prove Eq.~\eqref{claim-dual-bound-v}. 

Finally, we prove our main results. Indeed, Eq.~\eqref{claim-dual-bound-first} and Eq.~\eqref{claim-dual-bound-u} imply that  
\begin{equation*}
-\frac{\|C\|_{\infty}}{\eta} + \log\left(\min_{1 \leq i, j \leq n} \{r_i, c_j\}\right) \leq \min\limits_{1 \leq i \leq n} u_i^\star \leq 0,  
\end{equation*}  
and 
\begin{equation*}
0 \leq \max\limits_{1 \leq i \leq n} u_i^\star \leq \frac{\|C\|_\infty}{\eta} - \log\left(\min_{1 \leq i, j \leq n} \{r_i, c_j\}\right). 
\end{equation*}
Combining the above two inequalities with the definition of $R$ implies that $\|u^\star\|_\infty \leq R$. By the similar argument, we can prove that $\|v^\star\|_\infty \leq R$. As a consequence, we obtain the conclusion of the lemma.
\end{proof}
The upper bound for the $\ell_{\infty}$-norm of an optimal solution of dual entropic regularized OT in Lemma~\ref{lemma:dual-bound-infinity} directly leads to the following direct bound for the $\ell_2$-norm. 
\begin{corollary}\label{corollary:dual-bound-l2}
For the dual entropic regularized OT problem in Eq.~\eqref{prob:OT_regularized_dual}, there exists an optimal solution $(u^\star, v^\star)$ such that 
\begin{equation*}
\|u^\star\| \leq \sqrt{n}R, \qquad \|v^\star\| \leq \sqrt{n}R, 
\end{equation*}
where $R > 0$ is defined in Lemma~\ref{lemma:dual-bound-infinity}. 
\end{corollary}
Since the function $-H(X)$ is strongly convex with respect to the $\ell_1$-norm on the probability simplex $Q \subseteq \Rspace^{n \times n}$, the entropic regularized OT problem in Eq.~\eqref{prob:OT_regularized} is a special case of the following linearly constrained convex optimization problem: 
\begin{equation*}
\min_{x \in Q} \ f(x), \quad \st \ Ax = b, 
\end{equation*}
where $f$ is strongly convex with respect to the $\ell_1$-norm on the set $Q$: 
\begin{equation*}
f(x') - f(x) - (x' - x)^\top\nabla f(x) \geq \frac{\eta}{2}\|x' - x\|_1^2 \textnormal{ for any } x', x \in Q.  
\end{equation*}
By~\citet[Theorem~1]{Nesterov-2005-Smooth} with the $\ell_2$-norm for the dual space of the Lagrange multipliers, the dual objective function $\tilde{\varphi}$ satisfies the following inequality:  
\begin{equation*}
\widetilde{\varphi}(\alpha', \beta') - \widetilde{\varphi}(\alpha, \beta) - \begin{pmatrix} \alpha' - \alpha \\ \beta' - \beta\end{pmatrix}^\top \nabla\widetilde{\varphi}(\alpha, \beta) \leq \frac{\|A\|_{1 \rightarrow 2}^2}{2\eta}\left\|\begin{pmatrix} \alpha' - \alpha \\ \beta' - \beta\end{pmatrix}\right\|^2 \textnormal{ for any } (\alpha', \beta'), (\alpha, \beta) \in \Rspace^{2n}.   
\end{equation*}
Recall that the function $\tilde{\varphi}$ is given by 
\begin{equation}\label{def:dual-regOT-objective}
\widetilde{\varphi}(\alpha, \beta) = - \eta\log\left(\sum_{1 \leq i, j \leq n} e^{\eta^{-1}(\alpha_i + \beta_j - C_{ij})}\right) + \alpha^\top r + \beta^\top c. 
\end{equation}
We notice that the function $\varphi$ in Eq.~\eqref{prob:OT_regularized_dual} satisfies that $\varphi(u, v) = - \eta^{-1}\widetilde{\varphi}(\eta u, \eta v)$. After some simple calculations, we have
\begin{equation}\label{inequality-gradient-objective}
\varphi(u', v') - \varphi(u, v) -  \begin{pmatrix} u' - u \\ v' - v\end{pmatrix}^\top\nabla\varphi(u, v) \leq \left(\frac{\|A\|_{1 \rightarrow 2}^2}{2}\right)\left\|\begin{pmatrix} u' - u \\ v' - v\end{pmatrix}\right\|^2. 
\end{equation}
In the entropic regularized OT problem, each column of the matrix $A$ contains no more than two nonzero elements which are equal to one. Since $\|A\|_{1 \rightarrow 2}$ is equal to maximum $\ell_2$-norm of the column of this matrix, we have $\|A\|_{1 \rightarrow 2} = \sqrt{2}$. Thus, the dual objective function $\varphi$ is 2-gradient Lipschitz with respect to the $\ell_2$-norm.

%% file: sec/greenkhorn.tex
\begin{algorithm}[!t]\small
\caption{\textsc{Greenkhorn}$(C, \eta, r, c, \varepsilon')$} \label{Algorithm:GK}
\begin{algorithmic}
\STATE \textbf{Input:} $t = 0$ and $u^0 = v^0 = \zero_n$. 
\WHILE{$E_t > \varepsilon'$}
\STATE Compute $I = \argmax_{1 \leq i \leq n} \rho(r_i, r_i(B(u^t, v^t)))$ where $\rho(a, b) = b - a + a\log(a/b)$ and $(B(u^{t}, v^{t}))_{i'j'} = e^{u^{t}_{i'} + v^{t}_{j'} - \frac{C_{i'j'}}{\eta}}$ for all $(i', j')$. 
\STATE Compute $J = \argmax_{1 \leq j \leq n} \rho(c_j, c_j(B(u^t, v^t)))$.
\IF{$\rho(r_i, r_i(B(u^t, v^t))) > \rho(c_j, c_j(B(u^t, v^t)))$}
\STATE $u_I^{t+1} = u_I^t + \log(r_I) - \log(r_I(B(u^t, v^t)))$. 
\ELSE
\STATE $v_J^{t+1} = v_J^t + \log(c_J) - \log(c_J(B(u^t, v^t)))$.
\ENDIF
\STATE Increment by $t = t + 1$. 
\ENDWHILE
\STATE \textbf{Output:} $B(u^t, v^t)$.  
\end{algorithmic}
\end{algorithm} 
\section{Greenkhorn}\label{sec:greenkhorn}
In this section, we present a complexity analysis for Greenkhorn. In particular, we improve the existing best known complexity bound $\bigO(n^2\|C\|_\infty^3\log(n)\varepsilon^{-3})$~\citep{Altschuler-2017-Near} to $\bigO(n^2\|C\|_\infty^2 \log(n)\varepsilon^{-2})$, which matches the current state-of-the-art complexity bound for Sinkhorn~\citep{Dvurechensky-2018-Computational}. 

To facilitate the subsequent discussion, we present the pseudocode of Greenkhorn in Algorithm~\ref{Algorithm:GK} and its application to regularized OT in Algorithm~\ref{Algorithm:ApproxOT_GK}. The function for quantifying the progress in the dual objective value between two consecutive iterates is given by $\rho(a, b) = b - a + a\log(a/b)$ and we recall that $(u, v)$ is an optimal solution of the dual entropic regularized OT problem in Eq.~\eqref{prob:OT_regularized_dual} if $r(B(u, v)) - r = \zero_n$ and $c(B(u, v)) - c = \zero_n$. This leads to the quantity which measures the error of the $t$-th iterate in Algorithm~\ref{Algorithm:GK}:
\begin{equation*}
E_t \mydefn \|r(B(u^t, v^t)) - r\|_1 + \|c(B(u^t, v^t)) - c\|_1.
\end{equation*}
Both Sinkhorn and Greenkhorn can be interpreted as coordinate descent for minimizing the following convex function~\citep{Cuturi-2013-Sinkhorn, Altschuler-2017-Near, Dvurechensky-2018-Computational, Lin-2019-Efficient}: 
\begin{equation}\label{def:Greenkhorn-obj}
f(u, v) \mydefn \|B(u, v)\|_1 - u^\top r - v^\top c.  
\end{equation}
Comparing to the scheme of Sinkhorn that consists in the updates of \textit{all} rows and columns, Algorithm~\ref{Algorithm:GK} updates only \textit{one} row or column at each step. As such, Algorithm~\ref{Algorithm:GK} updates only $\bigO(n)$ entries per iteration rather than $\bigO(n^2)$ in Sinkhorn. It is also worth mentioning that Algorithm~\ref{Algorithm:GK} can be implemented such that each iteration runs in only $\bigO(n)$ arithmetic operations~\citep{Altschuler-2017-Near}. 

Despite cheap per-iteration computational cost, it is difficult to quantify the per-iteration progress of Algorithm~\ref{Algorithm:GK} and the proof techniques for Sinkhorn in~\citet{Dvurechensky-2018-Computational} are not applicable here. This motivates us to investigate another proof strategy which will be elaborated in the sequel.
\begin{algorithm}[!t]\small
\caption{Approximating OT by Algorithm~\ref{Algorithm:GK}} \label{Algorithm:ApproxOT_GK}
\begin{algorithmic}
\STATE \textbf{Input:} $\eta = \tfrac{\varepsilon}{4\log(n)}$ and $\varepsilon'=\tfrac{\varepsilon}{8\|C\|_\infty}$. 
\STATE \textbf{Step 1:} Let $\tilde{r} \in \Delta_n$ and $\tilde{c} \in \Delta_n$ be defined by $(\tilde{r}, \tilde{c}) = (1 - \tfrac{\varepsilon'}{8})(r, c) + \tfrac{\varepsilon'}{8n}(\one_n, \one_n)$. 
\STATE \textbf{Step 2:} Compute $\tilde{X} = \textsc{Greenkhorn}(C, \eta, \tilde{r}, \tilde{c}, \tfrac{\varepsilon'}{2})$. 
\STATE \textbf{Step 3:} Round $\tilde{X}$ to $\hat{X}$ using~\citet[Algorithm~2]{Altschuler-2017-Near} such that $\hat{X}\one_n = r$ and $\hat{X}^\top\one_n = c$.
\STATE \textbf{Output:} $\hat{X}$.  
\end{algorithmic}
\end{algorithm} 
\subsection{Complexity analysis---bounding dual objective values}
Given the definition of $E_t$, we first prove the following lemma which yields an upper bound for the objective values of the iterates.
\begin{lemma}\label{lemma:GK-descent}
Letting $\{(u^t, v^t)\}_{t \geq 0}$ be the iterates generated by Algorithm~\ref{Algorithm:GK}, we have
\begin{equation*}
f(u^t, v^t) - f(u^\star, v^\star) \leq 2E_t(\|u^\star\|_\infty + \|v^\star\|_\infty), 
\end{equation*}
where $(u^*, v^*)$ is a point that minimizes $f(u, v) = \|B(u, v)\|_1 - u^\top r - v^\top c$. 
\end{lemma}
\begin{proof}
By the definition, we have
\begin{equation*}
f(u, v) = \sum_{1 \leq i, j \leq n} e^{u_i + v_j - \frac{C_{ij}}{\eta}} - \sum_{i=1}^n u_i r_i - \sum_{j=1}^n v_j c_j. 
\end{equation*}
By definition, we have $\nabla_u f(u^t, v^t) = B(u^t, v^t)\one_n - r$ and $\nabla_v f(u^t, v^t) = B(u^t, v^t)^\top\one_n - c$. Thus, we have $E_t = \|\nabla_u f(u^t, v^t)\|_1 + \|\nabla_v f(u^t, v^t)\|_1$. Since $f$ is convex and minimized at $(u^\star, v^\star)$, we have
\begin{equation*}
f(u^t, v^t) - f(u^\star, v^\star) \leq (u^t - u^\star)^\top \nabla_u f(u^t, v^t) + (v^t - v^\star)^\top \nabla_v f(u^t, v^t). 
\end{equation*}
Combining H$\ddot{\text{o}}$lder's inequality and the definition of $E_t$ yields
\begin{equation}\label{inequality:GK-descent-first}
f(u^t, v^t) - f(u^\star, v^\star) \leq E_t (\|u^t - u^\star\|_\infty + \|v^t - v^\star\|_\infty). 
\end{equation}
Thus, it suffices to show that 
\begin{equation*}
\|u^t - u^\star\|_\infty + \|v^t - v^\star\|_\infty \leq 2\|u^\star\|_\infty + 2\|v^\star\|_\infty. 
\end{equation*}
The next result is the key observation that makes our analysis work for Greenkhorn. We use an induction argument to establish the following bound: 
\begin{equation}\label{inequality:GK-descent-second}
\max\{\|u^t - u^\star\|_\infty, \|v^t - v^\star\|_\infty\} \leq \max\{\|u^0 - u^\star\|_\infty, \|v^0 - v^\star\|_\infty\}. 
\end{equation}
It is clear that Eq.~\eqref{inequality:GK-descent-second} holds true when $t=0$. Suppose that the inequality holds true for $t \leq k_0$, we show that it also holds true for $t = k_0 + 1$. Without loss of generality, let $I$ be the index chosen at the $(k_0+1)$-th iteration. Then
\begin{align}
\|u^{k_0+1} - u^\star\|_\infty & \leq \max \{\|u^{k_0} - u^\star\|_\infty, |u_I^{k_0+1} - u_I^\star|\}, \label{inequality:GK-descent-bound-u} \\ 
\|v^{k_0+1} - v^\star\|_\infty & = \|v^{k_0} - v^\star\|_\infty. \label{inequality:GK-descent-bound-v}
\end{align} 
By the updating formula for $u_I^{k_0+1}$ and the optimality condition for $u_I^\star$, we have
\begin{equation*}
e^{u_I^{k_0+1}} = \frac{r_I}{\sum_{j=1}^n e^{-\frac{C_{ij}}{\eta} + v_j^{k_{0}}}}, \qquad e^{u_I^\star} = \frac{r_I}{\sum_{j=1}^n e^{-\frac{C_{ij}}{\eta} + v_j^\star}}.
\end{equation*}
Putting these pieces together with the inequality that $\frac{\sum_{i=1}^n a_i}{\sum_{i=1}^n b_i} \leq \max_{1 \leq j \leq n} \frac{a_i}{b_i}$ for all $a_i, b_i > 0$ yields
\begin{equation}\label{inequality:GK-descent-third}
|u_I^{k_0+1} - u_I^\star| = \left|\log\left(\frac{\sum_{j=1}^n e^{-\eta^{-1}C_{Ij} + v_j^{k_0}}}{\sum_{j=1}^n e^{-\eta^{-1}C_{Ij} + v_j^\star}}\right)\right| \leq \|v^{k_0} - v^\star\|_\infty.    
\end{equation}
Combining Eq.~\eqref{inequality:GK-descent-bound-u} and Eq.~\eqref{inequality:GK-descent-third} yields 
\begin{equation}\label{inequality:GK-descent-fourth}
\|u^{k_0+1} - u^\star\|_\infty \leq \max\{\|u^{k_0} - u^\star\|_\infty, \|v^{k_0} - v^\star\|_\infty\}. 
\end{equation}
Combining Eq.~\eqref{inequality:GK-descent-bound-v} and Eq.~\eqref{inequality:GK-descent-fourth} further implies Eq.~\eqref{inequality:GK-descent-second}. This together with $u^0 = v^0 = \zero_n$ implies
\begin{equation}\label{inequality:GK-descent-fifth}
\|u^t - u^\star\|_\infty + \|v^t - v^\star\|_\infty \leq 2(\|u^0 - u^\star\|_\infty + \|v^0 - v^\star\|_\infty) = 2\|u^\star\|_\infty + 2\|v^\star\|_\infty.
\end{equation}
Putting Eq.~\eqref{inequality:GK-descent-first} and Eq.~\eqref{inequality:GK-descent-fifth} together yields the desired result. 
\end{proof}
Our second lemma shows that at least one optimal solution $(u^\star, v^\star)$ of $f$ has an upper bound of $\eta^{-1}\|C\|_\infty + \log(n) - 2\log(\min_{1 \leq i, j \leq n} \{r_i, c_j\})$ in $\ell_\infty$-norm. This result is stronger than~\citet[Lemma~1]{Dvurechensky-2018-Computational} and generalizes~\citet[Lemma~10]{Blanchet-2018-Towards}. 
\begin{lemma}\label{lemma:GK-bound}
There exists an optimal solution $(u^\star, v^\star)$ of the function $f$ defined in Eq.~\eqref{def:Greenkhorn-obj} such that the following inequality holds true, 
\begin{equation*}
\|u^\star\|_\infty \leq R, \qquad \|v^\star\|_\infty \leq R, 
\end{equation*}
where $R \mydefn \eta^{-1}\|C\|_\infty + \log(n) - 2\log(\min_{1 \leq i, j \leq n} \{r_i, c_j\})$ depends on $C$, $r$ and $c$. 
\end{lemma}
\begin{proof}
By using the similar argument as in Lemma~\ref{lemma:dual-bound-infinity}, we can first show that there exists an optimal solution pair $(u^\star, v^\star)$ such that (but not for $v^\star$ simultaneously)
\begin{equation}\label{inequality:GK-bound-first}
\max_{1 \leq i \leq n} u_i^\star \geq 0 \geq \min_{1 \leq i \leq n} u_i^\star. 
\end{equation}
Then, we proceed to establish the bounds that are analogous to Eq.~\eqref{claim-dual-bound-u} and~\eqref{claim-dual-bound-v}:
\begin{align}
\max_{1 \leq i \leq n} u_i^\star - \min_{1 \leq i \leq n} u_i^\star & \leq \frac{\|C\|_\infty}{\eta} - \log\left(\min_{1 \leq i, j \leq n} \{r_i, c_j\}\right), \label{inequality:GK-bound-u} \\
\max_{1 \leq i \leq n} v_i^\star - \min_{1 \leq i \leq n} v_i^\star & \leq \frac{\|C\|_\infty}{\eta} - \log\left(\min_{1 \leq i, j \leq n} \{r_i, c_j\}\right). \label{inequality:GK-bound-v}
\end{align}
Indeed, for each $1 \leq i \leq n$, we have
\begin{equation*}
e^{-\eta^{-1}\|C\|_\infty + u_i^\star} \left(\sum_{j=1}^n e^{v_j^\star}\right) \leq \sum_{j=1}^n e^{-\eta^{-1}C_{ij} + u_i^\star + v_j^\star} = [B(u^\star, v^\star)\one_n]_i = r_i \leq  1,
\end{equation*}
which implies $u_i^\star \leq \eta^{-1}\|C\|_\infty - \log(\sum_{j=1}^n e^{v_j^\star})$. Furthermore, we have
\begin{equation*}
e^{u_i^\star}\left(\sum_{j=1}^n e^{v_j^\star}\right) \geq \sum_{j=1}^n e^{-\eta^{-1}C_{ij} + u_i^\star + v_j^\star} = [B(u^\star, v^\star)\one_n]_i = r_i \geq \min_{1 \leq i, j \leq n} \{r_i, c_j\},
\end{equation*}
which implies $u_i^\star \geq \log(\min_{1 \leq i, j \leq n} \{r_i, c_j\}) - \log(\sum_{j=1}^n e^{v_j^\star})$. Putting these pieces together yields Eq.~\eqref{inequality:GK-bound-u}. Using the similar argument, we can prove Eq.~\eqref{inequality:GK-bound-v} holds true. 

Finally, we prove our main results. Since $\max_{1 \leq i \leq n} u_i^\star \geq 0 \geq \min_{1 \leq i \leq n} u_i^\star$, we derive from Eq.~\eqref{inequality:GK-bound-u} that 
\begin{equation*}
- \frac{\|C\|_\infty}{\eta} + \log\left(\min_{1 \leq i, j \leq n} \{r_i, c_j\}\right) \leq \min_{1 \leq i \leq n} u_i^\star \leq \max_{1 \leq i \leq n} u_i^\star \leq \frac{\|C\|_\infty}{\eta} - \log\left(\min_{1 \leq i, j \leq n} \{r_i, c_j\}\right). 
\end{equation*}
This implies that $\|u^\star\|_\infty \leq R$. Then, we bound $\|v^\star\|_\infty$ by considering two different cases. 

For the former case, we assume that $\max_{1 \leq i \leq n} v_i^\star \geq 0$. Note that the optimality condition is $\sum_{i, j=1}^n e^{-\eta^{-1}C_{ij} + u_i^\star + v_j^\star} = 1$ and further implies that 
\begin{equation*}
\max_{1 \leq i \leq n} u_i^\star + \max_{1 \leq i \leq n} v_i^\star \leq \log\left(\max_{1 \leq i, j \leq n} e^{\eta^{-1}C_{ij}}\right) = \frac{\|C\|_\infty}{\eta}. 
\end{equation*}
Since $\max_{1 \leq i \leq n} u_i^* \geq 0$ and $\max_{1 \leq i \leq n} v_i^* \geq 0$, we have $0 \leq \max_{1 \leq i \leq n} v_i^\star \leq \frac{\|C\|_\infty}{\eta}$. Combining $\max_{1 \leq i \leq n} v_i^* \geq 0$ with Eq.~\eqref{inequality:GK-bound-v} yields that 
\begin{equation*}
 \min_{1 \leq i \leq n} v_i^\star \geq -\frac{\|C\|_\infty}{\eta} + \log\left(\min_{1 \leq i, j \leq n} \{r_i, c_j\}\right).
\end{equation*}
which implies that $\|v^\star\|_\infty \leq R$.

For the latter case, we assume that $\max_{1 \leq i \leq n} v_i^\star \leq 0$. Then, we have
\begin{equation*}
\min_{1 \leq i \leq n} v_i^\star \geq \log\left(\min_{1 \leq i, j \leq n} \{r_i, c_j\}\right) - \log\left(\sum_{i=1}^n e^{u_i^\star}\right). 
\end{equation*}
This together with $\|u^\star\|_\infty \leq \frac{\|C\|_\infty}{\eta} - \log(\min_{1 \leq i, j \leq n} \{r_i, c_j\})$ yields that $\|v^\star\|_\infty \leq R$.
\end{proof}
Putting Lemma~\ref{lemma:GK-descent} and~\ref{lemma:GK-bound} together, we have the following straightforward consequence: 
\begin{corollary}\label{corollary:GK-objgap}
Letting $\{(u^t, v^t)\}_{t \geq 0}$ be the iterates generated by Algorithm~\ref{Algorithm:GK}, we have
\begin{equation*}
f(u^t, v^t) - f(u^\star, v^\star) \leq 4RE_t.
\end{equation*}
\end{corollary}
\begin{remark}
The notation $R$ is also used in~\citet{Dvurechensky-2018-Computational} and has the same order as ours since $R$ in our paper only involves an term $\log(n) - \log(\min_{1\leq i, j \leq n} \{r_i, c_j\})$.
\end{remark}
\begin{remark}
We further comment on the proof techniques in this paper and~\cite{Dvurechensky-2018-Computational}. Indeed, the proof for~\citet[Lemma~2]{Dvurechensky-2018-Computational} depends on taking full advantage of the shift property of Sinkhorn; namely, either $B(\overline{u}^t, \overline{v}^t)\one_n = r$ or $B(\overline{u}^t, \overline{v}^t)^\top\one_n = c$ where $(\overline{u}^t, \overline{v}^t)$ stands for the iterate generated by Sinkhorn. Unfortunately, Greenkhorn does not enjoy such a shift property. We have thus proposed a different approach for bounding $f(u^t, v^t) - f(u^\star, v^\star)$ via appeal to the $\ell_\infty$-norm of the solution $(u^\star, v^\star)$.
\end{remark}

\subsection{Complexity analysis---bounding the number of iterations}
We proceed to provide an upper bound for the iteration number to achieve a desired tolerance $\varepsilon'$ in Algorithm~\ref{Algorithm:GK}. First, we start with a lower bound for the difference of function values between two consecutive iterates of Algorithm~\ref{Algorithm:GK}:
\begin{lemma} \label{lemma:GK-objgap}
Letting $\{(u^t, v^t)\}_{t \geq 0}$ be the iterates generated by Algorithm~\ref{Algorithm:GK}, we have
\begin{equation*}
f(u^t, v^t) - f(u^{t+1}, v^{t+1}) \geq \frac{(E_t)^2}{28n}. 
\end{equation*}
\end{lemma}
\begin{proof}
Combining~\citet[Lemma~5]{Altschuler-2017-Near} and the fact that the row or column update is chosen in a greedy manner, we have
\begin{equation*}
f(u^t, v^t) - f(u^{t+1}, v^{t+1}) \geq \frac{1}{2n}\left(\rho(r, r(B(u^t, v^t)) + \rho(c, c(B(u^t, v^t))\right). 
\end{equation*}
Furthermore,~\citet[Lemma~6]{Altschuler-2017-Near} implies that 
\begin{equation*}
\rho(r, r(B(u^t, v^t)) + \rho(c, c(B(u^t, v^t)) \geq \frac{1}{7}\left(\|r - r(B(u^t, v^t))\|_1^2 + \|c - c(B(u^t, v^t))\|_1^2\right). 
\end{equation*}
Putting these pieces together yields that 
\begin{equation*}
f(u^t, v^t) - f(u^{t+1}, v^{t+1}) \geq \frac{1}{14n}\left(\|r - r(B(u^t, v^t))\|_1^2 + \|c - c(B(u^t, v^t))\|_1^2\right). 
\end{equation*}
Combining the above inequality with the definition of $E_t$ implies the desired result.
\end{proof}
We are now able to derive the iteration complexity of Algorithm~\ref{Algorithm:GK}. 
\begin{theorem}\label{Theorem:GK-Total-Complexity}
Letting $\{(u^t, v^t)\}_{t \geq 0}$ be the iterates generated by Algorithm~\ref{Algorithm:GK}, the number of iterations required to satisfy $E_t \leq \varepsilon'$ is upper bounded by $t \leq 2 + \tfrac{112nR}{\varepsilon'}$ where $R > 0$ is defined in Lemma~\ref{lemma:GK-bound}.
\end{theorem}
\begin{proof}
Letting $\delta_t = f(u^t, v^t) - f(u^\star, v^\star)$, we derive from Corollary~\ref{corollary:GK-objgap} and Lemma~\ref{lemma:GK-objgap} that
\begin{equation*}
\delta_t - \delta_{t+1} \geq \max\left\{\frac{\delta_t^2}{448nR^2}, \frac{(\varepsilon')^2}{28n}\right\}, 
\end{equation*}
where $E_t \geq \varepsilon'$ as soon as the stopping criterion is not fulfilled. In the following step we apply a switching strategy introduced by~\citet{Dvurechensky-2018-Computational}. Given any $t \geq 1$, we have two estimates: 
\begin{itemize}
\item[(i)] Considering the process from the first iteration and the $t$-th iteration, we have
\begin{equation*}
\frac{\delta_{t+1}}{448nR^2} \leq \frac{1}{t + 448nR^2\delta_1^{-2}} \ \Longrightarrow \ t \leq 1 + \frac{448nR^2}{\delta_t} - \frac{448nR^2}{\delta_1}. 
\end{equation*}
\item[(ii)] Considering the process from the $(t+1)$-th iteration to the $(t+m)$-th iteration for any $m \geq 1$, we have
\begin{equation*}
\delta_{t+m} \leq \delta_t - \frac{(\varepsilon')^2 m}{28n} \ \Longrightarrow \ m \leq \frac{28n(\delta_t - \delta_{t+m})}{(\varepsilon')^2}. 
\end{equation*}
\end{itemize}
We then minimize the sum of two estimates by an optimal choice of $s \in (0, \delta_1]$: 
\begin{equation*}
t \leq \min_{0 < s \leq \delta_1} \left(2 + \frac{448nR^2}{s} - \frac{448nR^2}{\delta_1} + \frac{28ns}{(\varepsilon')^2}\right) = 
\left\{\begin{array}{ll}
2 + \frac{224nR}{\varepsilon'} - \frac{448nR^2}{\delta_1}, & \delta_1 \geq 4R\varepsilon', \\
2 + \frac{28n\delta_1}{(\varepsilon')^2}, & \delta_1 \leq 4R\varepsilon'. 
\end{array}
\right. 
\end{equation*}
This implies that $t \leq 2 + \frac{112nR}{\varepsilon'}$ in both cases and completes the proof. 
\end{proof}
Equipped with the result of Theorem~\ref{Theorem:GK-Total-Complexity} and the scheme of Algorithm~\ref{Algorithm:ApproxOT_GK}, we are able to establish the following result for the complexity of Algorithm~\ref{Algorithm:ApproxOT_GK}:
\begin{theorem}\label{Theorem:ApproxOT-GK-Total-Complexity}
The Greenkhorn scheme for approximating optimal transport (Algorithm~\ref{Algorithm:ApproxOT_GK}) returns an $\varepsilon$-approximate transportation plan (cf. Definition~\ref{def:eps-approximation}) in
\begin{equation*}
\bigO\left(\frac{n^2\left\|C\right\|_\infty^2\log(n)}{\varepsilon^2}\right)
\end{equation*}
arithmetic operations. 
\end{theorem}
\begin{proof}
We follow the proof steps in~\cite[Theorem~1]{Altschuler-2017-Near} and obtain that the transportation plan $\hat{X}$ returned by Algorithm~\ref{Algorithm:ApproxOT_GK} satisfies that 
\begin{eqnarray*}
\langle C, \hat{X}\rangle - \langle C, X^\star\rangle & \leq & 2\eta\log(n) + 4(\|\tilde{X}\one_n - r\|_1 + \|\tilde{X}^\top\one_n - c\|_1)\|C\|_\infty \\
& \leq & \frac{\varepsilon}{2} + 4(\|\tilde{X}\one_n - r\|_1 + \|\tilde{X}^\top\one_n - c\|_1)\|C\|_\infty, 
\end{eqnarray*}
where $X^\star$ is an optimal solution to the OT problem and $\tilde{X} = \textsc{Greenkhorn}(C, \eta, \tilde{r}, \tilde{c}, \tfrac{\varepsilon'}{2})$. The last inequality in the above display holds true since $\eta = \frac{\varepsilon}{4\log(n)}$. Furthermore,
\begin{eqnarray*}
\|\tilde{X}\one_n - r\|_1 + \|\tilde{X}^\top\one_n - c\|_1 & \leq & \|\tilde{X}\one_n - \tilde{r}\|_1 + \|\tilde{X}^\top\one_n - \tilde{c}\|_1 + \|r - \tilde{r}\|_1 + \|c - \tilde{c}\|_1 \\
& \leq & \frac{\varepsilon'}{2} + \frac{\varepsilon'}{4} + \frac{\varepsilon'}{4} = \varepsilon'.
\end{eqnarray*}
Putting these pieces together with $\varepsilon' = \frac{\varepsilon}{8\|C\|_\infty}$ yields that $\langle C, \hat{X}\rangle - \langle C, X^\star\rangle \leq \varepsilon$. 

The remaining step is to analyze the complexity bound. It follows from Theorem~\ref{Theorem:GK-Total-Complexity} and the definition of $\tilde{r}$ and $\tilde{c}$ in Algorithm~\ref{Algorithm:ApproxOT_GK} that 
\begin{eqnarray*}
t \ \leq \ 2 + \frac{112nR}{\varepsilon'} & \leq & 2 + \frac{96n\|C\|_\infty}{\varepsilon}\biggr(\frac{\|C\|_\infty}{\eta} + \log(n) - 2\log\left(\min_{1 \leq i, j \leq n} \{r_i, c_j\}\right)\biggr) \\
& \leq & 2 + \frac{96n\|C\|_\infty}{\varepsilon}\left(\frac{4\|C\|_\infty\log(n)}{\varepsilon} + \log(n) - 2\log\left(\frac{\varepsilon}{64n\|C\|_\infty}\right)\right) \\
& = & \bigO\left(\frac{n\|C\|_\infty^2\log(n)}{\varepsilon^2}\right). 
\end{eqnarray*}
The total iteration complexity in Step 2 of Algorithm~\ref{Algorithm:ApproxOT_GK} is bounded by $\bigO(n\|C\|_\infty^2\log(n)\varepsilon^{-2})$. Each iteration of Algorithm~\ref{Algorithm:GK} requires $\bigO(n)$ arithmetic operations. Thus, the total number of arithmetic operations is $\bigO(n^2\|C\|_\infty^2\log(n)\varepsilon^{-2})$. Moreover, $\tilde{r}$ and $\tilde{c}$ in Step 1 of Algorithm~\ref{Algorithm:ApproxOT_GK} can be found in $\bigO(n)$ arithmetic operations and~\citet[Algorithm~2]{Altschuler-2017-Near} requires $O(n^2)$ arithmetic operations. Therefore, we conclude that the total number of arithmetic operations is $\bigO(n^2\|C\|_\infty^2\log(n)\varepsilon^{-2})$. 
\end{proof}
The complexity results presented in Theorem~\ref{Theorem:ApproxOT-GK-Total-Complexity} improve the best known complexity bound $\bigOtil(n^2\epsilon^{-3})$ of Greenkhorn~\citep{Altschuler-2017-Near, Abid-2018-Greedy}, Notably, it matches the best known complexity bound of Sinkhorn~\citep{Dvurechensky-2018-Computational}. The key feature of our analysis is that the per-iteration progress of Greenkhorn can be lower bounded by a new quantity (cf. Lemmas~\ref{lemma:GK-descent} and~\ref{lemma:GK-bound}). It allows us to apply the switching strategy in Theorem~\ref{Theorem:GK-Total-Complexity} to improve the complexity upper bound of Greenkhorn.

In practice, Greenkhorn has been reported to outperform Sinkhorn~\citep{Altschuler-2017-Near} in terms of row/column updates and our improved complexity bound can provide the theoretical justification for this phenomenon.

%% file: sec/APDAMD.tex
\section{Adaptive Primal-Dual Accelerated Mirror Descent}\label{sec:apdamd}
In this section, we propose an adaptive primal-dual accelerated mirror descent (APDAMD) for solving the entropic regularized OT problem in Eq.~\eqref{prob:OT_regularized}. APDAMD and its application to the OT problem are presented in Algorithm~\ref{Algorithm:APDAMD} and~\ref{Algorithm:ApproxOT_APDAMD}. We prove the complexity bound of $\bigO(n^2\sqrt{\delta}\|C\|_\infty \log(n)\varepsilon^{-1})$ where $\delta>0$ stands for the regularity of the mirror mapping $\phi$. 
\begin{algorithm}[!t]\small
\caption{\textsc{Apdamd}$(\varphi, A, b, \varepsilon')$} \label{Algorithm:APDAMD}
\begin{algorithmic}
\STATE \textbf{Input:} $t = 0$. 
\STATE \textbf{Initialization:} $\bar{\alpha}^0 = \alpha^0 = 0$,  $z^0 = \mu^0 = \lambda^0 = \zero_{2n}$ and $L^0 = 1$. 
\REPEAT
\STATE Set $M^t = \frac{L^t}{2}$. 
\REPEAT
\STATE Set $M^t = 2M^t$. 
\STATE Compute the stepsize: $\alpha^{t+1} = \frac{1 + \sqrt{1 + 4\delta M^t\bar{\alpha}^t}}{2\delta M^t}$. 
\STATE Compute the average coefficient: $\bar{\alpha}^{t+1} = \bar{\alpha}^t + \alpha^{t+1}$. 
\STATE Compute the first average step: $\mu^{t+1} = \frac{\alpha^{t+1}z^t + \bar{\alpha}^t\lambda^t}{\bar{\alpha}^{t+1}}$. 
\STATE Compute the mirror descent: $z^{t+1} = \argmin\limits_{z \in \br^n} \{(z - \mu^{t+1})^\top \nabla \varphi(\mu^{t+1}) + \frac{B_\phi(z, z^t)}{\alpha^{t+1}}\}$. 
\STATE Compute the second average step: $\lambda^{t+1} = \frac{\alpha^{t+1}z^{t+1} + \bar{\alpha}^t\lambda^t}{\bar{\alpha}^{t+1}}$. 
\UNTIL{$\varphi(\lambda^{t+1}) - \varphi(\mu^{t+1}) - (\lambda^{t+1} - \mu^{t+1})^\top \nabla \varphi(\mu^{t+1}) \leq \frac{M^t}{2}\|\lambda^{t+1} - \mu^{t+1}\|_\infty^2$.}
\STATE Compute the main average step: $x^{t+1} = \frac{\alpha^{t+1} x(\mu^{t+1}) + \bar{\alpha}^t x^t}{\bar{\alpha}^{t+1}}$.  
\STATE Set $L^{t+1} = \frac{M^t}{2}$. 
\STATE Set $t = t + 1$. 
\UNTIL{$\|Ax^t - b\|_1 \leq \varepsilon'$.}
\STATE \textbf{Output:} $X^t$ where $x^t = \text{vec}(X^t)$.  
\end{algorithmic}
\end{algorithm}

\subsection{General setup}
We follow the setup in Section~\ref{sec:setup} and consider the following generalization of the entropic regularized OT problem in Eq.~\eqref{prob:OT_regularized}:
\begin{equation}
\min_{x \in Q} \ f(x), \quad \st \ Ax = b, \label{eq:general_problem}
\end{equation}
where $f$ is strongly convex with respect to the $\ell_1$-norm on the set $Q$: 
\begin{equation*}
f(x') - f(x) - (x' - x)^\top\nabla f(x) \geq \frac{\eta}{2}\|x' - x\|_1^2 \textnormal{ for any } x', x \in Q.  
\end{equation*}
Note that, in the specific setting of the entropic regularized OT problem, the function $f(x) = \sum_{i, j} C_{ij} x_{j+n(i-1)} + \eta \cdot x_{j+n(i-1)} \log(x_{j+n(i-1)})$ where $x_{j + n(i - 1)} = X_{ij}$ for any $i, j$ where $X$ is the transportation plan in equation~\eqref{prob:OT_regularized}, and the vector $b \in \mathbb{R}^{2n \times 1}$ is defined as: $b_{i} = r_{i}$ as $1 \leq i \leq n$ and $b_{i} = c_{i - n}$ when $n + 1 \leq i \leq 2n$. Furthermore, the matrix $A = (A_{ij}) \in \mathbb{R}^{2n \times n^2}$ is defined as: When $1 \leq i \leq n$, we denote $A_{ij} = 1$ if $1 + n(i - 1) \leq j \leq n \cdot i$ and $0$ otherwise; When $n + 1 \leq i \leq 2n$, we define $A_{ij} = 1$ if $j \in \{i - n + n(l - 1): 1 \leq l \leq n \}$ and 0 otherwise. To be consistent with the notations in Algorithms~\ref{Algorithm:ApproxOT_APDAMD} and~\ref{Algorithm:ApproxOT_APDAGD}, we specifically denote $A_{\text{ot}}$ as the matrix $A$ corresponding to the entropic regularized OT problem.

After some calculations with the general problem~\eqref{eq:general_problem}, we obtain that the dual problem is as follows:
\begin{equation}\label{prob:dualregOT-APDAMD}
\min_{\lambda \in \br^{2n}} \widetilde{\varphi}(\lambda) \mydefn \{\langle\lambda, b\rangle + \max_{x \in \br^{n^2}} \{-f(x) - \langle A^\top\lambda, x\rangle\}\}, 
\end{equation}
and $\nabla \widetilde{\varphi}(\lambda) = b - Ax(\lambda)$ where $x(\lambda) = \argmax_{x \in \br^n} \{-f(x) - \langle A^\top\lambda, x\rangle\}$; see the explicit form in Eq.~\eqref{def:dual-regOT-objective} with $\lambda = (\alpha, \beta)$. By~\citet[Theorem~1]{Nesterov-2005-Smooth} with $\ell_1$-norm for the dual space of the Lagrange multipliers, the dual objective function $\tilde{\varphi}$ satisfies the following inequality:  
\begin{equation}\label{def:varphi-smoothness}
\widetilde{\varphi}(\lambda') - \widetilde{\varphi}(\lambda) - (\lambda' - \lambda)^\top \nabla\widetilde{\varphi}(\lambda) \leq \frac{\|A\|_{1 \rightarrow 1}^2}{2\eta}\|\lambda' - \lambda\|_\infty^2.  
\end{equation}
In the entropic regularized OT problem, each column of the matrix $A_{\text{ot}}$ contains no more than two nonzero elements which are equal to one. Since $\|A_{\text{ot}}\|_{1 \rightarrow 1}$ is equal to maximum $\ell_1$-norm of the column of this matrix, we have $\|A_{\text{ot}}\|_{1 \rightarrow 1} = 2$. Thus, the dual objective function $\widetilde{\varphi}$ is $\frac{4}{\eta}$-gradient Lipschitz with respect to the $\ell_\infty$-norm. 

In addition, we define the Bregman divergence $B_\phi: \br^{2n} \times \br^{2n} \mapsto [0, +\infty)$ by 
\begin{equation*}
B_\phi(\lambda', \lambda) \mydefn \phi(\lambda') - \phi(\lambda) - (\lambda' - \lambda)^\top \nabla \phi(\lambda), 
\end{equation*}
where the mirror mapping $\phi$ is $\frac{1}{\delta}$-strongly convex and 1-smooth on $\br^{2n}$ with respect to $\ell_\infty$-norm; that is, 
\begin{equation*}
\frac{1}{2\delta}\|\lambda' - \lambda\|_\infty^2 \leq \phi(\lambda') - \phi(\lambda) - (\lambda'-\lambda)^\top\nabla\phi(\lambda) \leq \frac{1}{2}\|\lambda' - \lambda\|_\infty^2.
\end{equation*}
For example, we can choose $\phi(\lambda) = \frac{1}{2n}\|\lambda\|^2$ and $B_\phi(\lambda', \lambda) = \frac{1}{2n}\|\lambda'-\lambda\|^2$ in APDAMD where $\delta=n$. As such, $\delta > 0$ is a function of $n$ in general and it will appear in the complexity bound of APDAMD for approximating the OT problem (cf. Theorem~\ref{Theorem:ApproxOT-APDAMD-Total-Complexity}). It is worth noting that our algorithm uses a regularizer that acts only in the dual and our complexity bound is the best existing one among this group of algorithms~\citep{Dvurechensky-2018-Computational, Guo-2020-Fast, Guminov-2021-Combination}. A very recent work of~\citet{Jambulapati-2019-Direct} showed that the complexity bound can be improved to $\bigOtil(n^2\varepsilon^{-1})$ using a more advanced area-convex mirror mapping~\citep{Sherman-2017-Area}. 
\begin{algorithm}[!t]\small
\caption{Approximating OT by Algorithm~\ref{Algorithm:APDAMD}} \label{Algorithm:ApproxOT_APDAMD}
\begin{algorithmic}
\STATE \textbf{Input:} $\eta = \frac{\varepsilon}{4\log(n)}$ and $\varepsilon' = \frac{\varepsilon}{8\|C\|_\infty}$. 
\STATE \textbf{Step 1:} Let $\tilde{r} \in \Delta_n$ and $\tilde{c} \in \Delta_n$ be defined by $(\tilde{r}, \tilde{c}) = (1 - \frac{\varepsilon'}{8})(r, c) + \frac{\varepsilon'}{8n}(\one_n, \one_n)$. 
\STATE \textbf{Step 2:} Let $A_{\text{ot}} \in \br^{2n \times n^2}$ and $b \in \br^{2n}$ be defined by $A_{\text{ot}}\text{vec}(X) = \begin{pmatrix} X\one_n \\ X^\top\one_n \end{pmatrix}$ and $b = \begin{pmatrix} \tilde{r} \\ \tilde{c} \end{pmatrix}$. 
\STATE \textbf{Step 3:} Compute $\tilde{X} = \textsc{Apdamd}(\widetilde{\varphi}, A_{\text{ot}}, b, \tfrac{\varepsilon'}{2})$ where $\widetilde{\varphi}$ is defined by Eq.~\eqref{prob:dualregOT-APDAMD}. 
\STATE \textbf{Step 4:} Round $\tilde{X}$ to $\hat{X}$ using~\citet[Algorithm~2]{Altschuler-2017-Near} such that $\hat{X}\one_n = r$ and $\hat{X}^\top\one_n = c$. 
\STATE \textbf{Output:} $\hat{X}$.  
\end{algorithmic}
\end{algorithm} 
\subsection{Properties of APDAMD}
We present several important properties of Algorithm~\ref{Algorithm:APDAMD} that can be used later for entropic regularized OT problems. First, we prove the following result regarding the number of line search iterations in Algorithm~\ref{Algorithm:APDAMD} for solving the entropic regularized OT problem: 
\begin{lemma}\label{lemma:LS-bound}
The number of line search iterations in Algorithm~\ref{Algorithm:APDAMD} for solving the entropic OT problem is finite. Furthermore, the total number of gradient oracle calls after the $t$-th iteration is bounded as
\begin{equation*}
N_t \leq 4t + 4 + \frac{2\log(\frac{8}{\eta}) - 2\log(L^0)}{\log 2}. 
\end{equation*}
\end{lemma}
\begin{proof}
First, we observe that multiplying $M^t$ by two will not stop until the line search stopping criterion is satisfied. Then, Eq.~\eqref{def:varphi-smoothness} implies that the number of line search iterations in the line search strategy is finite and $M^t \leq \frac{2\|A_{\text{ot}}\|_{1\rightarrow 1}^2}{\eta}$ holds true for all $t \geq 0$. Otherwise, the line search stopping criterion is satisfied with $\frac{M^t}{2}$ since $\frac{M^t}{2} \geq \frac{\|A_{\text{ot}}\|_{1\rightarrow 1}^2}{\eta}$.

Letting $i_j$ denote the total number of multiplication at the $j$-th iteration, we have 
\begin{equation*}
i_0 \leq 1 + \frac{\log(\frac{M^0}{L^0})}{\log 2}, \qquad i_j \leq 2 + \frac{\log(\frac{M^j}{M^{j-1}})}{\log 2}. 
\end{equation*}
Then, the total number of line search iterations is bounded by
\begin{equation*}
\sum_{j=0}^t i_j \leq 1 + \frac{\log(\frac{M^0}{L^0})}{\log 2} + \sum_{j=1}^t \left(2 + \frac{\log(\frac{M^j}{M^{j-1}})}{\log 2}\right) \leq 2t + 1 + \frac{\log(\frac{2\|A_{\text{ot}}\|_{1\rightarrow 1}^2}{\eta}) - \log(L^0)}{\log 2}. 
\end{equation*}
Since each line search contains two gradient oracle calls and $\|A_{\text{ot}}\|_{1 \rightarrow 1} = 2$, we conclude the desired upper bound for the total number of gradient oracle calls after the $t$-th iteration. 
\end{proof}
The next lemma presents a property of the function $\widetilde{\varphi}$ in Algorithm~\ref{Algorithm:APDAMD}.
\begin{lemma}\label{lemma:ES-bound}
For each iteration $t$ of Algorithm~\ref{Algorithm:APDAMD} and any $z \in \br^{2n}$, we have
\begin{equation*}
\bar{\alpha}^t\widetilde{\varphi}(\lambda^t) \leq \sum_{j=0}^t (\alpha^j (\widetilde{\varphi}(\mu^j) + (z - \mu^j)^\top \nabla\widetilde{\varphi}(\mu^j))) + \|z\|_\infty^2.  
\end{equation*}
\end{lemma}
\begin{proof}
First, we claim that it holds for any $z \in \br^n$:  
\begin{equation} \label{inequality:ES-bound-first}
\alpha^{t+1} (z^t - z)^\top \nabla\widetilde{\varphi}(\mu^{t+1}) \leq \bar{\alpha}^{t+1}(\widetilde{\varphi}(\mu^{t+1}) - \widetilde{\varphi}(\lambda^{t+1})) + B_\phi(z, z^t) - B_\phi(z, z^{t+1}). 
\end{equation}
Indeed, the optimality condition in mirror descent implies that, for any $z \in \br^{2n}$, we have
\begin{equation*}
(z - z^{t+1})^\top \left(\nabla \widetilde{\varphi}(\mu^{t+1}) + \tfrac{\nabla \phi(z^{t+1}) - \nabla \phi(z^t)}{\alpha^{t+1}}\right) \geq 0. 
\end{equation*}
By definition, we have $B_\phi(z, z^t) - B_\phi(z, z^{t+1}) - B_\phi(z^{t+1}, z^t) = (z - z^{t+1})^\top(\nabla \phi(z^{t+1}) - \nabla \phi(z^t))$ and $B_\phi(z^{t+1}, z^t) \geq \frac{1}{2\delta}\|z^{t+1} - z^t\|_\infty^2$. Putting these pieces together yields that 
\begin{eqnarray}\label{inequality:ES-bound-second}
\lefteqn{\alpha^{t+1} (z^t - z)^\top \nabla \widetilde{\varphi}(\mu^{t+1}) = \alpha^{t+1}(z^t - z^{t+1})^\top \nabla \widetilde{\varphi}(\mu^{t+1}) + \alpha^{t+1}(z^{t+1} - z)^\top \nabla \widetilde{\varphi}(\mu^{t+1})} \nonumber \\
& \leq & \alpha^{t+1}(z^t - z^{t+1})^\top \nabla \widetilde{\varphi}(\mu^{t+1}) + (z - z^{t+1})^\top(\nabla \phi(z^{t+1}) - \nabla \phi(z^t)) \nonumber \\
& = & \alpha^{t+1}(z^t - z^{t+1})^\top \nabla \widetilde{\varphi}(\mu^{t+1}) + B_\phi(z, z^t) - B_\phi(z, z^{t+1}) - B_\phi(z^{t+1}, z^t) \nonumber \\
& \leq & \alpha^{t+1}(z^t - z^{t+1})^\top \nabla \widetilde{\varphi}(\mu^{t+1}) + B_\phi(z, z^t) - B_\phi(z, z^{t+1}) - \tfrac{\|z^{t+1} - z^t\|_\infty^2}{2\delta}. 
\end{eqnarray}
The update formulas of $\mu^{t+1}$, $\lambda^{t+1}$, $\alpha^{t+1}$ and $\bar{\alpha}^{t+1}$ imply that  
\begin{equation*}
\lambda^{t+1} - \mu^{t+1} = \frac{\alpha^{t+1}}{\bar{\alpha}^{t+1}}(z^{t+1} - z^t), \qquad \delta M^t(\alpha^{t+1})^2 = \bar{\alpha}^{t+1}.
\end{equation*}
Therefore, we have
\begin{equation*}
\alpha^{t+1}(z^t - z^{t+1})^\top \nabla \widetilde{\varphi}(\mu^{t+1}) = \bar{\alpha}^{t+1}(\mu^{t+1} - \lambda^{t+1})^\top \nabla \widetilde{\varphi}(\mu^{t+1}), 
\end{equation*}
and 
\begin{equation*}
\|z^{t+1} - z^t\|_\infty^2 = (\tfrac{\bar{\alpha}^{t+1}}{\alpha^{t+1}})^2\|\mu^{t+1} - \lambda^{t+1}\|_\infty^2 = \delta M^t\bar{\alpha}^{t+1}\|\mu^{t+1} - \lambda^{t+1}\|_\infty^2. 
\end{equation*}
Putting these pieces together with Eq.~\eqref{inequality:ES-bound-second} yields that  
\begin{eqnarray*}
\lefteqn{\alpha^{t+1} (z^t - z)^\top \nabla \widetilde{\varphi}(\mu^{t+1})} \\
& \leq & \bar{\alpha}^{t+1}(\mu^{t+1} - \lambda^{t+1})^\top \nabla \widetilde{\varphi}(\mu^{t+1}) + B_\phi(z, z^t) - B_\phi(z, z^{t+1}) - \tfrac{\bar{\alpha}^{t+1} M^t}{2}\|\mu^{t+1} - \lambda^{t+1}\|_\infty^2 \\
& = & \bar{\alpha}^{t+1}\left((\mu^{t+1} - \lambda^{t+1})^\top \nabla \widetilde{\varphi}(\mu^{t+1}) - \tfrac{M^t}{2}\|\mu^{t+1} - \lambda^{t+1}\|_\infty^2\right) + B_\phi(z, z^t) - B_\phi(z, z^{t+1}) \\
& \leq & \bar{\alpha}^{t+1}(\widetilde{\varphi}(\mu^{t+1}) - \widetilde{\varphi}(\lambda^{t+1})) + B_\phi(z, z^t) - B_\phi(z, z^{t+1}), \nonumber
\end{eqnarray*}
where the last inequality comes from the stopping criterion in the line search. This implies that Eq.~\eqref{inequality:ES-bound-first} holds true. 

The next step is to bound the iterative objective gap given by 
\begin{eqnarray}\label{inequality:ES-bound-third}
\lefteqn{\bar{\alpha}^{t+1}\widetilde{\varphi}(\lambda^{t+1}) - \bar{\alpha}^t\widetilde{\varphi}(\lambda^t)} \\
& \leq & \alpha^{t+1}(\widetilde{\varphi}(\mu^{t+1}) + (z - \mu^{t+1})^\top \nabla\widetilde{\varphi}(\mu^{t+1})) + B_\phi(z, z^t) - B_\phi(z, z^{t+1}). \nonumber
\end{eqnarray}
Indeed, by combining $\bar{\alpha}^{t+1} = \bar{\alpha}^t + \alpha^{t+1}$ and the update formula of $\mu^{t+1}$, we have
\begin{equation*}
\alpha^{t+1}(\mu^{t+1} - z^t) = (\bar{\alpha}^{t+1} - \bar{\alpha}^t)\mu^{t+1} - \alpha^{t+1} z^t = \alpha^{t+1}z^t + \bar{\alpha}^t\lambda^t - \bar{\alpha}^t\mu^{t+1} - \alpha^{t+1} z^t = \bar{\alpha}^t(\lambda^t - \mu^{t+1}).
\end{equation*}
This together with the convexity of $\widetilde{\varphi}$ implies that 
\begin{eqnarray*}
\lefteqn{\alpha^{t+1} (\mu^{t+1} - z)^\top \nabla \widetilde{\varphi}(\mu^{t+1})} \\
& = & \alpha^{t+1} (\mu^{t+1} - z^t)^\top \nabla \widetilde{\varphi}(\mu^{t+1}) + \alpha^{t+1} (z^t - z)^\top \nabla \widetilde{\varphi}(\mu^{t+1}) \\
& = & \bar{\alpha}^t(\lambda^t - \mu^{t+1})^\top \nabla \widetilde{\varphi}(\mu^{t+1}) + \alpha^{t+1} (z^t - z)^\top \nabla \widetilde{\varphi}(\mu^{t+1}) \\
& \leq & \bar{\alpha}^t(\widetilde{\varphi}(\lambda^t) - \widetilde{\varphi}(\mu^{t+1})) + \alpha^{t+1} (z^t - z)^\top \nabla \widetilde{\varphi}(\mu^{t+1}). 
\end{eqnarray*}
Furthermore, we derive from Eq.~\eqref{inequality:ES-bound-first} and $\bar{\alpha}^{t+1} = \bar{\alpha}^t + \alpha^{t+1}$ that 
\begin{eqnarray*}
\lefteqn{\bar{\alpha}^t(\widetilde{\varphi}(\lambda^t) - \widetilde{\varphi}(\mu^{t+1})) + \alpha^{t+1} (z^t - z)^\top \nabla \widetilde{\varphi}(\mu^{t+1})} \\
& \leq & \bar{\alpha}^t(\widetilde{\varphi}(\lambda^t) - \widetilde{\varphi}(\mu^{t+1})) + \bar{\alpha}^{t+1}(\widetilde{\varphi}(\mu^{t+1}) - \widetilde{\varphi}(\lambda^{t+1})) + B_\phi(z, z^t) - B_\phi(z, z^{t+1}) \\
& = & \bar{\alpha}^t\widetilde{\varphi}(\lambda^t) - \bar{\alpha}^{t+1}\widetilde{\varphi}(\lambda^{t+1}) + \alpha^{t+1}\widetilde{\varphi}(\mu^{t+1}) + B_\phi(z, z^t) - B_\phi(z, z^{t+1}). 
\end{eqnarray*}
Putting these pieces together yields that Eq.~\eqref{inequality:ES-bound-third} holds true. 

Finally, we prove our main results. By changing the index $t$ to $j$ in Eq.~\eqref{inequality:ES-bound-third} and summing up the resulting inequality over $j = 0, 1, \ldots, t-1$, we have
\begin{equation*}
\bar{\alpha}^t\widetilde{\varphi}(\lambda^t) - \bar{\alpha}^0\widetilde{\varphi}(\lambda^0) \leq \sum_{j=0}^{t-1} (\alpha^{j+1}(\widetilde{\varphi}(\mu^{j+1}) + (z - \mu^{j+1})^\top \nabla\widetilde{\varphi}(\mu^{j+1}))) + B_\phi(z, z^0) - B_\phi(z, z^t). 
\end{equation*}
Since $\alpha^0 = \bar{\alpha}^0 = 0$, $B_\phi(z, z^t) \geq 0$ and $\phi$ is 1-smooth with respect to $\ell_\infty$-norm, we have
\begin{eqnarray*}
\lefteqn{\bar{\alpha}^t\widetilde{\varphi}(\lambda^t) \ \leq \ \sum_{j=0}^t (\alpha^j (\widetilde{\varphi}(\mu^j) + (z - \mu^j)^\top \nabla \widetilde{\varphi}(\mu^j))) + B_\phi(z, z^0)} \\
& \leq & \sum_{j=0}^t (\alpha^j (\widetilde{\varphi}(\mu^j) + (z - \mu^j)^\top \nabla \widetilde{\varphi}(\mu^j))) + \|z-z^0\|_\infty^2 \\
& \overset{z^0=0}{=} & \sum_{j=0}^t (\alpha^j (\widetilde{\varphi}(\mu^j) + (z - \mu^j)^\top \nabla \widetilde{\varphi}(\mu^j))) + \|z\|_\infty^2.
\end{eqnarray*}
This completes the proof. 
\end{proof}
The final lemma provides us with a key lower bound for the accumulating parameter.
\begin{lemma}\label{lemma:AP-bound}
For each iteration $t$ of Algorithm~\ref{Algorithm:APDAMD}, we have $\bar{\alpha}^t \geq \frac{\eta(t+1)^2}{32\delta}$. 
\end{lemma}
\begin{proof}
For $t=1$, we have $\bar{\alpha}^1 = \alpha^1 = \frac{1}{\delta M^1} \geq \frac{\eta}{8\delta}$ since $M^1 \leq \frac{8}{\eta}$ was proven in Lemma~\ref{lemma:LS-bound}. Thus, the desired result holds true when $t=1$. Then we proceed to prove that it holds true for $t \geq 1$ using the induction. Indeed, we have 
\begin{eqnarray*}
\lefteqn{\bar{\alpha}^{t+1} \ = \ \bar{\alpha}^t + \alpha^{t+1} \ = \ \bar{\alpha}^t + \frac{1 + \sqrt{1 + 4\delta M^t\bar{\alpha}^t}}{2\delta M^t}} \\ 
& = & \bar{\alpha}^t + \frac{1}{2\delta M^t} + \sqrt{\frac{1}{4(\delta M^t)^2} + \frac{\bar{\alpha}^t}{\delta M^t}} \\
& \geq & \bar{\alpha}^t + \frac{1}{2\delta M^t} + \sqrt{\frac{\bar{\alpha}^t}{\delta M^t}} \\
& \geq & \bar{\alpha}^t + \frac{\eta}{16\delta} + \sqrt{\frac{\eta\bar{\alpha}^t}{8\delta}}, 
\end{eqnarray*}
where the last inequality comes from $M^t \leq \frac{8}{\eta}$ as shown in Lemma~\ref{lemma:LS-bound}. Suppose that the desired result holds true for $t = k_0$, we have
\begin{equation*}
\bar{\alpha}^{k_0+1} \geq \frac{\eta(k_0+1)^2}{32\delta} + \frac{\eta}{16\delta} + \sqrt{\frac{\eta^2(k_0+1)^2}{256\delta^2}} = \frac{\eta((k_0+1)^2 + 2 + 2(k_0+1))}{32\delta} \geq \frac{\eta(k_0+2)^2}{32\delta}. 
\end{equation*}
This completes the proof. 
\end{proof}

\subsection{Complexity analysis for APDAMD}\label{Sec:complex_APDAMD}
We are now ready to establish the complexity bound of APDAMD for solving the entropic regularized OT problem. Indeed, we recall that $\widetilde{\varphi}(\lambda)$ is defined with $\lambda = (\alpha, \beta)$ by 
\begin{equation*}
\widetilde{\varphi}(\alpha, \beta) = - \eta\log\left(\sum_{1 \leq i, j \leq n} e^{\eta^{-1}(\alpha_i + \beta_j - C_{ij})}\right) + \alpha^\top r + \beta^\top c. 
\end{equation*}
Since $(\alpha, \beta)$ can be obtain by $\alpha_i = \eta u_i$ and $\beta_j = \eta v_j$, we derive from Lemma~\ref{lemma:dual-bound-infinity} that 
\begin{equation*}
\|\alpha^\star\|_\infty \leq \eta R, \quad \|\beta^\star\|_\infty \leq \eta R. 
\end{equation*}
where $R$ is defined accordingly. Then, we proceed to the following key result determining an upper bound for the number of iterations for Algorithm~\ref{Algorithm:APDAMD} to reach a desired accuracy $\varepsilon'$:
\begin{theorem}\label{Theorem:APDAMD-Total-Complexity}
Letting $\{X^t\}_{t \geq 0}$ be the iterates generated by Algorithm~\ref{Algorithm:APDAMD}, the number of iterations required to satisfy $\|A_{\text{ot}}\textnormal{vec}(X^t) - b\|_1 \leq \varepsilon'$ is upper bounded by 
\begin{equation*}
t \leq 1 + \sqrt{\frac{128\delta R}{\varepsilon'}}, 
\end{equation*}
where $R > 0$ is defined in Lemma~\ref{lemma:dual-bound-infinity}.
\end{theorem}
\begin{proof}
From Lemma~\ref{lemma:ES-bound}, we have
\begin{equation*}
\bar{\alpha}^t\widetilde{\varphi}(\lambda^t) \leq \min_{z \in B_\infty(2\eta R)} \left\{\sum_{j=0}^t (\alpha^j (\widetilde{\varphi}(\mu^j) + (z - \mu^j)^\top \nabla\widetilde{\varphi}(\mu^j))) + \|z\|_\infty^2\right\}, 
\end{equation*}
where $B_\infty(r) \mydefn \{\lambda \in \br^n \mid \|\lambda\|_\infty \leq r\}$. This implies that
\begin{equation*}
\bar{\alpha}^t\widetilde{\varphi}(\lambda^t) \leq \min_{z \in B_\infty(2\eta R)} \left\{\sum_{j=0}^t (\alpha^j (\widetilde{\varphi}(\mu^j) + (z - \mu^j)^\top \nabla\widetilde{\varphi}(\mu^j)))\right\} + 4\eta^2 R^2. 
\end{equation*}
Since $\widetilde{\varphi}$ is the objective function of dual entropic regularized OT problem, we have
\begin{equation*}
\widetilde{\varphi}(\mu^j) + (z - \mu^j)^\top \nabla \widetilde{\varphi}(\mu^j) = - f(x(\mu^j)) + z^\top(b - A_{\text{ot}}x(\mu^j)). 
\end{equation*}
Therefore, we conclude that 
\begin{eqnarray*}
\bar{\alpha}^t \widetilde{\varphi}(\lambda^t) & \leq & \min_{z \in B_\infty(2 \eta R)} \left\{\sum_{j=0}^t (\alpha^j(\widetilde{\varphi}(\mu^j) + (z - \mu^j)^\top \nabla \widetilde{\varphi}(\mu^j)))\right\} + 4 \eta^2 R^2 \\
& \leq & 4\eta^2 R^2 -\bar{\alpha}^t f(x^t) + \min_{z \in B_\infty(2\eta R)} \{\bar{\alpha}^t z^\top(b - A_{\text{ot}}x^t)\} \\
& = & 4\eta^2 R^2 -\bar{\alpha}^t f(x^t) - 2 \bar{\alpha}^t \eta R\| A_{\text{ot}}x^t - b\|_1, 
\end{eqnarray*}
where the second inequality comes from the convexity of $f$ and the last equality comes from the fact that $\ell_1$-norm is the dual norm of $\ell_\infty$-norm. That is to say, 
\begin{equation*}
f(x^t) + \widetilde{\varphi}(\lambda^t) + 2\eta R\|A_{\text{ot}}x^t - b\|_1 \leq \frac{4\eta^2R^2}{\bar{\alpha}^t}. 
\end{equation*}
Suppose that $\lambda^\star$ is an optimal solution to dual entropic regularized OT problem satisfying $\|\lambda^\star\|_\infty \leq \eta R$, we have
\begin{eqnarray*}
f(x^t) + \widetilde{\varphi}(\lambda^t) & \geq & f(x^t) + \widetilde{\varphi}(\lambda^\star) \ = \ f(x^t) + b^\top \lambda^\star + \max_{x \in \br^{n^2}} \left\{- f(x) - (\lambda^\star)^\top A_{\text{ot}}x\right\} \\
& \geq & f(x^t) + b^\top \lambda^\star - f(x^t) - (\lambda^\star)^\top A_{\text{ot}}x^t \ = \ (b - A_{\text{ot}}x^t)\lambda^\star \\
& \geq & - \eta R\|A_{\text{ot}}x^t - b\|_1, 
\end{eqnarray*} 
Therefore, we conclude that 
\begin{equation*}
\|A_{\text{ot}}x^t - b\|_1 \leq \frac{4\eta R}{\bar{\alpha}^t} \leq \frac{128\delta R}{(t+1)^2}.  
\end{equation*}
This completes the proof. 
\end{proof}
Now, we are ready to present the complexity bound of Algorithm~\ref{Algorithm:ApproxOT_APDAMD} for approximating the OT problem.
\begin{theorem} \label{Theorem:ApproxOT-APDAMD-Total-Complexity}
The APDAMD scheme for approximating optimal transport (Algorithm~\ref{Algorithm:ApproxOT_APDAMD}) returns an $\varepsilon$-approximate transportation plan (cf. Definition~\ref{def:eps-approximation}) in
\begin{equation*}
\bigO\left(\frac{n^2 \sqrt{\delta} \|C\|_\infty\log(n)}{\varepsilon}\right)
\end{equation*}
arithmetic operations. 
\end{theorem}
\begin{proof}
Using the same argument as in Theorem~\ref{Theorem:ApproxOT-GK-Total-Complexity}, we have
\begin{equation*}
\langle C, \hat{X}\rangle - \langle C, X^\star\rangle \leq \frac{\varepsilon}{2} + 4(\|\tilde{X}\one_n - r\|_1 + \|\tilde{X}^\top\one_n - c\|_1)\|C\|_\infty,
\end{equation*}
where $\hat{X}$ is returned by Algorithm~\ref{Algorithm:ApproxOT_APDAMD}, $X^*$ is a solution to the OT problem and $\tilde{X} = \textsc{Apdamd}(\widetilde{\varphi}, A_{\text{ot}}, b, \tfrac{\varepsilon'}{2})$. Since $\|\tilde{X}\one_n - r\|_1 + \|\tilde{X}^\top\one_n - c\|_1 \leq \varepsilon'$ and $\varepsilon' = \frac{\varepsilon}{8\|C\|_\infty}$, we have $\langle C, \hat{X}\rangle - \langle C, X^\star\rangle \leq \frac{\varepsilon}{2} + \frac{\varepsilon}{2} = \varepsilon$. 

The remaining step is to analyze the complexity bound. If follows from Lemma~\ref{lemma:LS-bound} and Theorem~\ref{Theorem:APDAMD-Total-Complexity} that 
\begin{eqnarray*}
N_t & \leq & 4t + 4 + \frac{2\log(\frac{8}{\eta}) - 2\log(L^0)}{\log 2} \\
& \leq & 8 + \sqrt{\frac{2048\delta R}{\varepsilon'}} + \frac{2\log(\frac{8}{\eta}) - 2\log(L^0)}{\log 2} \\
& = & 8 + 256\sqrt{\frac{\delta R\|C\|_\infty\log(n)}{\varepsilon}} + \frac{2\log(\frac{32\log(n)}{\varepsilon}) - 2\log(L^0)}{\log 2}.
\end{eqnarray*}
Combining the definition of $R$ in Lemma~\ref{lemma:dual-bound-infinity} with the definition of $\eta$, $\tilde{r}$ and $\tilde{c}$ in Algorithm~\ref{Algorithm:ApproxOT_APDAMD}, we have 
\begin{equation*}
R \leq \frac{4\|C\|_\infty\log(n)}{\varepsilon} + \log(n) - 2\log\left(\frac{\varepsilon}{64n\|C\|_\infty}\right).
\end{equation*}
Therefore, we conclude that 
\begin{eqnarray*}
N_t & \leq & 256\sqrt{\frac{\delta\|C\|_\infty\log(n)}{\varepsilon}}\sqrt{\frac{4\|C\|_\infty\log(n)}{\varepsilon} + \log(n) - 2\log\left(\frac{\varepsilon}{64n\|C\|_\infty}\right)} \\ 
& & + \frac{2\log(\frac{32\log(n)}{\varepsilon}) - 2\log(L^0)}{\log 2} + 8 \ = \ \bigO\left(\frac{\sqrt{\delta}\left\|C\right\|_\infty\log(n)}{\varepsilon}\right). 
\end{eqnarray*}
The total iteration complexity in Step 3 of Algorithm~\ref{Algorithm:ApproxOT_APDAMD} is bounded by $\bigO(\sqrt{\delta}\|C\|_\infty\log(n)\varepsilon^{-1})$. Each iteration of Algorithm~\ref{Algorithm:APDAMD} requires $O(n^2)$ arithmetic operations. Therefore, the total number of arithmetic operations is $\bigO(n^2\sqrt{\delta}\|C\|_\infty\log(n)\varepsilon^{-1})$. Moreover, $\tilde{r}$ and $\tilde{c}$ in Step 1 of Algorithm~\ref{Algorithm:ApproxOT_APDAMD} can be found in $\bigO(n)$ arithmetic operations and \citet[Algorithm~2]{Altschuler-2017-Near} requires $\bigO(n^2)$ arithmetic operations. Therefore, we conclude that the total number of arithmetic operations is $\bigO(n^2\sqrt{\delta}\|C\|_\infty\log(n)\varepsilon^{-1})$.
\end{proof}
The complexity results in Theorem~\ref{Theorem:ApproxOT-APDAMD-Total-Complexity} suggests an interesting feature of the (regularized) OT problem. Indeed, the dependence of that bound on $\delta$ manifests the necessity of $\ell_\infty$-norm in the understanding of the complexity of the entropic regularized OT problem. This view is also in harmony with the proof technique of running time for Greenkhorn in Section~\ref{sec:greenkhorn}, where we rely on $\ell_\infty$-norm of optimal solutions of the dual entropic regularized OT problem to measure the progress in the objective value among the successive iterates. 

\subsection{Revisiting APDAGD}
We revisit APDAGD~\citep{Dvurechensky-2018-Computational} for the entropic regularized OT problem. First, we point out that the current complexity bound of $\bigOtil(\min\{n^{9/4}\varepsilon^{-1}, n^{2}\varepsilon^{-2}\})$ is not valid by a simple counterexample. Then, we establish a new complexity bound of APDAGD using our techniques in Section~\ref{Sec:complex_APDAMD}. Despite the issue with entropic regularized OT, we wish to emphasize that APDAGD is still an interesting and efficient accelerated algorithm for general linearly constrained convex optimization problem with solid theoretical guarantee. More precisely,~\citet[Theorem~3]{Dvurechensky-2018-Computational} is not applicable to entropic regularized OT since no dual solution exists with a constant bound in $\ell_2$-norm. However, it can be used for analyzing other problems with bounded optimal dual solution. 
\begin{algorithm}[!t]\small
\caption{Approximating OT by~\citet[Algorithm~3]{Dvurechensky-2018-Computational}}\label{Algorithm:ApproxOT_APDAGD}
\begin{algorithmic}
\STATE \textbf{Input:} $\eta = \frac{\varepsilon}{4\log(n)}$ and $\varepsilon' = \frac{\varepsilon}{8\|C\|_\infty}$. 
\STATE \textbf{Step 1:} Let $\tilde{r} \in \Delta_n$ and $\tilde{c} \in \Delta_n$ be defined by $(\tilde{r}, \tilde{c}) = (1 - \frac{\varepsilon'}{8})(r, c) + \frac{\varepsilon'}{8n}(\one_n, \one_n)$. 
\STATE \textbf{Step 2:} Let $A_{\text{ot}} \in \br^{2n \times n^2}$ and $b \in \br^{2n}$ be defined by $A_{\text{ot}}\text{vec}(X) = \begin{pmatrix} X\one_n \\ X^\top\one_n \end{pmatrix}$ and $b = \begin{pmatrix} \tilde{r} \\ \tilde{c} \end{pmatrix}$. 
\STATE \textbf{Step 3:} Compute $\tilde{X} = \textsc{Apdagd}(\widetilde{\varphi}, A_{\text{ot}}, b, \frac{\varepsilon'}{2})$ where $\widetilde{\varphi}$ is defined by Eq.~\eqref{prob:dualregOT-APDAMD}. 
\STATE \textbf{Step 4:} Round $\tilde{X}$ to $\hat{X}$ using~\citet[Algorithm~2]{Altschuler-2017-Near} such that $\hat{X}\one_n = r$ and $\hat{X}^\top\one_n = c$. 
\end{algorithmic}
\end{algorithm}

To facilitate the ensuing discussion, we first present the complexity bound for entropic regularized OT in \citet{Dvurechensky-2018-Computational} using our notation. Indeed, we recall that APDAGD is developed for solving the optimization problem with the objective function $\widehat{\varphi}$ defined as follows,   
\begin{equation}\label{prob:dualregOT-APDAGD}
\min_{\alpha, \beta \in \br^n} \ \widehat{\varphi}(\alpha, \beta) \mydefn \eta\left(\sum_{i,j=1}^n e^{-\frac{C_{ij} - \alpha_i - \beta_j}{\eta} -1}\right) - \alpha^\top r - \beta^\top c. 
\end{equation}
\begin{theorem}[Theorem 4 in \citet{Dvurechensky-2018-Computational}]
Let $\overline{R}>0$ be defined such that there exists an optimal solution to the dual entropic regularized OT problem in Eq.~\eqref{prob:dualregOT-APDAMD}, denoted by $(u^\star, v^\star)$, satisfying $\|(u^\star, v^\star)\| \leq \overline{R}$, the APDAGD scheme for approximating optimal transport (cf. Algorithm~\ref{Algorithm:ApproxOT_APDAGD}) returns an $\varepsilon$-approximate transportation plan (cf. Definition~\ref{def:eps-approximation}) in
\begin{equation*}
\bigO\left(\min\left\{\frac{n^{9/4}\sqrt{\overline{R} \|C\|_\infty\log(n)}}{\varepsilon}, \frac{n^{2}\overline{R}\|C\|_\infty\log(n)}{\varepsilon^2}\right\} \right),
\end{equation*}
arithmetic operations. 
\end{theorem}
From the above theorem,~\citet{Dvurechensky-2018-Computational} claims that the complexity bound for APDAGD is $\bigOtil(\min\{n^{9/4}\varepsilon^{-1}, n^2\varepsilon^{-2}\})$. However, there are two issues: 
\begin{enumerate}
\item The upper bound $\overline{R}$ is assumed to be independent of $n$, which is not correct; see our counterexample in Proposition~\ref{prop:tight_upper_bound}. 
\item The known upper bound $\overline{R}$ for the optimal solution depends on $\min_{1 \leq i, j \leq n} \{r_i, c_j\}$ (cf.~\citet[Lemma~1]{Dvurechensky-2018-Computational} or Lemma~\ref{lemma:dual-bound-infinity} in our paper). This implies that the valid algorithm needs to take the rounding error with $r$ and $c$ into account. 
\end{enumerate}
\paragraph{Corrected upper bound $\overline{R}$.} Corollary~\ref{corollary:dual-bound-l2} and Lemma~\ref{lemma:GK-bound} imply that a straightforward upper bound for $\overline{R}$ is $\bigOtil(\sqrt{n})$. Given a tolerance $\varepsilon \in (0, 1)$, we further show that $\overline{R}$ is indeed $\Omega(\sqrt{n})$ by using a specific entropic regularized OT problem as follows. 
\begin{proposition}\label{prop:tight_upper_bound}
Suppose that $C = \one_n\one_n^\top$ and $r = c = \frac{1}{n}\one_n$. Given a tolerance $\varepsilon \in (0, 1)$ and the regularization term $\eta = \frac{\varepsilon}{4\log(n)}$, all the optimal solutions of the dual entropic regularized OT problem in Eq.~\eqref{prob:dualregOT-APDAGD} satisfy that $\|(\alpha^*, \beta^*)\| \gtrsim \sqrt{n}$.
\end{proposition}
\begin{proof}
By the definition $r$, $c$ and $\eta$, we rewrite the dual function $\widehat{\varphi}(\alpha, \beta)$ as follows:
\begin{equation*}
\widehat{\varphi}(\alpha, \beta) = \frac{\varepsilon}{4e\log(n)}\sum_{1 \leq i, j \leq n} e^{- \frac{4\log(n)(1 - \alpha_i - \beta_j)}{\varepsilon}} - \frac{\alpha^\top \one_n}{n} - \frac{\beta^\top \one_n}{n}.
\end{equation*}
Since $(\alpha^\star, \beta^\star)$ is an optimal solution of dual entropic regularized OT problem, we have
\begin{equation}\label{inequality:tight_upper_bound}
e^{\frac{4\log(n)\alpha_i^\star}{\varepsilon}} \sum_{j=1}^n e^{-\frac{4\log(n)(1 - \beta_j^\star)}{\varepsilon}} = e^{\frac{4\log(n)\beta_i^\star}{\varepsilon}} \sum_{j=1}^n e^{- \frac{4\log(n)(1 - \alpha_j^\star)}{\varepsilon}} = \frac{e}{n} \quad \text{for all } i \in [n]. 
\end{equation}
This implies $\alpha_i^\star = \alpha_j^\star$ and $\beta_i^\star = \beta_j^\star$ for all $i, j \in [n]$, and $\alpha_i^\star + \beta_i^\star$ are the same for all $i \in [n]$.  Without loss of generality, we can let $\alpha_i^\star = 0$ in Eq.~\eqref{inequality:tight_upper_bound} and obtain that 
\begin{equation*}
\beta_i^\star = 1 + \frac{\varepsilon}{4\log(n)}\left(1 - 2\log(n)\right) = 1 + \frac{\varepsilon}{4\log(n)} - \frac{\varepsilon}{2}. 
\end{equation*}
which implies that $\alpha_i^\star + \beta_i^\star = 1 + \frac{\varepsilon}{4\log(n)} - \frac{\varepsilon}{2} \geq \frac{1}{2}$ for all $i \in [n]$. Thus, we have
\begin{align*}
\|(\alpha^\star, \beta^\star)\| \geq \sqrt{\frac{\sum_{i = 1}^{n} (\alpha_i^\star + \beta_i^\star)^2}{2}} = \frac{1}{2}\sqrt{\frac{n}{2}} \gtrsim \sqrt{n}. 
\end{align*}
As a consequence, we achieve the conclusion of the proposition.
\end{proof}
\paragraph{Approximation algorithm for OT by APDAGD.} It is worth noting that the rounding procedure is missing in~\citet[Algorithm~4]{Dvurechensky-2018-Computational} and we improve it to Algorithm~\ref{Algorithm:ApproxOT_APDAGD}. In particular,~\citet[Algorithm~3]{Dvurechensky-2018-Computational} is used in Step 3 of Algorithm~\ref{Algorithm:ApproxOT_APDAGD} for another function $\widetilde{\varphi}$ defined in Eq.~\eqref{def:dual-regOT-objective}. Given the corrected upper bound $\overline{R}$ and Algorithm~\ref{Algorithm:ApproxOT_APDAGD} for approximating OT, 
we provide a new complexity bound of Algorithm~\ref{Algorithm:ApproxOT_APDAGD} in the following proposition. 
\begin{proposition}\label{prop:ApproxOT-APDAGD-Total-Complexity}
The APDAGD scheme for approximating optimal transport (Algorithm~\ref{Algorithm:ApproxOT_APDAGD}) returns an $\varepsilon$-approximate transportation plan (cf. Definition~\ref{def:eps-approximation}) in
\begin{equation*}
\bigO\left(\frac{n^{5/2}\|C\|_\infty\sqrt{\log(n)}}{\varepsilon}\right)
\end{equation*}
arithmetic operations.
\end{proposition}
\begin{proof}
The proof is a simple modification of the proof for~\citet[Theorem 4]{Dvurechensky-2018-Computational} and we only give a proof sketch here. In particular, we can obtain that the number of iterations for Algorithm~\ref{Algorithm:ApproxOT_APDAGD} required to reach the tolerance $\varepsilon$ is
\begin{equation}\label{inequality-proposition-APDAGD}
t \leq \bigO\left(\max \left\{\min \left\{\frac{n^{1/4}\sqrt{\overline{R} \|C\|_\infty\log(n)}}{\varepsilon}, 
\frac{\overline{R}\|C\|_\infty\log(n)}{\varepsilon^2}\right\}, \frac{\overline{R} \sqrt{\log n}}{\varepsilon}\right\}\right).
\end{equation}
Moreover, we have $\overline{R} \leq \sqrt{n}\eta R$ where $R = \eta^{-1}\|C\|_\infty + \log(n) - 2\log(\min_{1 \leq i, j \leq n} \{r_i, c_j\})$. Therefore, the total iteration complexity in Step 3 of Algorithm~\ref{Algorithm:ApproxOT_APDAGD} is $\bigO(\sqrt{n\log(n)}\|C\|_\infty\varepsilon^{-1})$. Each iteration of APDAGD requires $\bigO(n^2)$ arithmetic operations. Therefore, the total number of arithmetic operations is $\bigO(n^{5/2}\|C\|_\infty\sqrt{\log(n)}\varepsilon^{-1})$. Note that $\tilde{r}$ and $\tilde{c}$ in Step 1 of Algorithm~\ref{Algorithm:ApproxOT_APDAGD} can be found in $\bigO(n)$ arithmetic operations and~\citet[Algorithm~2]{Altschuler-2017-Near} requires $\bigO(n^2)$ arithmetic operations. Therefore, we conclude that the total number of arithmetic operations is $\bigO(n^{5/2}\|C\|_\infty\sqrt{\log(n)}\varepsilon^{-1})$.
\end{proof}
\begin{remark}
As indicated in Proposition~\ref{prop:ApproxOT-APDAGD-Total-Complexity}, the corrected 
complexity bound of APDAGD for the entropic regularized OT is similar to that of APDAMD when we choose $\phi(\cdot) = \frac{1}{2n}\|\cdot\|^2$ and have $\delta = n$. From this perspective, our algorithm can be viewed as a generalization of APDAGD. Since our algorithm utilizes $\ell_\infty$-norm in the line search criterion, it is more robust than APDAGD in practice; see Section~\ref{sec:experiments} for the details. 
\end{remark}

%% file: sec/acceleration.tex
\section{Accelerating Sinkhorn}\label{sec:acceleration}
In this section, we present an accelerated variant of Sinkhorn for solving the entropic regularized OT problem in Eq.~\eqref{prob:OT_regularized}. Combined with a rounding scheme, our algorithm can be used for solving the OT problem in Eq.~\eqref{prob:OT} and achieves a complexity bound of $\widetilde{O}(n^{7/3}\varepsilon^{-4/3})$, which improves that of Sinkhorn in terms of $1/\varepsilon$ and APDAGD and AAM~\citep{Guminov-2021-Combination} in terms of $n$. The idea comes from a novel combination of Nesterov's estimated sequence and the techniques for analyzing Sinkhorn. 
\begin{algorithm}[!t]\small
\caption{\textsc{Accelerated Sinkhorn}$(C, \eta, r, c, \varepsilon')$} \label{Algorithm:acceleration}
\begin{algorithmic}
\STATE \textbf{Input:} $t = 0$, $\theta_0 = 1$ and $\check{u}^0 = \tilde{u}^0 = \check{v}^0 = \tilde{v}^0 = \zero_n$.  
\WHILE{$E_t > \varepsilon'$}
\STATE Compute $\begin{pmatrix} \bar{u}^t \\ \bar{v}^t \end{pmatrix} = (1 - \theta_t)\begin{pmatrix} \check{u}^t \\ \check{v}^t \end{pmatrix} + \theta_t \begin{pmatrix} \tilde{u}^t \\ \tilde{v}^t \end{pmatrix}$. 
\STATE Compute $\tilde{u}^{t+1}$ and $\tilde{v}^{t+1}$ by 
\begin{equation*}
\tilde{u}^{t+1} = \tilde{u}^t - \frac{1}{2\theta_t}\left(\frac{r(B(\bar{u}^t, \bar{v}^t))}{\|B(\bar{u}^t, \bar{v}^t)\|_1} - r\right), \qquad  \tilde{v}^{t+1} = \tilde{v}^t - \frac{1}{2\theta_t}\left(\frac{c(B(\bar{u}^t, \bar{v}^t))}{\|B(\bar{u}^t, \bar{v}^t)\|_1}  - c\right). 
\end{equation*}
\STATE Compute $\begin{pmatrix} \grave{u}^t \\ \grave{v}^t \end{pmatrix} = \begin{pmatrix} \bar{u}^t \\ \bar{v}^t \end{pmatrix} + \theta_t\left(\begin{pmatrix} \tilde{u}^{t+1} \\ \tilde{v}^{t+1} \end{pmatrix} - \begin{pmatrix} \tilde{u}^t \\ \tilde{v}^t \end{pmatrix}\right)$.
\IF{$t$ is even}
\STATE $\widehat{u}^t = \grave{u}^t + \log(r) - \log(r(B(\grave{u}^t, \grave{v}^t)))$ and $\widehat{v}^t = \grave{v}^t$. 
\ELSE
\STATE $\widehat{u}^t = \grave{u}^t$ and $\widehat{v}^t = \grave{v}^t + \log(c) - \log(c(B(\grave{u}^t, \grave{v}^t)))$. 
\ENDIF
\STATE Compute $\begin{pmatrix} u^t \\ v^t \end{pmatrix} = \argmin\left\{\varphi(u, v) \Big\vert \begin{pmatrix} u \\ v \end{pmatrix} \in \left\{\begin{pmatrix} \check{u}^t \\ \check{v}^t \end{pmatrix}, \begin{pmatrix} \widehat{u}^t \\ \widehat{v}^t \end{pmatrix}\right\}\right\}$.
\IF{$t$ is even}
\STATE $\check{u}^{t+1} = u^t + \log(r) - \log(r(B(u^t, v^t)))$ and $\check{v}^{t+1} = v^t$. 
\ELSE
\STATE $\check{u}^{t+1} = u^t$ and $\check{v}^{t+1} = v^t + \log(c) - \log(c(B(u^t, v^t)))$. 
\ENDIF  
\STATE Compute $\theta_{t+1} = \frac{\theta_t(\sqrt{\theta_t^2 + 4} - \theta_t)}{2}$. 
\STATE Set $t = t + 1$. 
\ENDWHILE
\STATE \textbf{Output:} $B(u^t, v^t)$.  
\end{algorithmic}
\end{algorithm}

\subsection{Algorithmic procedure} 
We present the pseudocode of accelerated Sinkhorn in Algorithm~\ref{Algorithm:acceleration}. This algorithm achieves the acceleration by using Nesterov's estimate sequences~\citep{Nesterov-2018-Lectures}. While our algorithm can be interpreted as an accelerated block coordinate descent algorithm, it is worth mentioning that our algorithm is purely \textit{deterministic} and thus differs from other accelerated randomized algorithms~\citep{Lin-2015-Accelerated, Fercoq-2015-Accelerated, Lu-2018-Accelerating} in the optimization literature. 

Algorithm~\ref{Algorithm:acceleration} is a novel combination of Nesterov's estimate sequences, a monotone search step, the choice of greedy coordinate and two coordinate updates. It is applied to solve the dual entropic regularized OT problem in Eq.~\eqref{prob:OT_regularized_dual}: 
\begin{equation*}
\min \limits_{u, v} \ \varphi(u, v) \mydefn \log(\|B(u, v)\|_1) - u^\top r - v^\top c.  
\end{equation*}
More specifically, Nesterov's estimate sequences are responsible for optimizing a dual objective function $\varphi$ in a fast rate. The coordinate update guarantees that $\varphi(\widehat{u}^t, \widehat{v}^t) \leq \varphi(\grave{u}^t, \grave{v}^t)$ and $\|B(\widehat{u}^t, \widehat{v}^t)\|_1=1$. The monotone search step guarantees that $\varphi(u^t, v^t) \leq \varphi(\widehat{u}^t, \widehat{v}^t)$. The greedy coordinate update guarantees that $\varphi(\check{u}^{t+1}, \check{v}^{t+1}) \leq \varphi(u^t, v^t)$ with sufficient progress. 

Furthermore, we also use the same quantity as that in Greekhorn to measure the per-iteration residue of Algorithm~\ref{Algorithm:acceleration}: 
\begin{equation}\label{Def:residue-acceleration}
E_t = \|r(B(u^t, v^t)) - r\|_1 + \|c(B(u^t, v^t)) - c\|_1.
\end{equation}
The computationally expensive step is to compute $\frac{r(B(\bar{u}^t, \bar{v}^t))}{\|B(\bar{u}^t, \bar{v}^t)\|_1}$ and $\frac{c(B(\bar{u}^t, \bar{v}^t))}{\|B(\bar{u}^t, \bar{v}^t)\|_1}$. Since $B(\bar{u}^t, \bar{v}^t)$ does not have any special property, it is difficult to design some implementation trick to reduce the order of $n$. As such, the arithmetic operations for each iteration is $\bigO(n^2)$ and is exactly the same as Sinkhorn~\citep{Cuturi-2013-Sinkhorn}, APDAGD~\citep{Dvurechensky-2018-Computational} and AAM~\citep{Guminov-2021-Combination}. Combining Algorithm~\ref{Algorithm:acceleration} and~\citet[Algorithm~2]{Altschuler-2017-Near}, we are ready to present the pseudocode of our main algorithm in Algorithm~\ref{Algorithm:ApproxOT_Acceleration}. The regularization parameter $\eta$ is set as before, and Step 1 is necessary since accelerated Sinkhorn is not well behaved if the marginal distributions have sparse support.
\begin{algorithm}[!t]\small
\caption{Approximating OT by Algorithm~\ref{Algorithm:acceleration}} \label{Algorithm:ApproxOT_Acceleration}
\begin{algorithmic}
\STATE \textbf{Input:} $\eta = \frac{\varepsilon}{4\log(n)}$ and $\varepsilon'=\frac{\varepsilon}{8\|C\|_\infty}$. 
\STATE \STATE \textbf{Step 1:} Let $\tilde{r} \in \Delta_n$ and $\tilde{c} \in \Delta_n$ be defined by $(\tilde{r}, \tilde{c}) = (1 - \frac{\varepsilon'}{8})(r, c) + \frac{\varepsilon'}{8n}(\one_n, \one_n)$. 
\STATE \textbf{Step 2:} Compute $\tilde{X} = \textsc{Accelerated Sinkhorn}(C, \eta, \tilde{r}, \tilde{c}, \frac{\varepsilon'}{2})$.
\STATE \textbf{Step 3:} Round $\tilde{X}$ to $\hat{X}$ using~\citet[Algorithm~2]{Altschuler-2017-Near} such that $\hat{X}\one_n = r$ and $\hat{X}^\top\one_n = c$. 
\STATE \textbf{Output:} $\widehat{X}$. 
\end{algorithmic}
\end{algorithm} 

\subsection{Technical lemmas}
We first present two technical lemmas which are essential in the analysis of Algorithm~\ref{Algorithm:acceleration}. The first lemma provides an inductive relationship on the quantity 
\begin{equation}\label{def:residue_acceleration}
\delta_t = \varphi(\check{u}^t, \check{v}^t) - \varphi(u^\star, v^\star), 
\end{equation}
where $(u^\star, v^\star)$ is an optimal solution of the dual entropic regularized OT problem in Eq.~\eqref{prob:OT_regularized_dual} that satisfies Lemma~\ref{corollary:dual-bound-l2}. To facilitate the discussion, we recall Eq.~\eqref{inequality-gradient-objective} with $\|A\|_{1 \rightarrow 2} = \sqrt{2}$ as follows, 
\begin{equation}\label{smooth:dualregOT}
\varphi(u', v') - \varphi(u, v) -  \begin{pmatrix} u' - u \\ v' - v\end{pmatrix}^\top\nabla\varphi(u, v) \leq \left\|\begin{pmatrix} u' - u \\ v' - v\end{pmatrix}\right\|^2,
\end{equation}
which will be used in the proof of the first lemma. 
\begin{lemma}\label{lemma:acceleration-descent}
Let $\{(\check{u}^t, \check{v}^t)\}_{t \geq 0}$ be the iterates generated by Algorithm~\ref{Algorithm:acceleration} and $(u^\star, v^\star)$ be an optimal solution of the dual entropic regularized OT problem. Then, we have
\begin{equation*}
\delta_{t+1} \leq (1-\theta_t)\delta_t + \theta_t^2\left(\left\|\begin{pmatrix} u^\star - \tilde{u}^t \\ v^\star - \tilde{v}^t \end{pmatrix}\right\|^2 - \left\|\begin{pmatrix} u^\star - \tilde{u}^{t+1} \\ v^\star - \tilde{v}^{t+1} \end{pmatrix}\right\|^2\right).
\end{equation*}
\end{lemma}
\begin{proof}
Using Eq.~\eqref{smooth:dualregOT} with $(u', v') = (\grave{u}^t, \grave{v}^t)$ and $(u, v) = (\bar{u}^t, \bar{v}^t)$, we have
\begin{equation*}
\varphi(\grave{u}^t, \grave{v}^t) \leq \varphi(\bar{u}^t, \bar{v}^t) + \theta_t\begin{pmatrix} \tilde{u}^{t+1} - \tilde{u}^t \\ \tilde{v}^{t+1} - \tilde{v}^t \end{pmatrix}^\top\nabla\varphi(\bar{u}^t, \bar{v}^t) + \theta_t^2\left\|\begin{pmatrix} \tilde{u}^{t+1} - \tilde{u}^t \\ \tilde{v}^{t+1} - \tilde{v}^t \end{pmatrix}\right\|^2. 
\end{equation*}
After simple calculations, we find that
\begin{eqnarray*}
\varphi(\bar{u}^t, \bar{v}^t) & = & (1-\theta_t)\varphi(\bar{u}^t, \bar{v}^t) + \theta_t \varphi(\bar{u}^t, \bar{v}^t), \\
\begin{pmatrix} \tilde{u}^{t+1} - \tilde{u}^t \\ \tilde{v}^{t+1} - \tilde{v}^t \end{pmatrix}^\top \nabla\varphi(\bar{u}^t, \bar{v}^t) & = & - \begin{pmatrix} \tilde{u}^t - \bar{u}^t \\ \tilde{v}^t - \bar{v}^t \end{pmatrix}^\top \nabla\varphi(\bar{u}^t, \bar{v}^t) + \begin{pmatrix} \tilde{u}^{t+1} - \bar{u}^t \\ \tilde{v}^{t+1} - \bar{v}^t \end{pmatrix}^\top \nabla\varphi(\bar{u}^t, \bar{v}^t). 
\end{eqnarray*}
Putting these pieces together yields that 
\begin{eqnarray}\label{claim-acceleration-descent-main}
\varphi(\grave{u}^t, \grave{v}^t) & \leq & \theta_t\left(\underbrace{\varphi(\bar{u}^t, \bar{v}^t) + \begin{pmatrix} \tilde{u}^{t+1} - \bar{u}^t \\ \tilde{v}^{t+1} - \bar{v}^t \end{pmatrix}^\top \nabla\varphi(\bar{u}^t, \bar{v}^t) + \theta_t\left\| \begin{pmatrix} \tilde{u}^{t+1} - \tilde{u}^t \\ \tilde{v}^{t+1} - \tilde{v}^t \end{pmatrix}\right\|^2}_{\textnormal{\bf I}}\right) \\
& & + \underbrace{(1-\theta_t)\varphi(\bar{u}^t, \bar{v}^t) - \theta_t\begin{pmatrix} \tilde{u}^t - \bar{u}^t \\ \tilde{v}^t - \bar{v}^t \end{pmatrix}^\top \nabla\varphi(\bar{u}^t, \bar{v}^t)}_{\textnormal{\bf II}}. \nonumber
\end{eqnarray}
We first bound the term $\textnormal{\bf I}$. Indeed, by the update formula for $(\tilde{u}^{t+1}, \tilde{v}^{t+1})$ and the definition of $\nabla \varphi$, we have
\begin{equation*}
\begin{pmatrix} u - \tilde{u}^{t+1} \\ v - \tilde{v}^{t+1} \end{pmatrix}^\top \left(\nabla\varphi(\bar{u}^t, \bar{v}^t) + 2\theta_t\begin{pmatrix} \tilde{u}^{t+1} - \tilde{u}^t \\ \tilde{v}^{t+1} - \tilde{v}^t \end{pmatrix}\right) = 0 \textnormal{ for all } (u, v) \in \br^{2n}. 
\end{equation*}
Letting $(u, v) = (u^\star, v^\star)$ and rearranging the resulting equation yields that 
\begin{eqnarray*}
\lefteqn{\begin{pmatrix} \tilde{u}^{t+1} - \bar{u}^t \\ \tilde{v}^{t+1} - \bar{v}^t \end{pmatrix}^\top \nabla\varphi(\bar{u}^t, \bar{v}^t) \ = \ \begin{pmatrix} u^\star - \bar{u}^t \\ v^\star - \bar{v}^t \end{pmatrix}^\top\nabla \varphi(\bar{u}^t, \bar{v}^t)} \\ 
& & + \theta_t\left(\left\| \begin{pmatrix} u^\star - \tilde{u}^t \\ v^\star - \tilde{v}^t \end{pmatrix} \right\|^2 - \left\| \begin{pmatrix} u^\star - \tilde{u}^{t+1} \\ v^\star - \tilde{v}^{t+1} \end{pmatrix} \right\|^2 - \left\| \begin{pmatrix} \tilde{u}^{t+1} - \tilde{u}^t \\ \tilde{v}^{t+1} - \tilde{v}^t \end{pmatrix}\right\|^2 \right).  
\end{eqnarray*}
Using the convexity of $\varphi$, we have 
\begin{equation*}
\begin{pmatrix} u^\star - \bar{u}^t \\ v^\star - \bar{v}^t \end{pmatrix}^\top\nabla \varphi(\bar{u}^t, \bar{v}^t) \leq \varphi(u^\star, v^\star) - \varphi(\bar{u}^t, \bar{v}^t). 
\end{equation*}
Putting these pieces together yields that 
\begin{equation}\label{claim-acceleration-descent-first}
\textnormal{\bf I} \leq \varphi(u^\star, v^\star) + \theta_t\left(\left\| \begin{pmatrix} u^\star - \tilde{u}^t \\ v^\star - \tilde{v}^t \end{pmatrix} \right\|^2 - \left\| \begin{pmatrix} u^\star - \tilde{u}^{t+1} \\ v^\star - \tilde{v}^{t+1} \end{pmatrix} \right\|^2\right). 
\end{equation}
We then bound the term $\textnormal{\bf II}$. Indeed, we see from the definition of $(\bar{u}^t, \bar{v}^t)$ that
\begin{equation*}
- \theta_t\begin{pmatrix} \tilde{u}^t - \bar{u}^t \\ \tilde{v}^t - \bar{v}^t \end{pmatrix} = \theta_t \begin{pmatrix} \bar{u}^t \\ \bar{v}^t \end{pmatrix} + (1 - \theta_t)\begin{pmatrix} \check{u}^t \\ \check{v}^t \end{pmatrix} - \begin{pmatrix} \bar{u}^t \\ \bar{v}^t \end{pmatrix} = (1-\theta_t)\begin{pmatrix} \check{u}^t - \bar{u}^t \\ \check{v}^t - \bar{v}^t \end{pmatrix}. 
\end{equation*}
Combining the above equation with the convexity of $\varphi$, we have
\begin{equation}\label{claim-acceleration-descent-second}
\textnormal{\bf II} = (1-\theta_t)\left(\varphi(\bar{u}^t, \bar{v}^t) + \begin{pmatrix} \check{u}^t - \bar{u}^t \\ \check{v}^t - \bar{v}^t \end{pmatrix}^\top\nabla \varphi(\bar{u}^t, \bar{v}^t)\right) \leq (1-\theta_t)\varphi(\check{u}^t, \check{v}^t). 
\end{equation}
Plugging Eq.~\eqref{claim-acceleration-descent-first} and Eq.~\eqref{claim-acceleration-descent-second} into Eq.~\eqref{claim-acceleration-descent-main} yields that 
\begin{equation*}
\varphi(\grave{u}^t, \grave{v}^t) \leq (1-\theta_t)\varphi(\check{u}^t, \check{v}^t) + \theta_t\varphi(u^\star, v^\star) + \theta_t^2\left(\left\|\begin{pmatrix} u^\star - \tilde{u}^t \\ v^\star - \tilde{v}^t \end{pmatrix}\right\|^2 - \left\|\begin{pmatrix} u^\star - \tilde{u}^{t+1} \\ v^\star - \tilde{v}^{t+1} \end{pmatrix}\right\|^2\right). 
\end{equation*}
Since $(\check{u}^{t+1}, \check{v}^{t+1})$ is obtained by a coordinate update from $(u^t, v^t)$, we have $\varphi(u^t, v^t) \geq \varphi(\check{u}^{t+1}, \check{v}^{t+1})$. By the definition of $(u^t, v^t)$, we have $\varphi(\widehat{u}^t, \widehat{v}^t) \geq \varphi(u^t, v^t)$. Since $(\widehat{u}^t, \widehat{v}^t)$ is obtained by a coordinate update from $(\grave{u}^t, \grave{v}^t)$, we have $\varphi(\grave{u}^t, \grave{v}^t) \geq \varphi(\widehat{u}^t, \widehat{v}^t)$. Collecting all of these results leads to
\begin{align*}
\lefteqn{\varphi(\check{u}^{t+1}, \check{v}^{t+1}) - \varphi(u^\star, v^\star) \ \leq \ (1-\theta_t)(\varphi(\check{u}^t, \check{v}^t) - \varphi(u^\star, v^\star))} & \\
& & \hspace{ 12 em} + \theta_t^2\left(\left\|\begin{pmatrix} u^\star - \tilde{u}^t \\ v^\star - \tilde{v}^t \end{pmatrix}\right\|^2 - \left\|\begin{pmatrix} u^\star - \tilde{u}^{t+1} \\ v^\star - \tilde{v}^{t+1} \end{pmatrix}\right\|^2\right).
\end{align*}
This completes the proof. 
\end{proof}
The second lemma provides an upper bound for $\delta_t$ defined by Eq.~\eqref{def:residue_acceleration} where $\{(\check{u}^t, \check{v}^t)\}_{t \geq 0}$ are generated by Algorithm~\ref{Algorithm:acceleration} and $(u^\star, v^\star)$ is an optimal solution defined by Corollary~\ref{corollary:dual-bound-l2}.  
\begin{lemma}\label{lemma:acceleration-bound}
Let $\{(\check{u}^t, \check{v}^t)\}_{t \geq 0}$ be the iterates generated by Algorithm~\ref{Algorithm:acceleration} and $(u^\star, v^\star)$ be an optimal solution of the dual entropic regularized OT problem satisfying that $\|(u^\star, v^\star)\| \leq \sqrt{2n}R$ where $R$ is defined in Corollary~\ref{corollary:dual-bound-l2}. Then, we have
\begin{equation*}
\delta_t \leq \frac{8nR^2}{(t+1)^2}.  
\end{equation*}
\end{lemma}
\begin{proof}
By simple calculations, we derive from the definition of $\theta_t$ that $\frac{\theta_{t+1}}{\theta_t} = \sqrt{1 - \theta_{t+1}}$. Therefore, we conclude from Lemma~\ref{lemma:acceleration-descent} that 
\begin{equation*}
\left(\frac{1-\theta_{t+1}}{\theta_{t+1}^2}\right)\delta_{t+1} - \left(\frac{1-\theta_t}{\theta_t^2}\right)\delta_t \leq \left\|\begin{pmatrix} u^\star - \tilde{u}^t \\ v^\star - \tilde{v}^t \end{pmatrix}\right\|^2 - \left\|\begin{pmatrix} u^\star - \tilde{u}^{t+1} \\ v^\star - \tilde{v}^{t+1} \end{pmatrix}\right\|^2.
\end{equation*}
Equivalently, we have
\begin{equation*}
\left(\frac{1-\theta_t}{\theta_t^2}\right)\delta_t + \left\|\begin{pmatrix} u^\star - \tilde{u}^t \\ v^\star - \tilde{v}^t \end{pmatrix}\right\|^2 \leq \left(\frac{1-\theta_0}{\theta_0^2}\right)\delta_0 + \left\|\begin{pmatrix} u^\star - \tilde{u}^0 \\ v^\star - \tilde{v}^0 \end{pmatrix}\right\|^2. 
\end{equation*}
Since $\theta_0 = 1$ and $\tilde{u}^0 = \tilde{v}^0 = \zero_n$, we have $\delta_t \leq \theta_{t-1}^2\|(u^\star, v^\star)\|^2 \leq 2nR^2\theta_{t-1}^2$. 

The remaining step is to show that $0 < \theta_t \leq \frac{2}{t+2}$. Indeed, the claim holds when $t=0$ as we have $\theta_0 = 1$. Assume that the claim holds for $t \leq t_0$, i.e., $\theta_{t_0} \leq \frac{2}{t_0+2}$, we have
\begin{equation*}
\theta_{t_0+1} = \frac{2}{1 + \sqrt{1 + \frac{4}{\theta_{t_0}^2}}} \leq \frac{2}{t_0+3}. 
\end{equation*}
Putting these pieces together yields the desired inequality for $\delta_t$. 
\end{proof}

\subsection{Main results}
We present an upper bound for the number of iterations required by Algorithm~\ref{Algorithm:acceleration}. Note that the per-iteration progress of Algorithm~\ref{Algorithm:acceleration} is measured by the function $\rho: \br_+^n \times \br_+^n \rightarrow \br_+$ given by: $\rho(a, b) : = \one_n^\top(b - a) + \sum_{i=1}^n a_i\log(\frac{a_i}{b_i})$. 
\begin{theorem}\label{Theorem:acceleration-iteration}
Let $\{(u^t, v^t)\}_{t \geq 0}$ be the iterates generated by Algorithm~\ref{Algorithm:acceleration}. The number of iterations required to reach the stopping criterion $E_t \leq \varepsilon'$ satisfies
\begin{equation*}
t \leq 1 + \left(\frac{16\sqrt{n}R}{\varepsilon'}\right)^{2/3}, 
\end{equation*}
where $R > 0$ is defined in Lemma~\ref{lemma:dual-bound-infinity}.
\end{theorem}
\begin{proof}
We first claim that 
\begin{equation}\label{claim-acceleration-iteration-main}
\varphi(u^t, v^t) - \varphi(\check{u}^{t+1}, \check{v}^{t+1}) \geq \frac{1}{2}\left(\|r(B(u^t, v^t)) - r\|_1^2 + \|c(B(u^t, v^t)) - c\|_1^2\right).
\end{equation}
By the definition of $\varphi$, we have
\begin{eqnarray}\label{inequality-acceleration-iteration-first}
\lefteqn{\varphi(u^t, v^t) - \varphi(\check{u}^{t+1}, \check{v}^{t+1}) \ = \ \log(\|B(u^t, v^t)\|_1)} \\
& & - \log(\|B(\check{u}^{t+1}, \check{v}^{t+1})\|_1) - (u^t - \check{u}^{t+1})^\top r - (v^t - \check{v}^{t+1})^\top c. \nonumber
\end{eqnarray}
From the update formula for $(\widehat{u}^t, \widehat{v}^t)$ and $(\check{u}^{t+1}, \check{v}^{t+1})$, it is clear that $\|B(\widehat{u}^t, \widehat{v}^t)\|_1 = 1$ and $\|B(\check{u}^{t+1}, \check{v}^{t+1})\|_1 = 1$ for all $t \geq 0$. Then, we derive from the update formula for $(u^t, v^t)$  that $\|B(u^t, v^t)\|_1 = 1$ for all $t \geq 1$. Therefore, we have
\begin{align*}
\varphi(u^t, v^t) - \varphi(\check{u}^{t+1}, \check{v}^{t+1}) & =  -(u^t - \check{u}^{t+1})^\top r - (v^t - \check{v}^{t+1})^\top c  \\
& =  (\log(r) - \log(r(B(u^t, v^t))))^\top r + (\log(c) - \log(c(B(u^t, v^t))))^\top c.  
\end{align*}
Since $\one_n^\top r = \one_n^\top r(B(u^t, v^t)) = \one_n^\top c = \one_n^\top c(B(u^t, v^t)) = 1$, we have 
\begin{equation*}
\varphi(u^t, v^t) - \varphi(\check{u}^{t+1}, \check{v}^{t+1}) = \rho(r, r(B(u^t, v^t))) + \rho(c, c(B(u^t, v^t))).
\end{equation*}
Using~\citet[Lemma~4]{Altschuler-2017-Near}, we derive Eq.~\eqref{claim-acceleration-iteration-main} as desired. 

By the definition of $(u^t, v^{t})$, we have $\varphi(\check{u}^t, \check{v}^t) \geq \varphi(u^t, v^t)$. Plugging this inequality into Eq.~\eqref{claim-acceleration-iteration-main} together with the Cauchy-Schwarz inequality yields 
\begin{equation*}
\varphi(\check{u}^t, \check{v}^t) - \varphi(\check{u}^{t+1}, \check{v}^{t+1}) \geq \frac{1}{4} E_t^2.  
\end{equation*}
Therefore, we conclude that 
\begin{equation*}
\varphi(\check{u}^t, \check{v}^t) - \varphi(\check{u}^{t+1}, \check{v}^{t+1}) \geq \frac{1}{4}\left(\sum_{i=j}^t E_i^2\right) \textnormal{ for any } j \in \{1, 2, \ldots, t\}.   
\end{equation*}
Since $\varphi(\check{u}^{t+1}, \check{v}^{t+1}) \geq \varphi(u^\star, v^\star)$ for all $t \geq 1$, we have $\varphi(\check{u}^j, \check{v}^j) - \varphi(\check{u}^{t+1}, \check{v}^{t+1}) \leq \delta_j$. Then, it follows from Lemma~\ref{lemma:acceleration-bound} that 
\begin{equation*}
\sum_{i=j}^t E_i^2 \leq \frac{32nR^2}{(j+1)^2}.
\end{equation*}
Putting these pieces together with the fact that $E_t \geq \varepsilon'$ as soon as the stopping criterion is not fulfilled yields 
\begin{equation*}
\frac{32nR^2}{(j+1)^2(t-j+1)} \geq (\varepsilon')^2.  
\end{equation*}
Since this inequality holds true for all $j \in \{1, 2, \ldots, t\}$, we assume without loss of generality that $t$ is even and let $j = t/2$. Then, we obtain that 
\begin{equation*}
t \leq 1 + \left(\frac{16\sqrt{n}R}{\varepsilon'}\right)^{2/3}.
\end{equation*}
This completes the proof of the theorem. 
\end{proof}
We are ready to present the complexity bound of Algorithm~\ref{Algorithm:ApproxOT_Acceleration} for solving the OT problem in Eq.~\eqref{prob:OT}. Note that $\varepsilon'=\frac{\varepsilon}{8\|C\|_\infty}$ is defined using the desired accuracy $\varepsilon > 0$.
\begin{theorem}\label{Theorem:ApproxOT-Acceleration-Total-Complexity}
The accelerated Sinkhorn scheme for approximating optimal transport (Algorithm~\ref{Algorithm:ApproxOT_Acceleration}) returns an $\varepsilon$-approximate transportation plan (cf. Definition~\ref{def:eps-approximation}) in 
\begin{equation*}
\bigO\left(\frac{n^{7/3}\|C\|_\infty^{4/3}(\log(n))^{1/3}}{\varepsilon^{4/3}}\right)
\end{equation*}
arithmetic operations.
\end{theorem}
\begin{proof}
Applying the same argument which is used in Theorem~\ref{Theorem:ApproxOT-GK-Total-Complexity}, we obtain that $\langle C, \widehat{X}\rangle - \langle C, X^\star\rangle \leq \varepsilon$ where $\tilde{X} = \textsc{Accelerated Sinkhorn}(C, \eta, \tilde{r}, \tilde{c}, \frac{\varepsilon'}{2})$ in Step 2 of Algorithm~\ref{Algorithm:ApproxOT_Acceleration}. 

It remains to bound the number of iterations required by Algorithm~\ref{Algorithm:acceleration} to reach the stopping criterion $E_t \leq \frac{\varepsilon'}{2}$. Using Theorem~\ref{Theorem:acceleration-iteration}, we have
\begin{equation*}
t \leq 1 + \left(\frac{32\sqrt{n}R}{\varepsilon'}\right)^{2/3}.  
\end{equation*}
By the definition of $R$ (cf. Lemma~\ref{lemma:dual-bound-infinity}), $\eta = \frac{\varepsilon}{4\log(n)}$ and $\varepsilon'=\frac{\varepsilon}{8\|C\|_\infty}$, we have 
\begin{eqnarray*}
t & \leq & 1 + \left(\frac{32\sqrt{n}R}{\varepsilon'}\right)^{2/3} \\
& \leq & 1 + \left(\frac{256\sqrt{n}\|C\|_\infty}{\varepsilon}\left(\frac{\|C\|_\infty}{\eta} + \log(n) - \log\left(\min_{1 \leq i, j \leq n} \{r_i, c_j\}\right)\right)\right)^{2/3} \\
& \leq & 1 + \left(\frac{256\sqrt{n}\|C\|_\infty}{\varepsilon}\left(\frac{4\log(n)\|C\|_\infty}{\varepsilon} + \log(n) - \log\left(\frac{\varepsilon}{64n\|C\|_\infty}\right)\right)\right)^{2/3} \\
& = & \bigO\left(\frac{n^{1/3} \|C\|_\infty^{4/3}(\log(n))^{1/3}}{\varepsilon^{4/3}}\right). 
\end{eqnarray*}
Since each iteration of Algorithm~\ref{Algorithm:acceleration} requires $\bigO(n^2)$ arithmetic operations, the total number of arithmetic operations required by Step 2 of Algorithm~\ref{Algorithm:ApproxOT_Acceleration} is $\bigO(n^{7/3}\|C\|_\infty^{4/3}(\log(n))^{1/3}\varepsilon^{-4/3})$. Computing two vectors $\tilde{r}$ and $\tilde{c}$ in Step 1 of Algorithm~\ref{Algorithm:ApproxOT_Acceleration} requires $\bigO(n)$ arithmetic operations and~\citet[Algorithm~2]{Altschuler-2017-Near} requires $\bigO(n^2)$ arithmetic operations. Therefore, the complexity bound of Algorithm~\ref{Algorithm:ApproxOT_Acceleration} is $\bigO(n^{7/3}\|C\|_\infty^{4/3}(\log(n))^{1/3}\varepsilon^{-4/3})$. 
\end{proof}
\begin{remark}
Theorem~\ref{Theorem:ApproxOT-Acceleration-Total-Complexity} shows that the complexity bound of accelerated Sinkhorn is better than that of Sinkhorn and Greenkhorn in terms of $1/\varepsilon$ but appears not to be near-linear in $n^2$. Thus, our algorithm is recommended when $n \ll 1/\varepsilon$. This occurs if the desired solution accuracy is relatively small, saying $10^{-4}$, and the examples include the application problems from economics and operations research. In contrast, Sinkhorn and Greenkhorn are recommended when $n \gg 1/\varepsilon$. This occurs if the desired solution accuracy is relatively large, saying $10^{-2}$, and the examples include the application problems from image processing. 
\end{remark}

%% file: sec/experiment.tex
\section{Experiments}\label{sec:experiments}
In this section, we conduct the experiments to evaluate Greenkhorn, accelerated Sinkhorn and APDAMD on synthetic data and real images from the MNIST Digits dataset\footnote{http://yann.lecun.com/exdb/mnist/}. We use Sinkhorn~\citep{Cuturi-2013-Sinkhorn}, APDAGD~\citep{Dvurechensky-2018-Computational} and GCPB\footnote{GCPB is simply an application of stochastic averaged gradient~\citep{Schmidt-2017} for solving the dual entropic regularized OT problem.}~\citep{Genevay-2016-Stochastic} as the baseline approaches. Since the focus of this paper is the entropic regularized algorithms, we exclude the combinatorial algorithms from our experiment and refer to~\citet{Dong-2020-study} for an excellent comparative study. 

In the literature, Greenkhorn and APDAGD were shown to outperform the Sinkhorn algorithm in terms of row/column updates~\citep{Altschuler-2017-Near, Dvurechensky-2018-Computational} and we repeat the comparisons for the sake of completeness. For parameter tuning in the implementation of Greenkhorn, accelerated Sinkhorn and APDAMD, we follow most of the setups as shown in Algorithm~\ref{Algorithm:GK},~\ref{Algorithm:APDAMD} and~\ref{Algorithm:acceleration} but employ more aggressive choice of stepsize for the coordinate gradient updates in Algorithm~\ref{Algorithm:acceleration}. To obtain an optimal value of the OT problem, we employ the default LP solver in MATLAB.
\begin{figure*}[!t]
\begin{minipage}[b]{.32\textwidth}
\includegraphics[width=55mm, height=42mm]{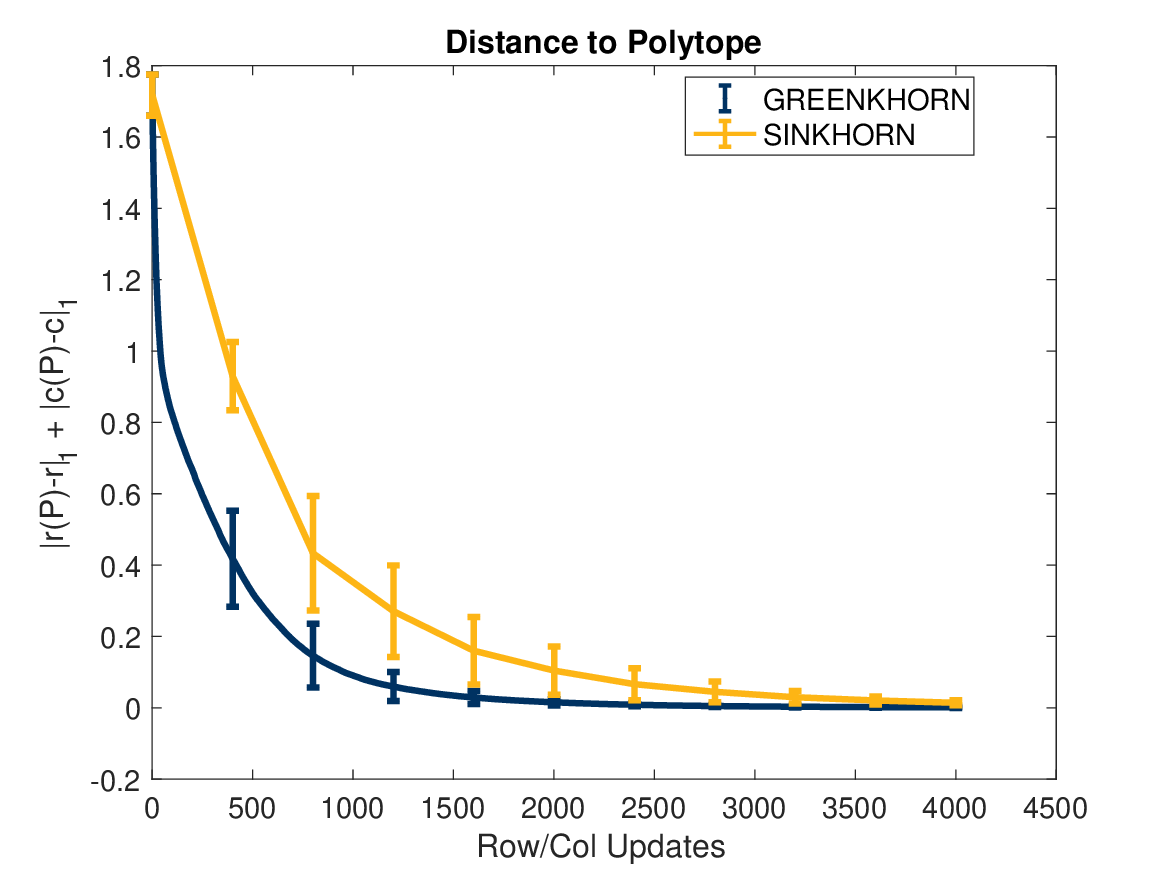}
\end{minipage}
\begin{minipage}[b]{.32\textwidth}
\includegraphics[width=55mm, height=42mm]{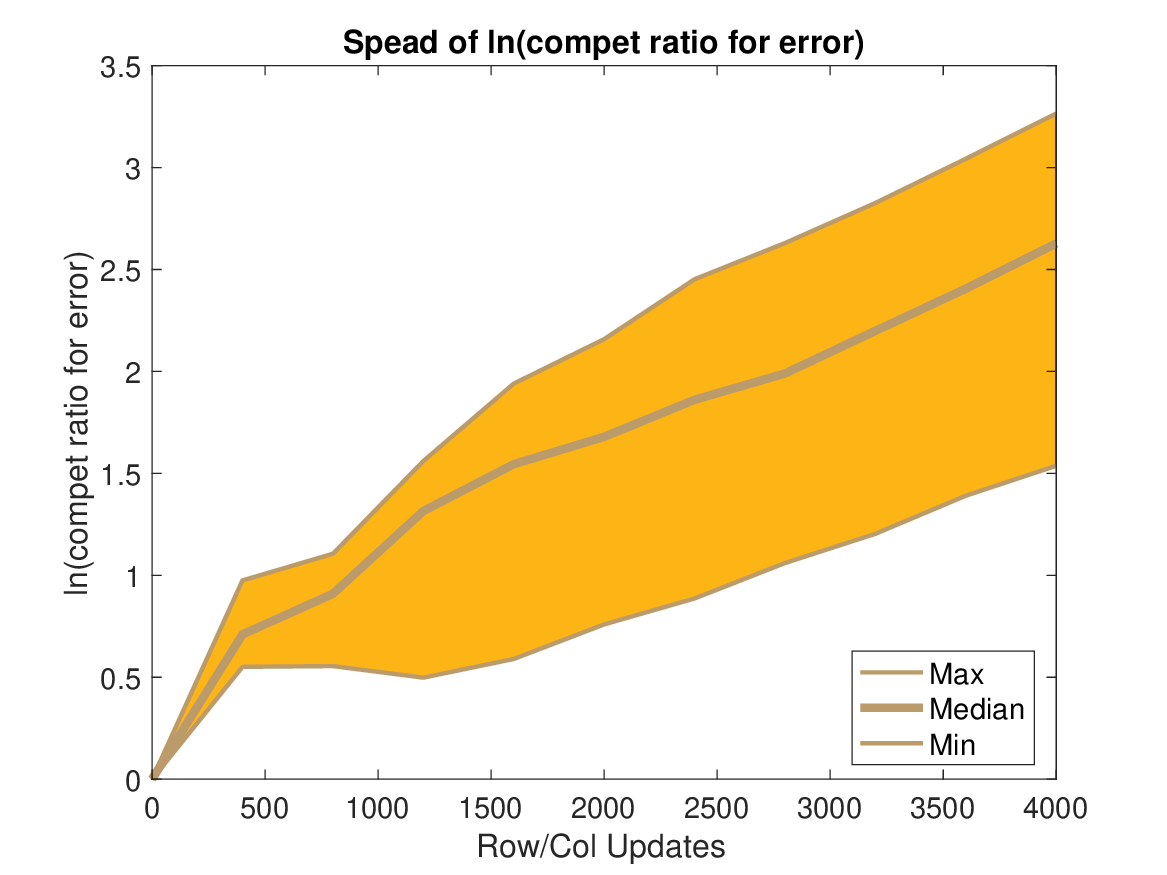}
\end{minipage}
\begin{minipage}[b]{.32\textwidth}
\includegraphics[width=55mm, height=42mm]{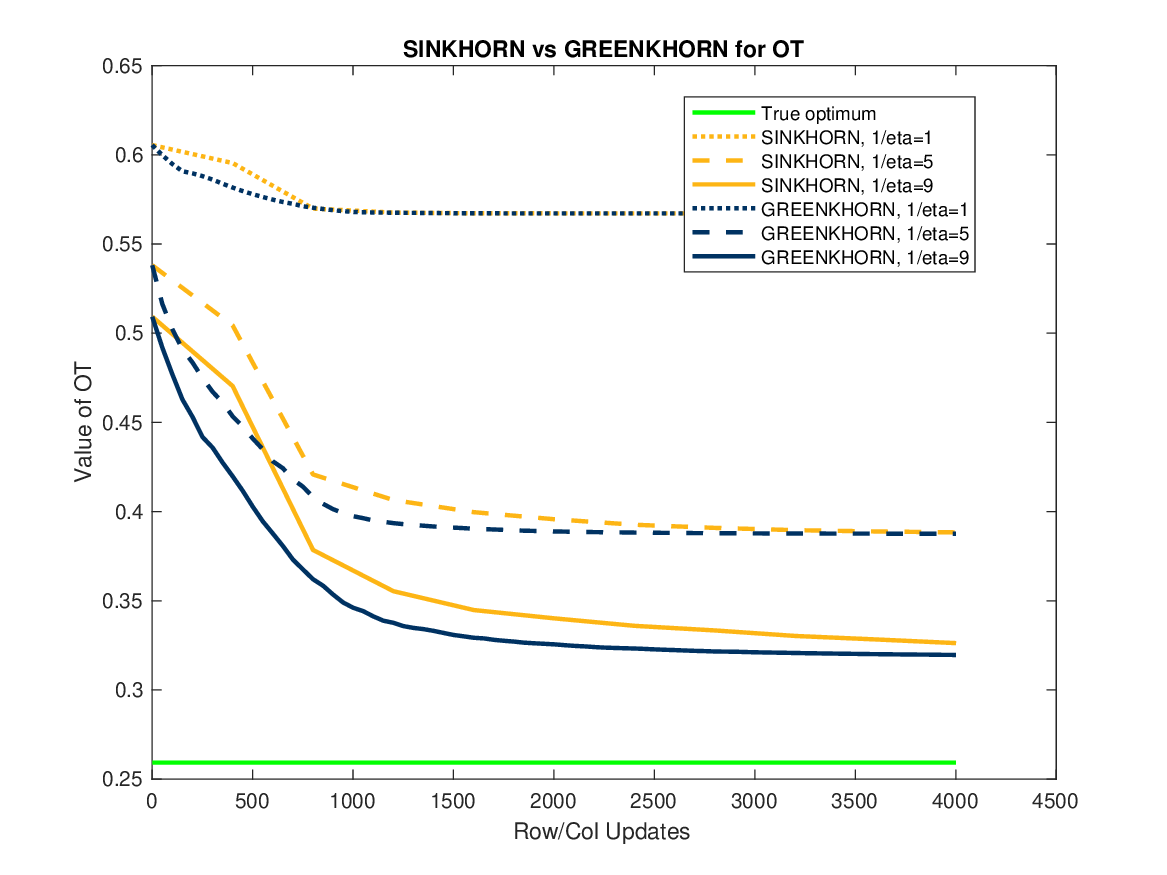}
\end{minipage} \\
\begin{minipage}[b]{.32\textwidth}
\includegraphics[width=55mm, height=42mm]{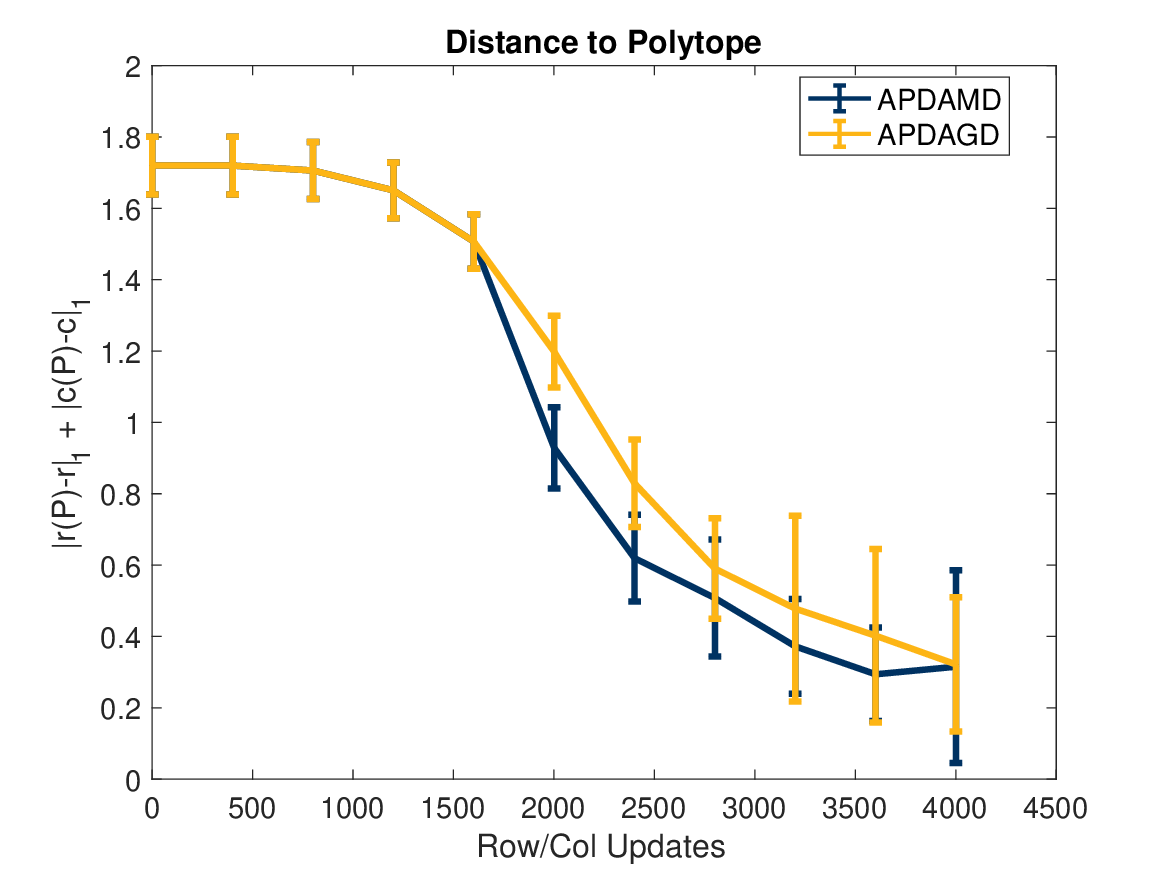}
\end{minipage}
\begin{minipage}[b]{.32\textwidth}
\includegraphics[width=55mm, height=42mm]{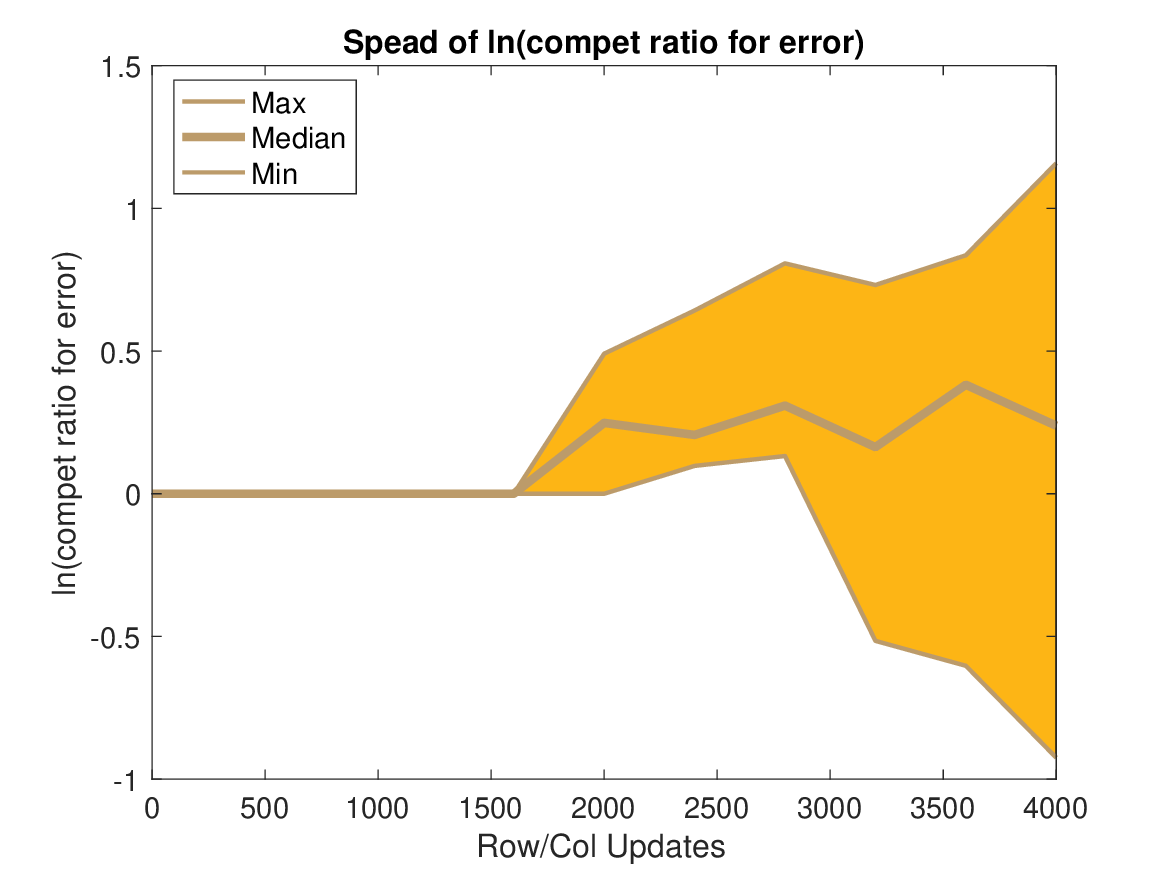}
\end{minipage}
\begin{minipage}[b]{.32\textwidth}
\includegraphics[width=55mm, height=42mm]{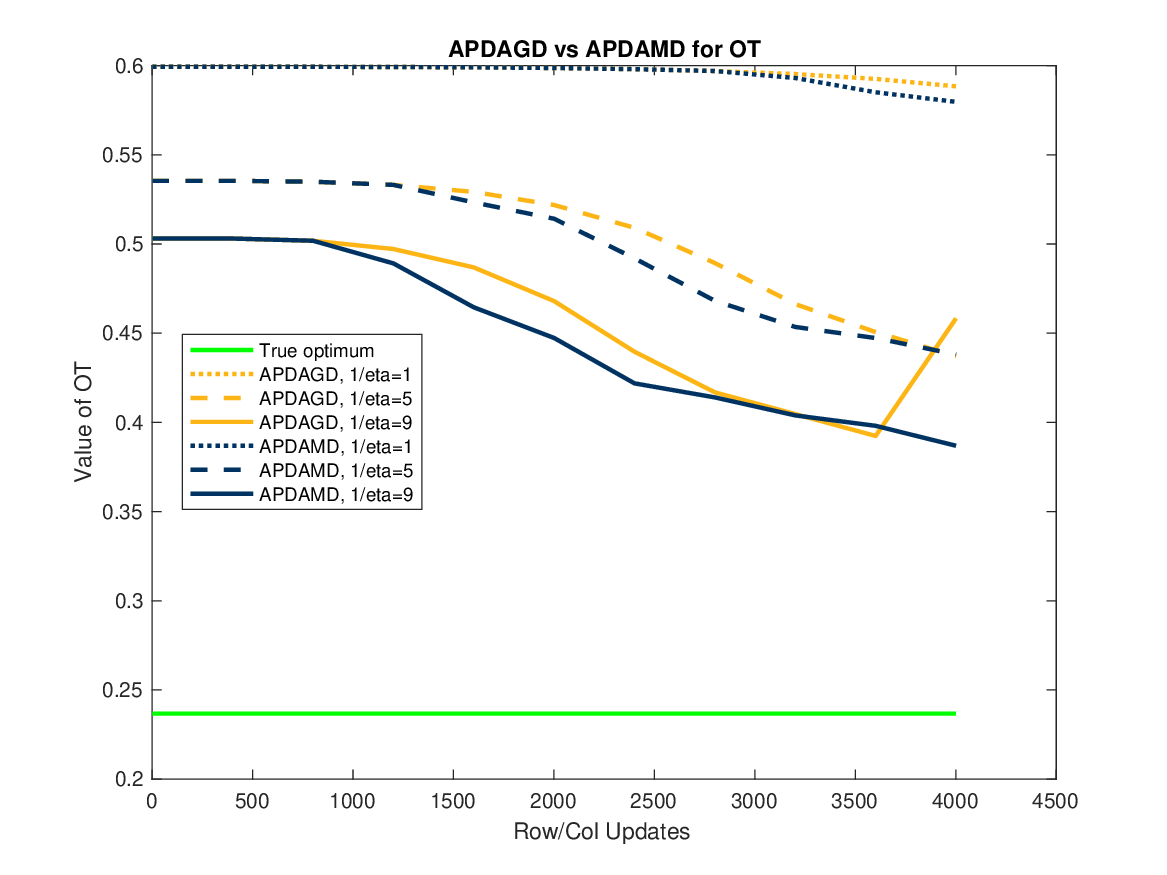}
\end{minipage} \\ 
\begin{minipage}[b]{.32\textwidth}
\includegraphics[width=55mm, height=42mm]{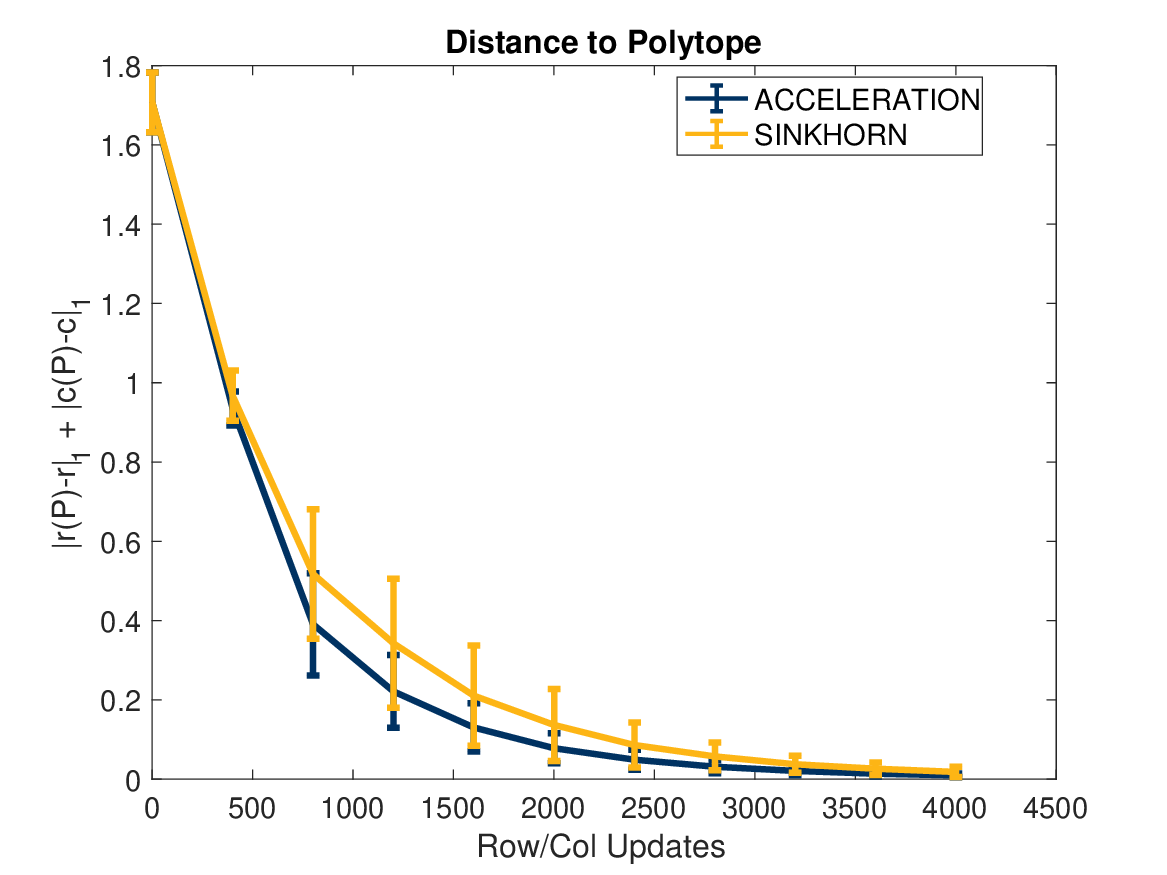}
\end{minipage}
\begin{minipage}[b]{.32\textwidth}
\includegraphics[width=55mm, height=42mm]{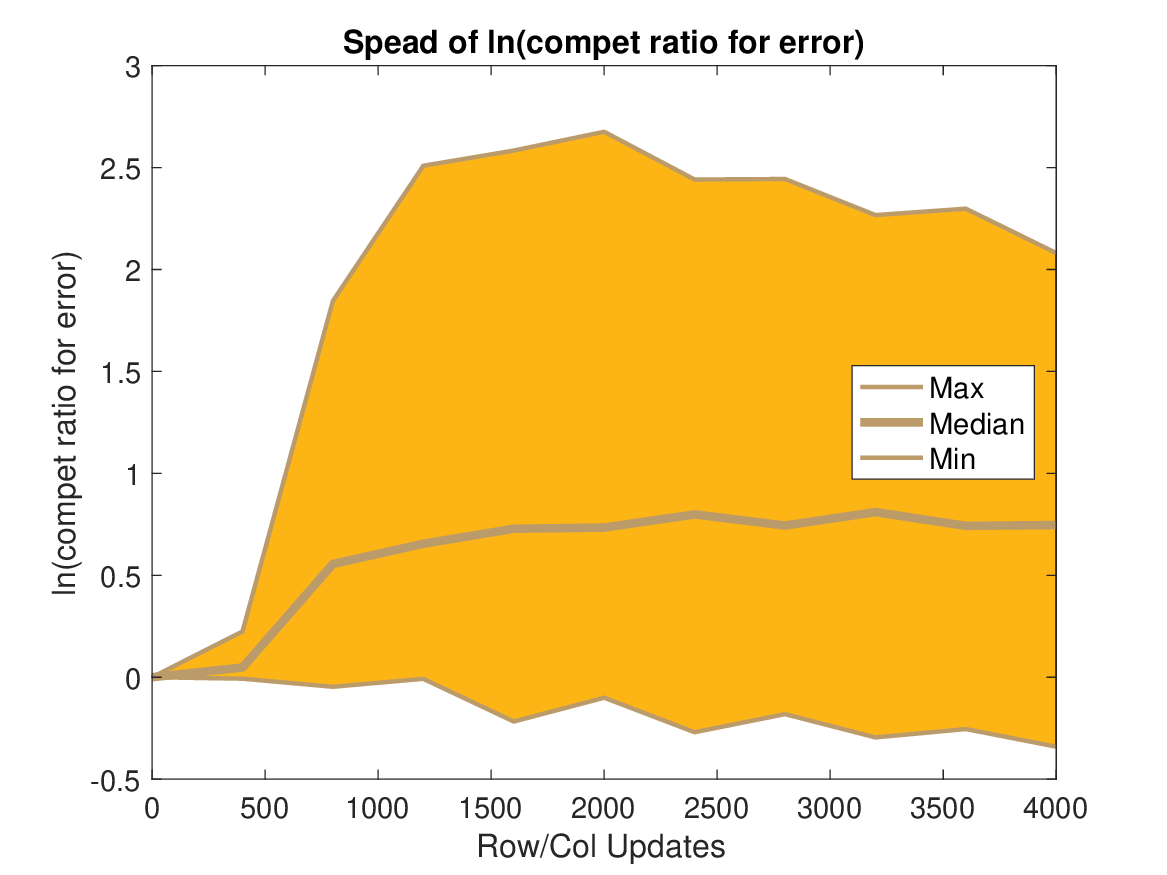}
\end{minipage}
\begin{minipage}[b]{.32\textwidth}
\includegraphics[width=55mm, height=42mm]{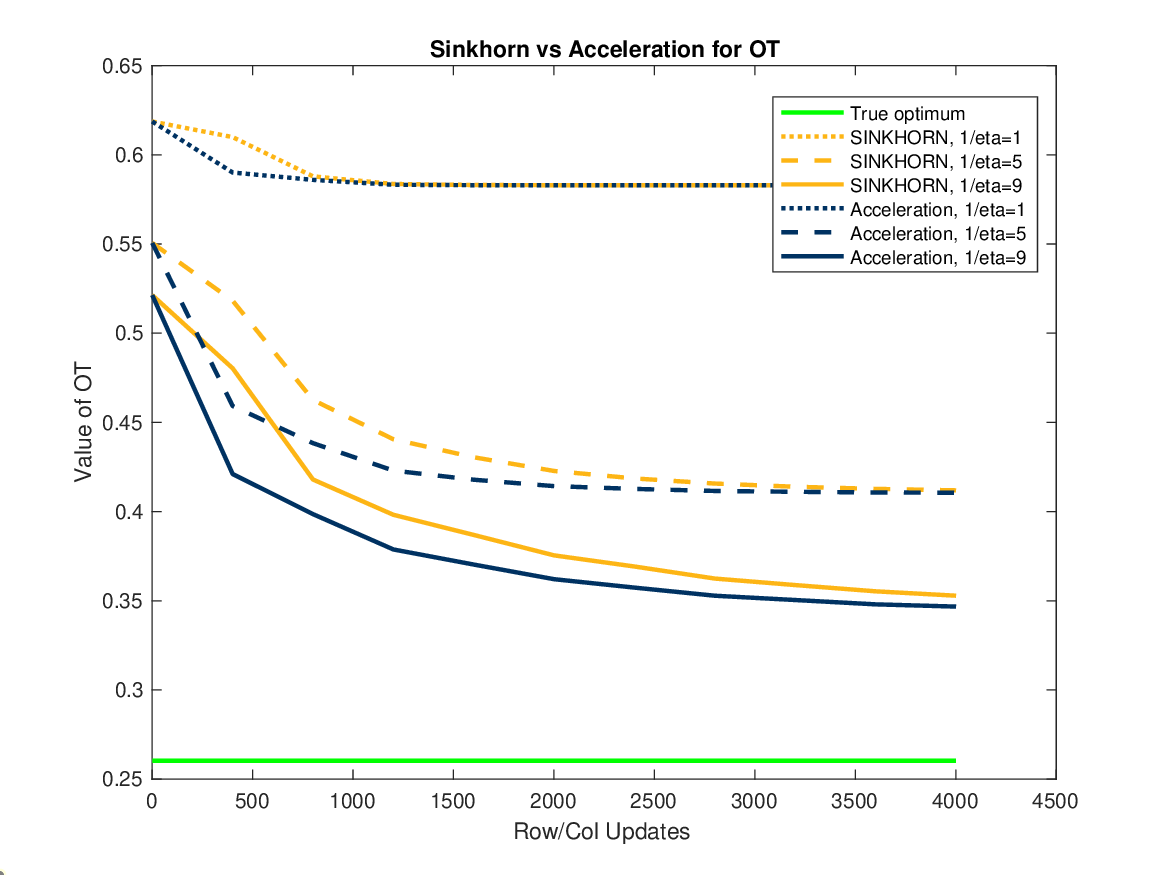}
\end{minipage}
\caption{Comparative performance of Sinkhorn v.s. Greenkhorn, APDAGD v.s. APDAMD and Sinkhorn v.s. accelerated Sinkhorn on synthetic images.}\label{fig:synthetic}\vspace*{-1em}
\end{figure*}
\subsection{Synthetic images}
To generate the synthetic images, we adopt the process from~\citet{Altschuler-2017-Near} and evaluate the performance of different algorithms on these synthetic images. The transportation distance is defined between two synthetic images while the cost matrix is defined based on the $\ell_1$ distances among locations of pixel in the images. Each image is of size 20 by 20 pixels and generated by means of randomly placing a foreground square in a black background. Furthermore, a uniform distribution on $[0, 1]$ is used for the intensities of the pixels in the background while a uniform distribution on $[0, 50]$ is employed for the pixels in the foreground. We fix the proportion of the size of the foreground square as $10\%$ of the whole images and implement all candidate algorithms. 

We use the standard metrics to assess the performance of all the candidate algorithms. The first metric $d(.)$ is an $\ell_1$ distance between the row, column outputs of some algorithm $\mathcal{A}$ and the corresponding transportation polytope of the probability measures, which is given by: $$d(\mathcal{A}) : = \|r(\mathcal{A}) - r\|_1 + \|c(\mathcal{A}) - c\|_1$$ where $r(\mathcal{A})$ and $c(\mathcal{A})$ are the row and column obtained from the output of the algorithm $\mathcal{A}$ and $r$ and $c$ are row and column vectors of the original probability measures. The second metric is defined as competitive ratio $\log(d(\mathcal{A}_{1})/d(\mathcal{A}_{2}))$ where $d(\mathcal{A}_{1})$ and $d(\mathcal{A}_{2})$ are the distances between the row, column outputs of algorithms $\mathcal{A}_1$ and $\mathcal{A}_2$ and the transportation polytope. We perform three pairwise comparative experiments on 10 randomly generated data: Sinkhorn v.s. Greekhorn, APDAGD v.s. APDAMD and Sinkhorn v.s. accelerated Sinkhorn. To further evaluate these algorithms, we compare their performance with respect to different choices of regularization parameter $\eta \in \{1, \frac{1}{5}, \frac{1}{9}\}$ while using the value of the OT problem as the baseline approach. The maximum number of iterations is $T=5$.
\begin{figure*}[!t]
\begin{minipage}[b]{.32\textwidth}
\includegraphics[width=55mm,height=45mm]{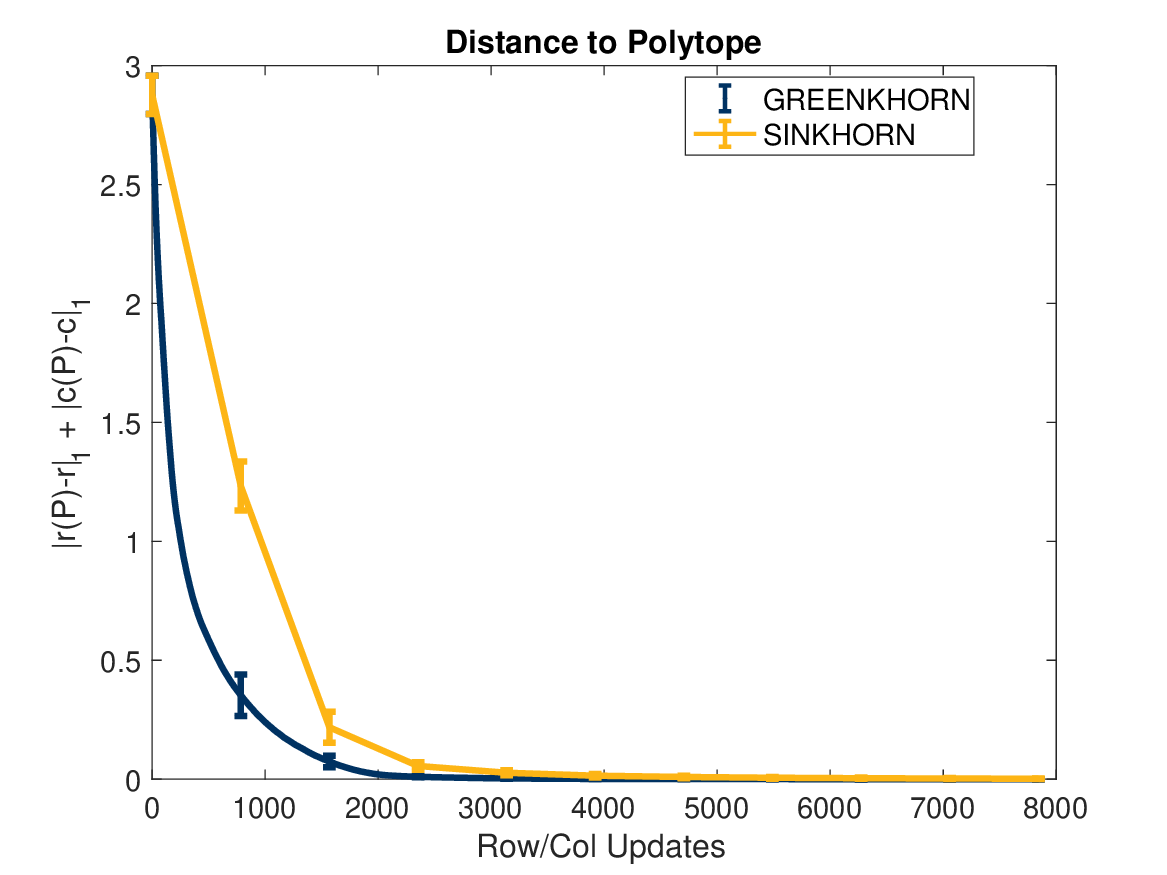}
\end{minipage}
\begin{minipage}[b]{.32\textwidth}
\includegraphics[width=55mm,height=45mm]{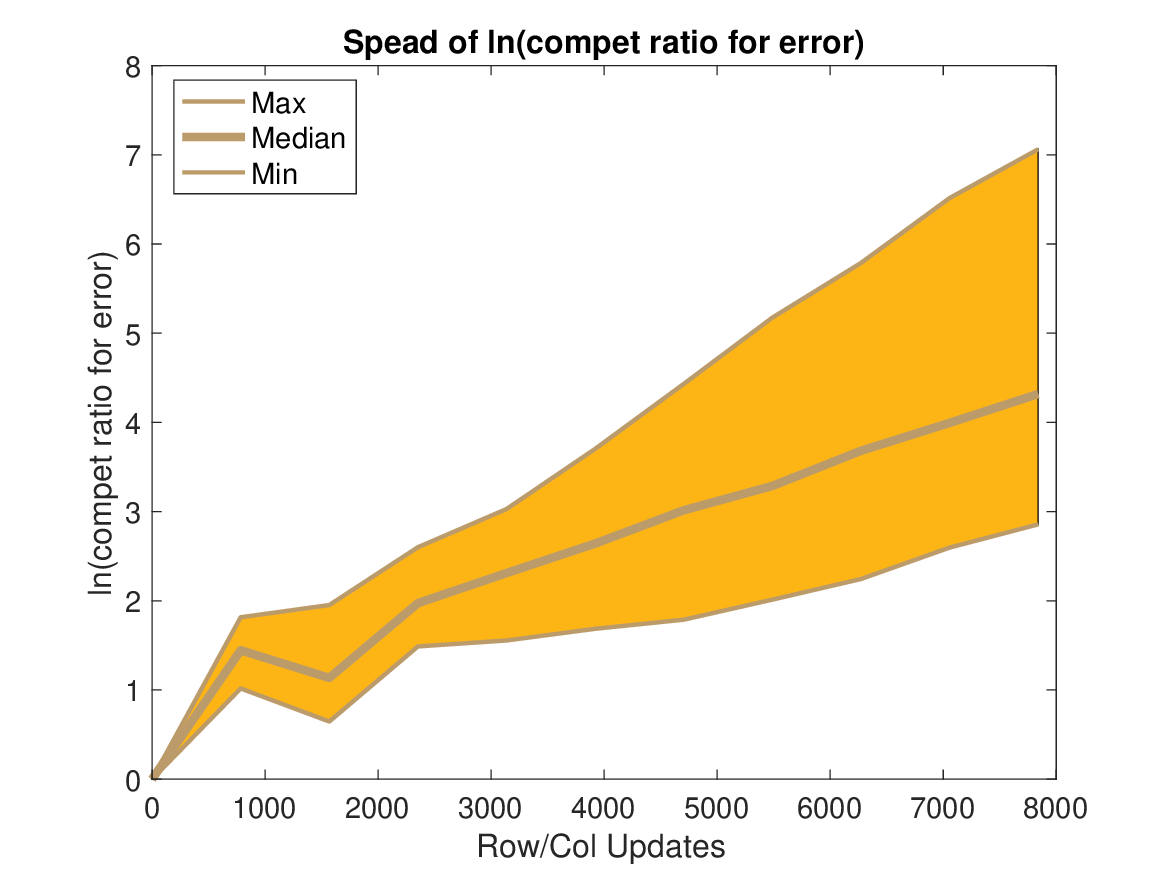}
\end{minipage}
\begin{minipage}[b]{.32\textwidth}
\includegraphics[width=55mm,height=45mm]{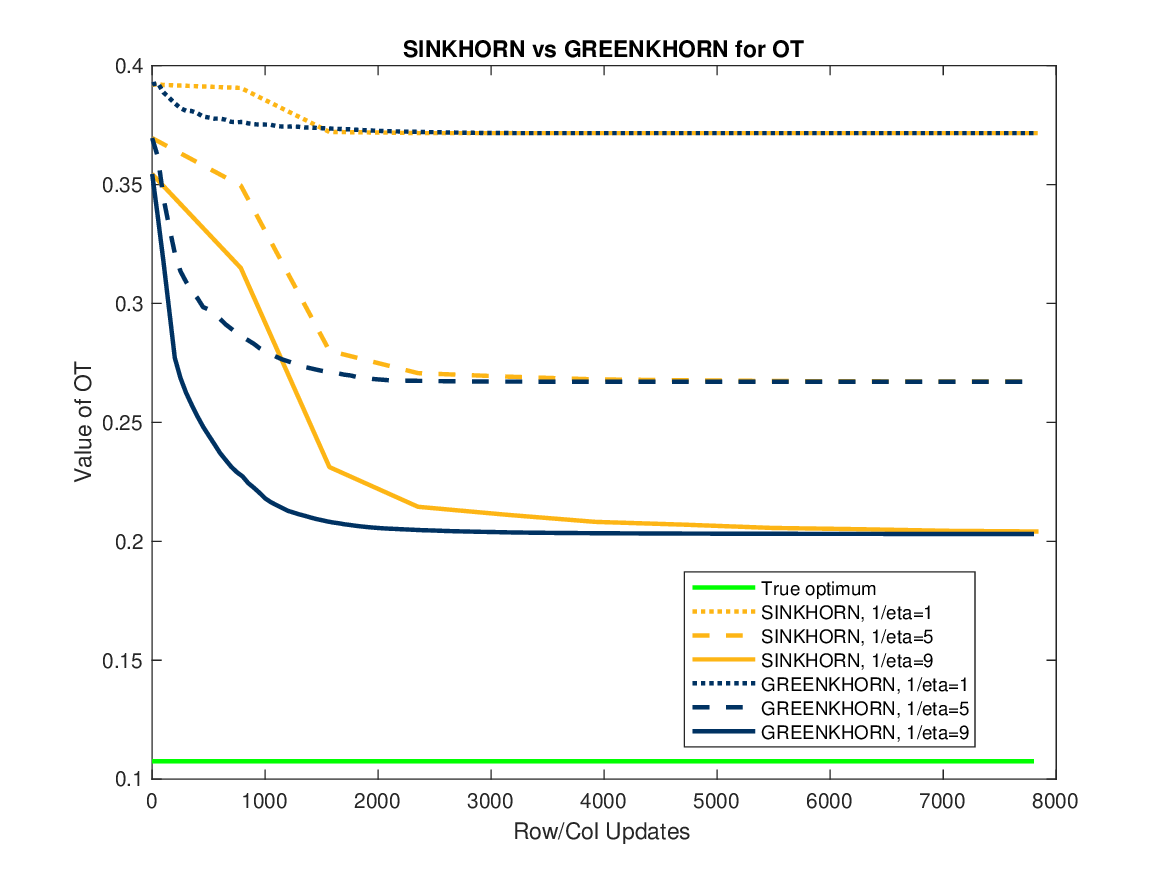}
\end{minipage} \\
\begin{minipage}[b]{.32\textwidth}
\includegraphics[width=55mm,height=42mm]{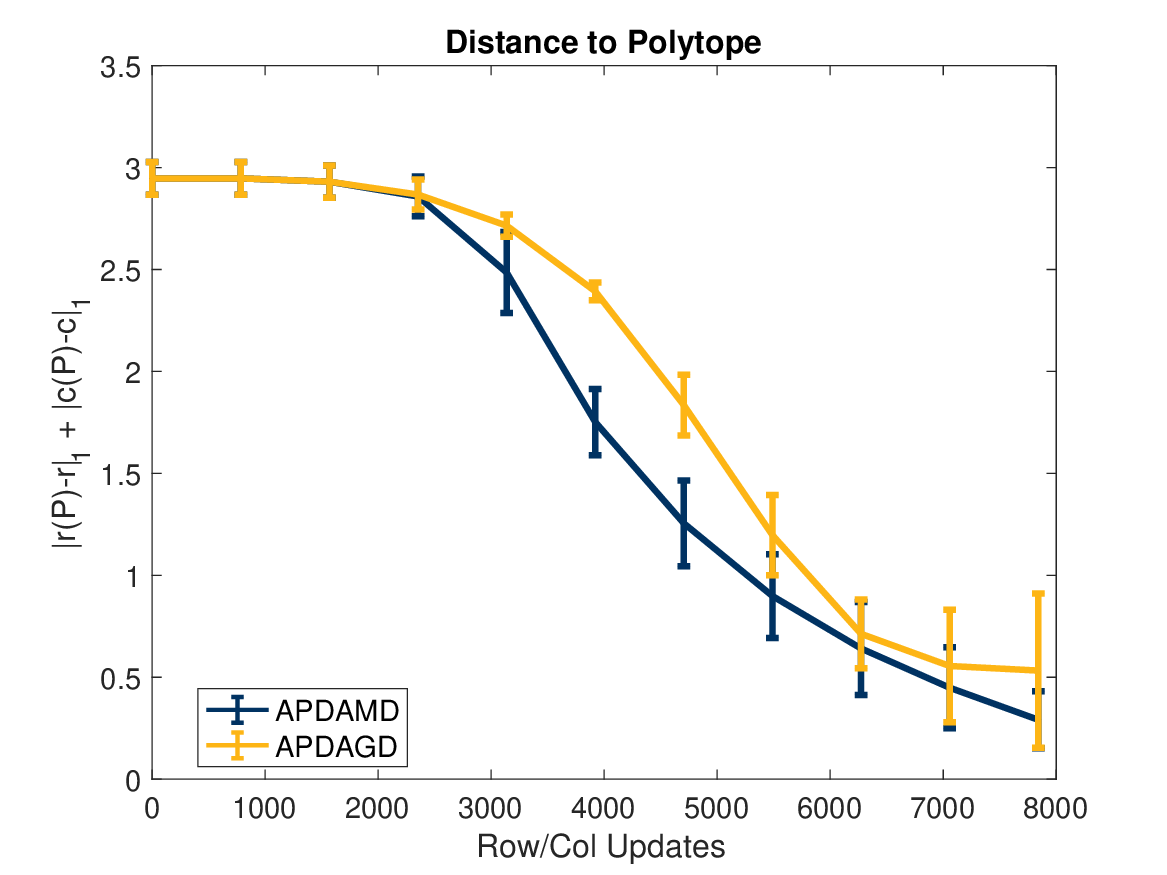}
\end{minipage}
\begin{minipage}[b]{.32\textwidth}
\includegraphics[width=55mm,height=42mm]{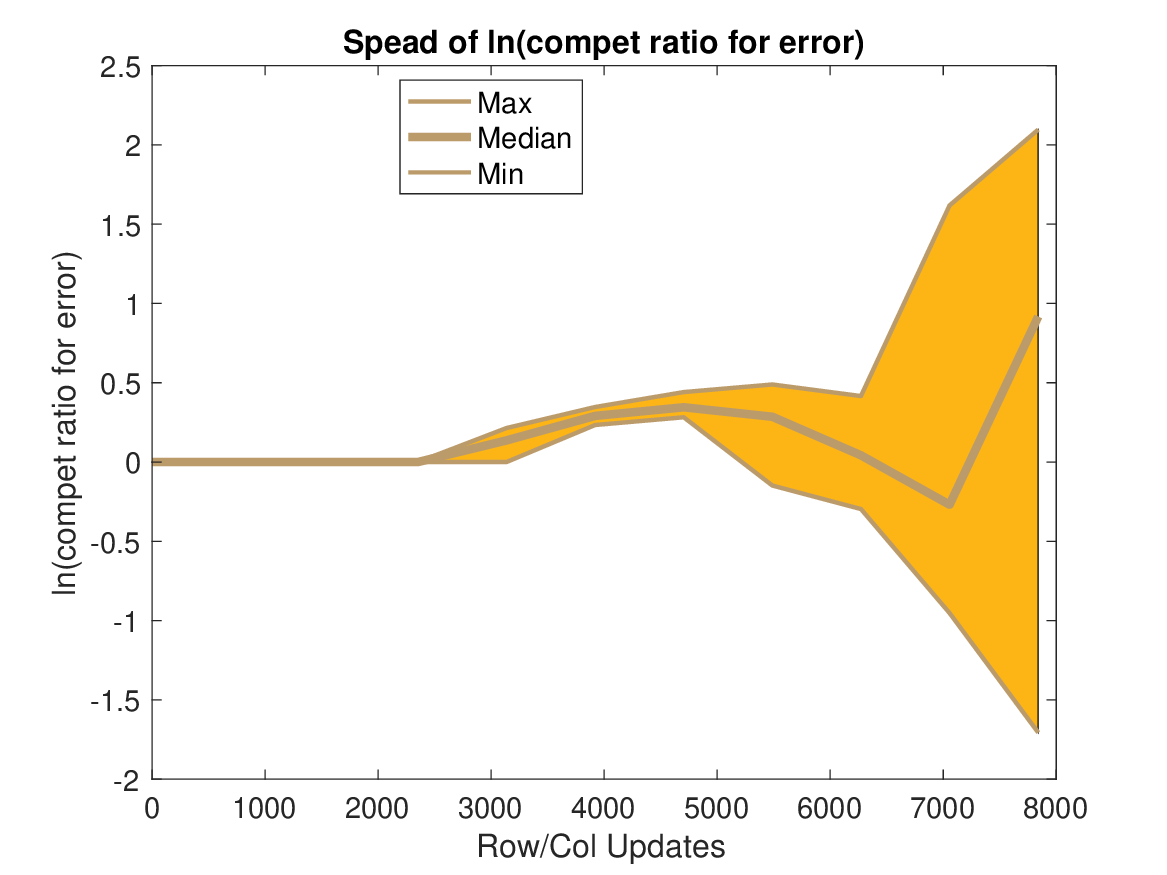}
\end{minipage}
\begin{minipage}[b]{.32\textwidth}
\includegraphics[width=55mm,height=42mm]{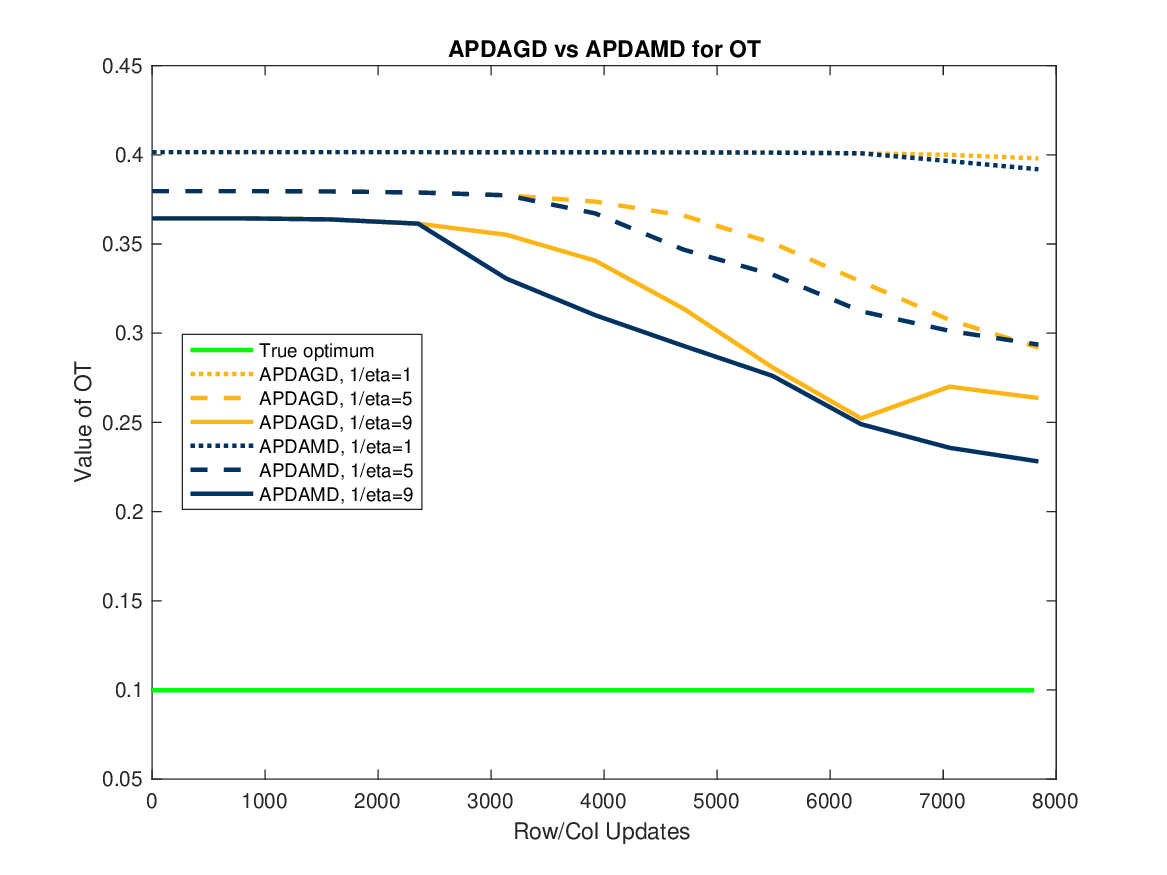}
\end{minipage} \\
\begin{minipage}[b]{.32\textwidth}
\includegraphics[width=55mm,height=42mm]{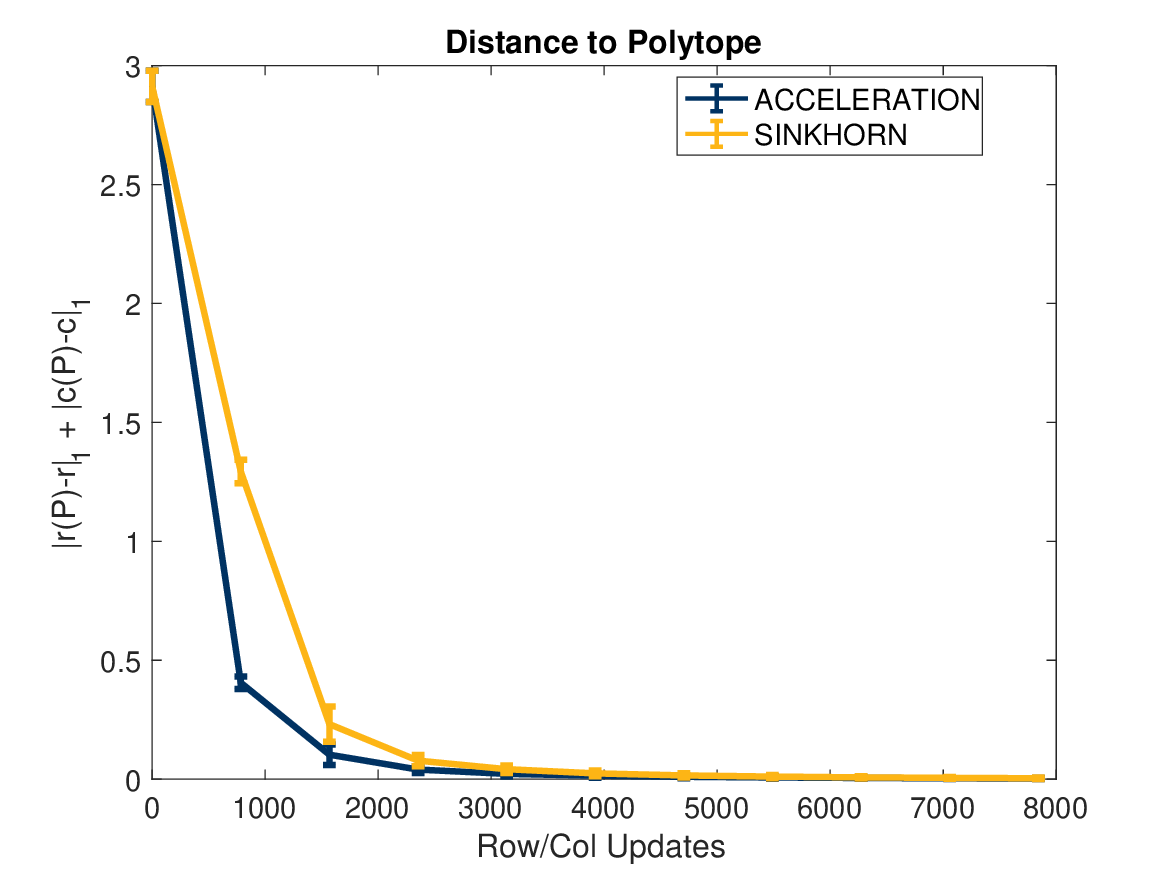}
\end{minipage}
\begin{minipage}[b]{.32\textwidth}
\includegraphics[width=55mm,height=42mm]{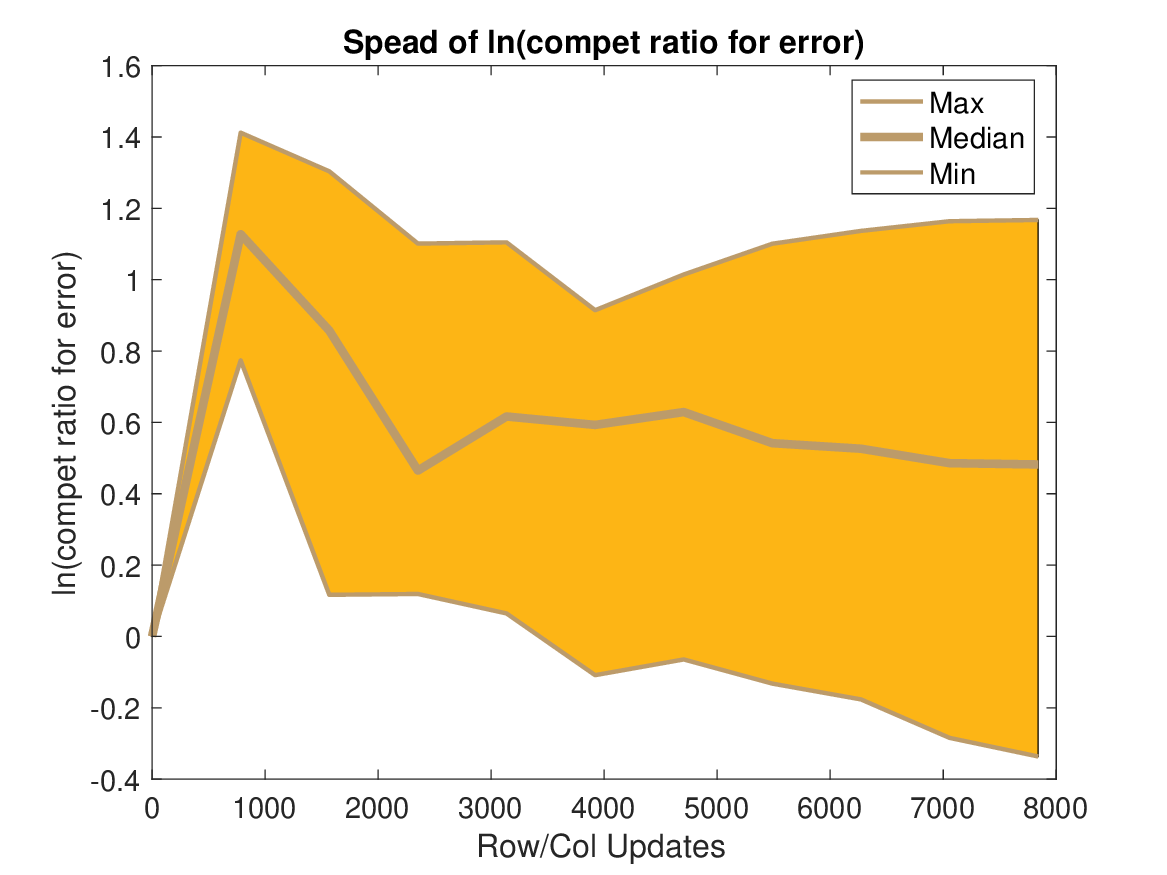}
\end{minipage}
\begin{minipage}[b]{.32\textwidth}
\includegraphics[width=55mm,height=42mm]{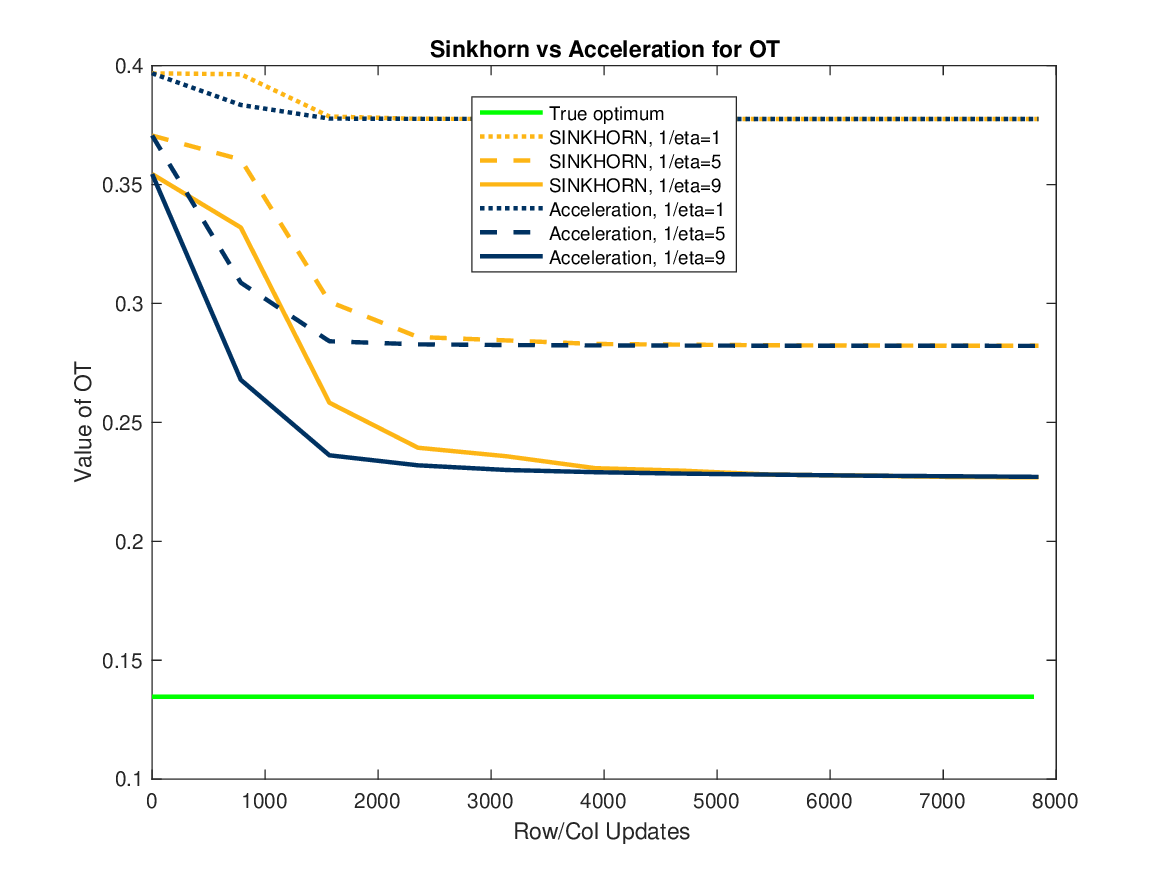}
\end{minipage}
\caption{Comparative performance of Sinkhorn v.s. Greenkhorn, APDAGD v.s. APDAMD and Sinkhorn v.s. accelerated Sinkhorn on the MNIST real images.}\vspace*{-1em}
\label{fig:MNIST}
\end{figure*}
Figure~\ref{fig:synthetic} summarizes the experimental results. The images in the first row show the comparative performance of Sinkhorn and Greenkhorn in terms of the row/column updates. In the leftmost image, the comparison uses distance to transportation polytope $d(\mathcal{A})$ where $\mathcal{A}$ is either Sinkhorn or Greenkhorn. In the middle image, the maximum, median and minimum values of the competitive ratios $\log(d(\mathcal{A}_1)/d(\mathcal{A}_2))$ on 10 images are utilized for the comparison where $\mathcal{A}_{1}$ is Sinkhorn and $\mathcal{A}_{2}$ is Greenkhorn. In the rightmost image, we vary the regularization parameter $\eta \in \{1, \frac{1}{5}, \frac{1}{9}\}$ with these algorithms and using the value of the unregularized OT problem as the baseline. The other rows of images present comparative results for APDAGD v.s. APDAMD and Sinkhorn v.s. accelerated Sinkhorn. We find that (i) Greenkhorn outperforms Sinkhorn in terms of row/column updates, illustrating the improvement from \textit{greedy coordinate descent}; (ii) APDAMD with $\delta = n$ and $\phi=(1/2n)\|\cdot\|^2$ is more robust than APDAGD, illustrating the advantage of using \textit{mirror descent} and line search with $\|\cdot\|_\infty$; (iii) accelerated Sinkhorn outperforms Sinkhorn in terms of row/column updates, illustrating the improvement from \textit{estimated sequence} and \textit{monotone search}. 

\subsection{MNIST images}
We proceed to the comparison between different algorithms on real images, using essentially the same evaluation metrics as in the synthetic images. The MNIST dataset consists of 60,000 images of handwritten digits of size 28 by 28 pixels. To ensure that the masses of probability measures are dense, which leads to a tight dependence on $n$ for our algorithms, we add a very small noise term ($10^{-6}$) to all zero elements in the measures and then normalize them so that their sum is 1. The maximum number of iterations is $T=5$.

Figures~\ref{fig:MNIST} and~\ref{fig:ot-MNIST} summarize the experimental results on MNIST. In the first row of Figure~\ref{fig:MNIST}, we compare Sinkhorn and Greenkhorn in terms of row/column updates. The leftmost image specifies the distances $d(\mathcal{A})$ to the transportation polytope for the algorithm $\mathcal{A}$, which is either Sinkhorn or Greenkhorn; the middle image specifies the maximum, median and minimum of competitive ratios $\log(d(\mathcal{A}_1)/d(\mathcal{A}_2))$ on ten random pairs of MNIST images, where $\mathcal{A}_1$ and $\mathcal{A}_2$ respectively correspond to Sinkhorn and Greenkhorn; the rightmost image specifies the values of the entropic regularized OT problem with varying regularization parameters $\eta \in \{1, \frac{1}{5}, \frac{1}{9}\}$. The remaining rows present comparative results for APDAGD v.s.APDAMD and Sinkhorn v.s.accelerated Sinkhorn. We observe that (i) the comparative performances of Sinkhorn v.s.Greenkhorn and APDAGD v.s.APDAMD are consistent with those on synthetic images; (ii) accelerated Sinkhorn deteriorates but remains better than Sinkhorn; (iii) APDAMD is more robust than APDAGD and GCPB. 
\begin{figure*}[!t]
\begin{minipage}[b]{.32\textwidth}
\includegraphics[width=55mm,height=45mm]{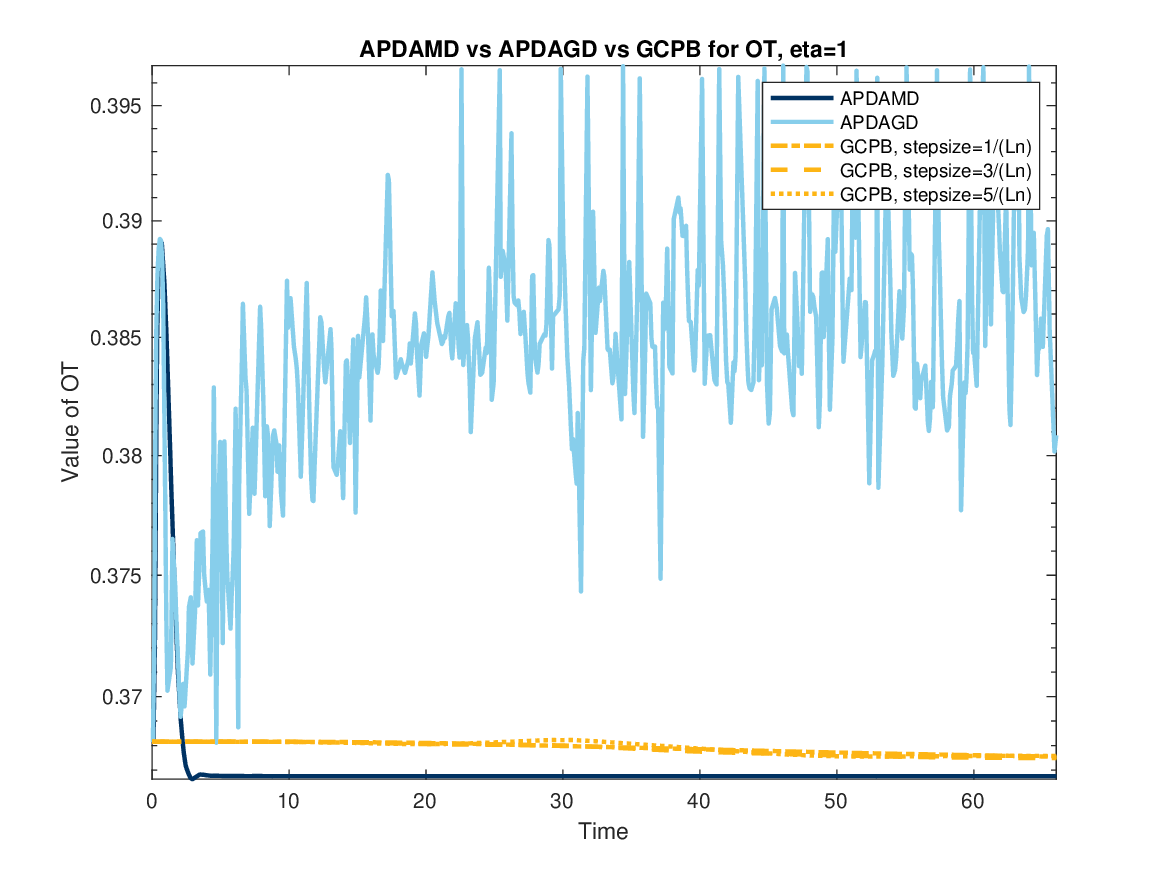}
\end{minipage}
\begin{minipage}[b]{.32\textwidth}
\includegraphics[width=55mm,height=45mm]{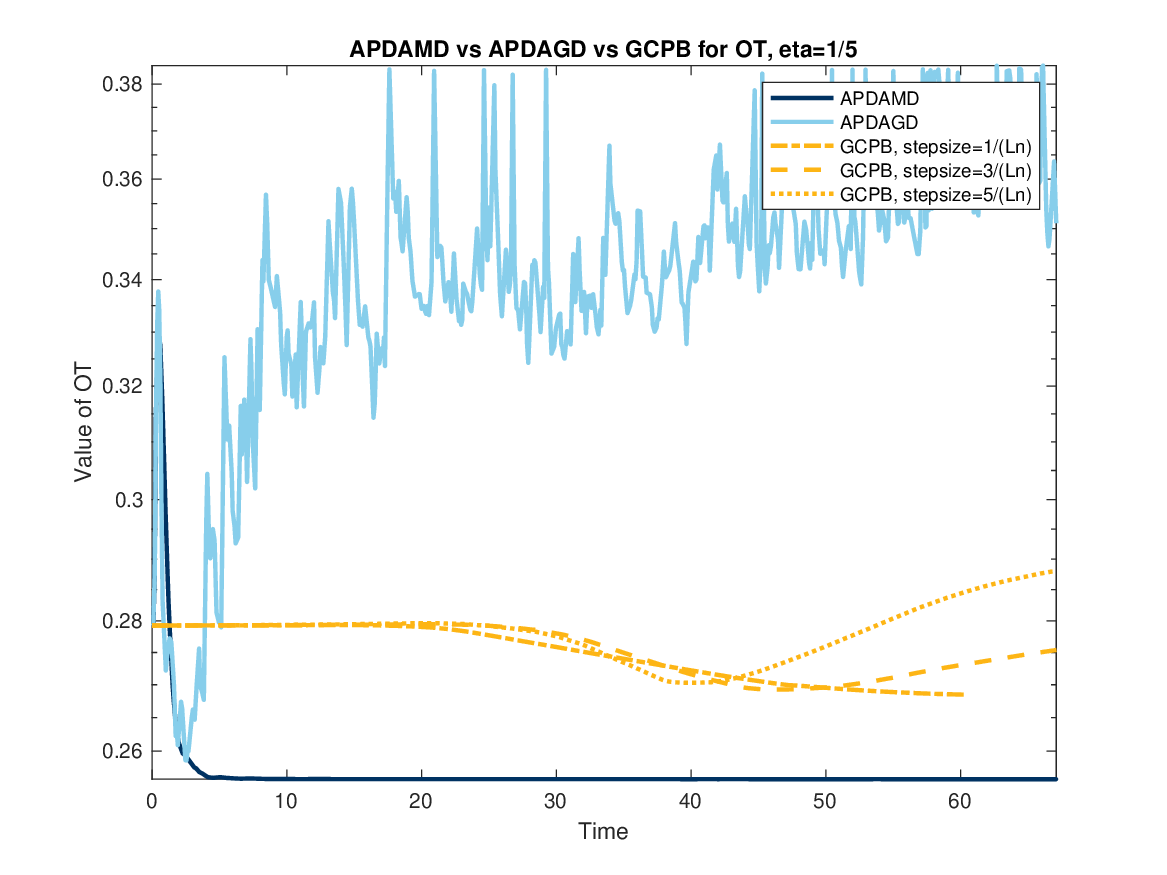}
\end{minipage}
\begin{minipage}[b]{.32\textwidth}
\includegraphics[width=55mm,height=45mm]{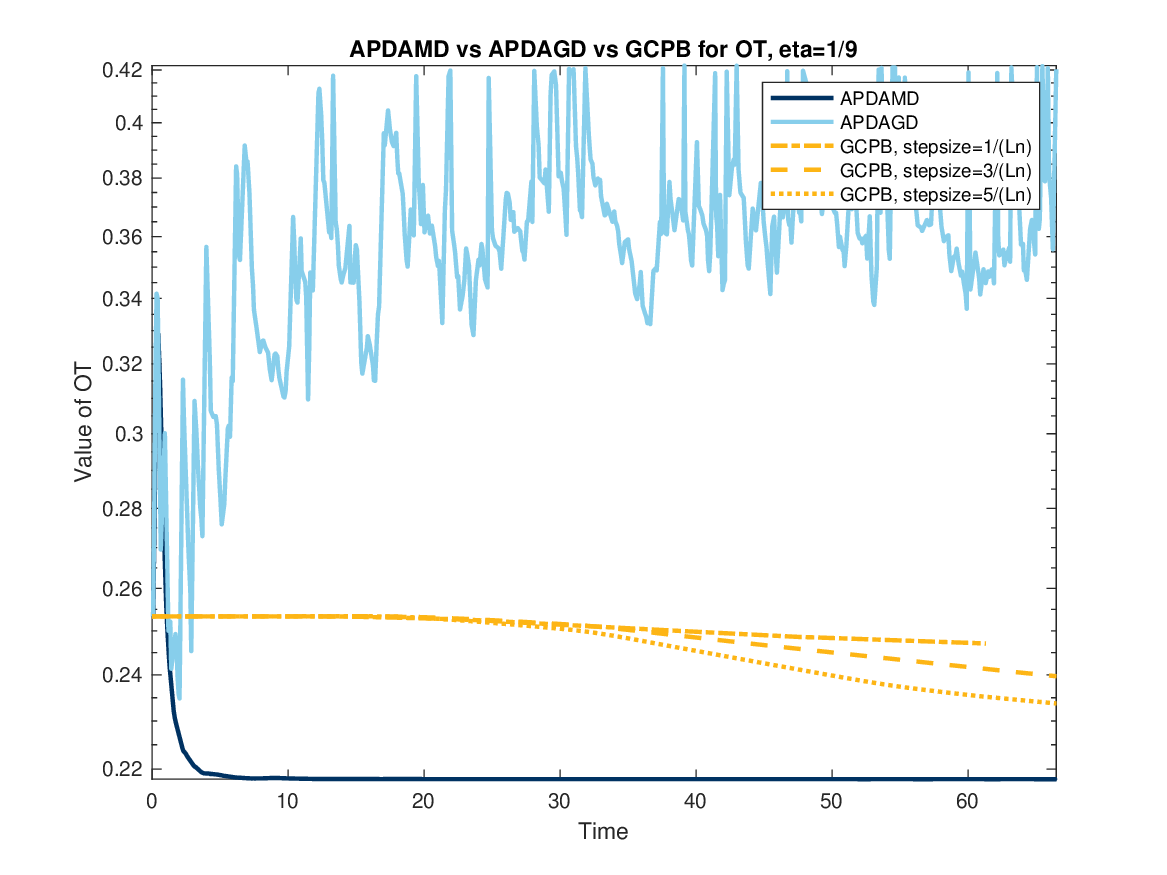}
\end{minipage}
\caption{Performance of GCPB, APDAGD and APDAMD in term of time on the MNIST real images. These images specify the values of entropic regularized OT with varying regularization parameter $\eta \in \{1, \frac{1}{5}, \frac{1}{9}\}$, demonstrating the robustness of APDAMD. }\vspace*{-1em}
\label{fig:ot-MNIST}
\end{figure*}

%% file: sec/conclusion.tex
\section{Conclusion}\label{sec:conclusion}
We first show that the complexity bound of Greenkhorn can be improved to $\bigOtil(n^2\varepsilon^{-2})$, which matches the best known bound of Sinkhorn. Then, we propose APDAMD by generalizing APDAGD with a prespecified mirror mapping $\phi$ and show that it achieves the complexity bound of $\bigOtil(n^2\sqrt{\delta}\varepsilon^{-1})$ where $\delta>0$ refers to the regularity of $\phi$. We prove that the complexity bound of $\bigOtil(\min\{n^{9/4}\varepsilon^{-1}, n^2\varepsilon^{-2}\})$ proved for APDAGD is invalid and prove a refined complexity bound of $\bigOtil(n^{5/2}\varepsilon^{-1})$. Moreover, we propose a \textit{deterministic} accelerated variant of Sinkhorn via appeal to estimate sequence techniques and prove the complexity bound of $\bigOtil(n^{7/3}\varepsilon^{-4/3})$. As such, we see that accelerated Sinkhorn outperforms Sinkhorn and Greenkhorn in terms of $1/\varepsilon$ and APDAGD and AAM in terms of $n$. Experiments on synthetic data and real images demonstrate the efficiency of our algorithms. 

There are a few promising future directions arising from this work. First, it is important to develop fast algorithms to compute dimension-reduced versions of OT. Indeed, the OT suffers from the curse of dimensionality~\citep{Dudley-1969-Speed, Fournier-2015-Rate}, which means that a large amount of samples from two continuous measures is necessary to approximate the true OT between them. This can be mitigated when data lie on low-dimensional manifolds~\citep{Weed-2019-Sharp, Paty-2019-Subspace} but the sample complexity still remain pessimistic even in that case. This motivates recent works on efficient dimension-reduced OT, e.g., the sliced OT~\citep{Bonneel-2015-Sliced}, generalized sliced OT~\citep{Kolouri-2019-Generalized}, distributional sliced OT~\citep{Nguyen-2021-Distributional}, further inspiring us to explore the application of our algorithms to these settings and eventually automatic differentiation schemes. Second, there have been several application problems arising from the interplay between OT and adversarial ML; see~\citet{Bhagoji-2019-Lower} and~\citet{Pydi-2020-Adversarial} for example. However, it is known that OT has robustness issues when there are outliers in the supports of probability measures. Robust OT had been introduced to deal with these robustness issues~\citep{Balaji-2020-Robust} where the idea is to relax the marginal constraints via certain probability divergences, such as KL divergence. It is to limit the amount of masses that the transportation plan will assign for the outliers in the supports of measures. Similar to OT, a key practical question with robust OT is computational. As such, we manage to develop efficient algorithms for the robust OT problem in the future work.